\documentclass[acmsmall,screen]{acmart}
\setcopyright{rightsretained}

\usepackage[T1]{fontenc}
\usepackage[utf8]{inputenc}
\usepackage{fontawesome5}

\usepackage{mathpartir}
\usepackage{graphicx}
\usepackage{hyperref}
\usepackage{multicol}
\usepackage{iris}
\usepackage{xspace}
\usepackage{xcolor}
\usepackage{listings}
\usepackage{amsfonts}
\usepackage{adjustbox}
\usepackage{relsize}
\usepackage{stmaryrd}
\usepackage{wasysym}
\usepackage{minted}

\usepackage{caption}
\usepackage{subcaption}

\usepackage[nameinlink,noabbrev]{cleveref}
\crefname{section}{\S\!}{\S}
\Crefname{section}{Section}{Sections}
\crefname{figure}{Figure}{Figures}
\crefname{theorem}{Theorem}{Theorems}
\newtheorem{theorem}{Theorem}[section]

\newcommand{\Citeappendix}[1]{The Appendix~\citep{extendedversion}}
\newcommand{\citeappendix}[1]{the Appendix~\citep{extendedversion}}

\newif\ifappendix
\appendixtrue 
\ifappendix
\renewcommand{\Citeappendix}[1]{\Cref{#1}}
\renewcommand{\citeappendix}[1]{\Cref{#1}}
\fi

\newcommand{\RULE}[1]{\hyperlink{#1}{\textsc{#1}}\xspace}

\newenvironment{DIFnomarkup}{}{}

\newcommand{\captionlabel}[2]{%
  \vspace{-0.7em}%
  \caption{#1}%
  \Description{#1}%
  \label{#2}
  \vspace{-0.7em}%
}

\newcommand{\oldlang}{DisLang\xspace}
\newcommand{\lang}{DisLang$_2$\xspace}
\newcommand{\dislog}{DisLog$_2$\xspace}
\newcommand{\olddislog}{DisLog\xspace}
\newcommand{\typedis}{TypeDis\xspace}

\newcommand{\val}{v}
\newcommand{\vals}{\vec\val}
\newcommand{\Values}{\mathcal{V}}
\newcommand{\vunit}{()}
\newcommand{\wal}{w}
\newcommand{\wals}{\vec\wal}
\renewcommand{\var}{x}
\newcommand{\vfold}[1]{\textsf{vfold}\,#1}
\newcommand{\loc}{\ell}

\newcommand{\Loc}{\mathcal{L}}
\newcommand{\Var}{\mathcal{V}}
\newcommand{\Vertices}{\mathcal{T}}
\newcommand{\vint}{i}
\newcommand{\vbool}{b}
\newcommand{\vtrue}{\kw{true}}
\newcommand{\vfalse}{\kw{false}}
\newcommand{\primitive}{\bowtie}

\newcommand{\block}{r}

\newcommand{\ectx}{K}
\newcommand{\efillctx}[2]{#1[#2]}
\newcommand{\khole}{\square}

\newcommand{\allocsize}{n}

\newcommand{\locs}[1]{roots(#1)}

\newcommand{\kw}[1]{\textsf{#1}}
\newcommand{\expr}{e}
\newcommand{\ealloc}[2]{\kw{alloc}\;#1\;#2}
\newcommand{\eload}[2]{#1.\kern-0.1em[#2]}
\newcommand{\estore}[3]{#1.\kern-0.1em[#2]\!\leftarrow\! #3 }
\newcommand{\elength}[1]{\kw{length}\,#1}
\newcommand{\eif}[3]{\kw{if}\,#1\,\kw{then}\,#2\,\kw{else}\,#3}
\newcommand{\elet}[3]{\kw{let}\,#1=#2\,\kw{in}\,#3}
\newcommand{\ecas}[4]{\kw{CAS}\,#1\,#2\,#3\,#4}
\newcommand{\eeq}[2]{#1\,\symbeq\,#2}
\newcommand{\blockrepeat}[2]{#2^{#1}}

\newcommand{\symbeq}{\mathord{==}}
\newcommand{\symbor}{\lor}
\newcommand{\symband}{\land}
\newcommand{\sequence}{\mathord{;}}

\newcommand{\ecall}[2]{#1\,#2}

\newcommand{\bclo}[3]{\hat\mu #1.\,\lambda#2.\,#3}
\newcommand{\eclo}[3]{\mu #1.\,\lambda#2.\,#3}

\newcommand{\erunapp}[2]{#1\,#2}

\newcommand{\ecallprim}[2]{#1\primitive#2}
\newcommand{\funcname}{f}
\newcommand{\argsname}{\vec{\var}}
\newcommand{\body}{\expr}
\newcommand{\efold}[1]{\textsf{fold}\,#1}
\newcommand{\eunfold}[1]{\textsf{unfold}\,#1}
\newcommand{\eprod}[2]{(#1,#2)}
\newcommand{\svec}[1]{\smash{\vec{#1}}}

\renewcommand{\subst}[3]{[#2/#1]#3}

\newcommand{\eparname}{\textsf{par}\xspace}
\newcommand{\epar}[2]{\textsf{par}({#1},{#2})}
\newcommand{\erunparsymb}{\|\kern-0.8pt|}
\newcommand{\erunpar}[2]{#1\,\|\, #2}

\newcommand{\varoff}{\iota}
\newcommand{\offleft}{1}
\newcommand{\offright}{2}

\newcommand{\ofs}{\vint}
\newcommand{\singletonmap}[2]{[#1:=#2]}
\newcommand{\blockupd}[3]{[#2:=#3]#1}
\newcommand{\length}[1]{|#1|}
\newcommand{\eproj}[2]{\textsf{proj}_{#1}\,#2}

\newcommand{\einj}[2]{\textsf{inj}_{#1}\,#2}
\newcommand{\einl}[1]{\einj{\offleft}{#1}}
\newcommand{\einr}[1]{\einj{\offright}{#1}}
\newcommand{\ecase}[5]{\textsf{match}\,#1\,\textsf{with}\,\einl{#2}\,\Rightarrow\,#3\mid\einr{#4}\,\Rightarrow\,#5\,\textsf{end}}

\newcommand{\smallpreceq}{\mathrel{\preccurlyeq}}
\newcommand{\prece}[2]{#1 \smallpreceq #2}

\makeatletter
\newcommand{\hourglass}{}
\DeclareRobustCommand{\hourglass}{\mathrel{\mathpalette\hour@glass\relax}}

\newcommand\hour@glass[2]{%
  \vcenter{\hbox{\resizebox{0.65em}{0.5em}{%
    \rotatebox[origin=c]{90}{$\m@th#1\bowtie$}%
  }}}\kern3pt%
}
\makeatother

\newcommand{\tarray}[1]{\textsf{array}(#1)}
\newcommand{\judg}[5]{#2\mid#3 \vdash #4 : #5 \,\triangleright\, #1}
\newcommand{\tforall}[3]{\forall #1::#2.\;#3}
\newcommand{\tunit}{\vunit}
\newcommand{\tproduct}[2]{(#1 \times #2)}
\newcommand{\tsum}[2]{(#1 + #2)}
\newcommand{\tfun}[5]{\forall #1\,#2.\;#3 \rightarrow^{#4} #5}
\newcommand{\trec}[3]{\mu #1.\,#2@#3}
\newcommand{\kind}{\kappa}
\newcommand{\kindsucc}[1]{\kern3pt\hourglass \Rightarrow #1}
\newcommand{\skindsucc}[1]{\hourglass \Rightarrow #1}
\newcommand{\tvar}{\alpha}
\newcommand{\tlambda}[2]{\lambda #1.\,#2}
\newcommand{\tapp}[2]{#1\,#2}
\newcommand{\tbool}{\textsf{bool}}
\newcommand{\tinteger}{\textsf{int}}

\newcommand{\tref}[1]{\textsf{ref}(#1)}
\newcommand{\refref}{\textsf{newref}\xspace}
\newcommand{\refget}{\textsf{get}\xspace}
\newcommand{\refset}{\textsf{set}\xspace}

\newcommand{\verypurename}{\textsf{veryPure}\xspace}
\newcommand{\verypure}[1]{\verypurename\,#1}

\newcommand{\wellkinded}[3]{#1 \vdash #2 :: #3 }
\newcommand{\subtime}[4]{#1 \vdash #3 \subseteq_{#2} #4}
\newcommand{\validat}[5]{#1 \mid #2 \mapsto #3\vdash_{#4} #5}
\newcommand{\reachable}[3]{#1 \vdash \prece{#2}{#3}}

\newcommand{\mapstopersist}{\mapsto_\boxempty}

\newcommand{\singleton}[1]{\{#1\}}
\newcommand{\fvtyp}[1]{\textsf{fv}(#1)}

\newcommand{\store}{\sigma}
\newcommand{\amap}{\alpha}
\newcommand{\graph}{G}
\newcommand{\vertex}{t}

\newcommand{\minispace}{\kern1pt}
\newcommand{\sepconfig}{\minispace{}/\minispace{}}
\newcommand{\antisepconfig}{\minispace{}\backslash\minispace{}}

\newcommand{\hconfig}[3]{#1\antisepconfig#2\antisepconfig#3}

\newcommand{\steph}[4]{#1\,,#2 \,\mathrel{\mathlarger{\vdash}}\, #3 \;\longrightarrow\;#4}

\newcommand{\headstep}[8]{\steph{#1}{#2}{\hconfig{#3}{#4}{#5}}{\hconfig{#6}{#7}{#8}}}
\newcommand{\structheadstep}[6]
{\headstep{#1}{#2}{#3}{#4}{#5}{#3}{#4}{#6}}

\newcommand{\steprightarrow}[2]
{\mathrel{\xrightarrow[#2]{\raisebox{-0.2ex}[0ex][0ex]{#1}}}}
\newcommand{\steprightarrowrule}[1]
{\mathrel{\xrightarrow{\adjustbox{margin=0ex 0ex 0ex
        1.2ex}{\raisebox{-0.2ex}[0ex][0ex]{#1}}}}}

\newcommand{\petit}[1]{\text{\scriptsize\rm\sf#1}}
\newcommand{\step}[2]{#1 \;\steprightarrowrule{\petit{step}}\; #2}
\newcommand{\stepinline}[2]{#1 \;\steprightarrow{\petit{step}}\; #2}
\newcommand{\rtcstep}[2]{#1 \;\steprightarrow{\petit{step}}\!\!^{\ast}\; #2}
\newcommand{\puresteparrow}{\steprightarrowrule{\petit{pure}}}
\newcommand{\purestep}[2]{#1 \;\puresteparrow\; #2}

\newcommand{\vertices}[1]{\textsf{vertices}(#1)}
\newcommand{\leaves}[1]{\textsf{leaves}(#1)}
\newcommand{\graphedge}[2]{(#1,#2)}

\newcommand{\parastep}[2]{#1 \;\steprightarrowrule{\petit{sched}}\; #2}
\newcommand{\parastepinline}[2]{#1 \;\steprightarrow{\petit{sched}}\; #2}

\newcommand{\oobname}{\textsf{OOB}\xspace}
\newcommand{\oob}[2]{\oobname\,#1\,#2}
\newcommand{\reduciblename}{\textsf{AllRedOrOOB}\xspace}
\newcommand{\reducible}[3]{\reduciblename\;\phruple{#1}{#2}{#3}}

\newcommand{\pre}{\Phi}
\newcommand{\post}{\Psi}
\newcommand{\pure}[1]{\ulcorner #1 \urcorner}
\newcommand{\itrue}{\pure{True}}

\newcommand{\qp}{p}

\let\oldstar\star
\let\basekind\oldstar
\renewcommand{\star}{\ast}
\newcommand{\morespacingaroundstar}{%
\let\oldstar\star
\renewcommand{\star}{\;\oldstar\;}%
}

\newcommand{\abefsymbol}{\text{\clock}}
\newcommand{\abef}[2]{#1 \,\abefsymbol\, #2}
\newcommand{\iProp}{\textdom{iProp}\xspace}

\newcommand{\bigast}[2]{\mathop{\textstyle\Sep}_{#1}\,#2}

\newcommand{\almostsafe}{always safe and disentangled\xspace}
\newcommand{\newapproach}{cyclic approach\xspace}

\newcommand{\proofconfig}[2]{\langle #1,\, #2 \rangle}
\newcommand{\wpname}{\textsf{wp}\xspace}
\newcommand{\basewp}[3]{\wpname\,\proofconfig{#1}{#2}\,\{#3\}}

\newcommand{\oldwp}[5]{\basewp{#1}{#2}{\lambda\,#3\,#4.\;#5}}
\renewcommand{\wp}[4]{\basewp{#1}{#2}{\lambda\,#3.\;#4}}

\newcommand{\unboxed}{\textsf{Unboxed}}
\newcommand{\nonsense}{\textsf{Nonsense}}
\newcommand{\correct}[1]{\textsf{Timestamp}\,#1}
\newcommand{\answer}{\textsf{answer}}
\newcommand{\trootname}{\textsf{rootf}\xspace}
\newcommand{\trootpre}[3]{\trootname\,#1\,#2\,#3}
\newcommand{\troot}[4]{\trootpre{#1}{#2}{#3}\,#4}

\newcommand{\envd}{h}
\newcommand{\envt}{m}
\newcommand{\project}[1]{#1} 
\newcommand{\abeftyp}[4]{\textsf{root}\,#3 \preccurlyeq^{#1}_{#2} #4}
\newcommand{\spaceeqdef}{\;\;\eqdef\;\;}
\newcommand{\sinterp}[3]{\llparenthesis #3 \rrparenthesis^{#1}_{#2}}

\newcommand{\edinterp}[2]{\llbracket#2\rrbracket^{#1}}
\newcommand{\tequiv}[2]{#1 \approx #2}

\newcommand{\eupd}[3]{[#1 := #2]#3}
\newcommand{\tupd}[5]{[#1 := (#2, (#3,#4))]#5}

\newcommand{\tbr}[1]{\llbracket\mkern-5mu\llbracket #1 \rrbracket\mkern-5mu\rrbracket}

\newcommand{\sigmainterp}[3]{\tbr{#3}^{#1}_{#2}}
\newcommand{\rfunc}{r}
\newcommand{\proper}[2]{\textsf{proper}_{#1}\,#2}
\newcommand{\regular}[2]{\textsf{regular}_{#1}\,#2}

\newcommand{\rhointerp}[4]{\llbracket #4 \rrbracket^{#1}_{#2}\,#3}

\newcommand{\dotabef}[2]{\abefsymbol_{#1}\,#2}

\newcommand{\ikindname}{\textsf{fkind}\xspace}
\newcommand{\ikind}[2]{\ikindname\,#1\,#2}
\newcommand{\tinterpzero}{\Values \rightarrow \iProp}
\newcommand{\tinterp}[1]{\ikind{#1}{(\tinterpzero)}}
\newcommand{\trfunc}[1]{\ikind{#1}{\Vertices}}

\newcommand{\dapp}{\mathop{+\!\!+}}

\newcommand{\dedupname}{\textsf{dedup}}

\newcommand{\haddname}{\textsf{add}\xspace}
\newcommand{\hadd}[4]{\ecall\haddname{[#1;#2;#3;#4]}}

\newcommand{\filtercompactname}{\textsf{filter\_compact}\xspace}
\newcommand{\filtercompact}[2]{\ecall\filtercompactname{[#1;#2]}}

\newcommand{\htsize}{C}
\newcommand{\htable}{a}
\newcommand{\helem}{x}

\newcommand{\letinline}[2]{\kw{let}\;#1\;=\;#2}

\newcommand{\eqdefspace}{\;\;\eqdef\;\;}

\newcommand{\emptyelem}{d}

\newcommand{\parforname}{\textsf{parfor}\xspace}
\newcommand{\parfor}[3]{\ecall\parforname{[#1; #2; #3]}}

\newcommand{\selfname}{f}
\newcommand{\lowbound}{a}
\newcommand{\highbound}{b}
\newcommand{\diffname}{(\highbound - \lowbound)}
\newcommand{\midname}{mid}

\newcommand{\hashname}{h}
\newcommand{\taskname}{k}
\newcommand{\zero}{0}
\newcommand{\one}{1}

\newcommand{\parforarg}{k}

\newcounter{remark}[section]

\usepackage[american]{babel}
\newcommand{\defn}[1]{\emph{\textbf{#1}}}

\newcommand{\codeinline}[1]{\mintinline{ocaml}{#1}}

\newmintinline[ocaml]{ocaml}{escapeinside=??}

\begin{document}

\title{TypeDis: A Type System for Disentanglement}
\ifappendix
\subtitle{Extended Version}
\fi

\author{Alexandre Moine}
\orcid{0000-0002-2169-1977}
\email{alexandre.moine@nyu.edu}
\affiliation{%
  \institution{New York University}
  \city{New York}
  \country{USA}
}

\author{Stephanie Balzer}
\orcid{0000-0002-8347-3529}
\email{balzers@cs.cmu.edu}
\affiliation{%
  \institution{Carnegie Mellon University}
  \city{Pittsburgh}
  \country{USA}
}

\author{Alex Xu}
\orcid{0009-0003-6455-9217}
\email{alexxu@andrew.cmu.edu}
\affiliation{%
  \institution{Carnegie Mellon University}
  \city{Pittsburgh}
  \country{USA}
}

\author{Sam Westrick}
\orcid{0000-0003-2848-9808}
\email{shw8119@nyu.edu}
\affiliation{%
  \institution{New York University}
  \city{New York}
  \country{USA}
}

\setcopyright{cc}
\setcctype{by}
\acmDOI{10.1145/3776655}
\acmYear{2026}
\acmJournal{PACMPL}
\acmVolume{10}
\acmNumber{POPL}
\acmArticle{13}
\acmMonth{1}
\received{2025-07-10}
\received[accepted]{2025-11-06}

\begin{abstract}
Disentanglement is a runtime property of parallel
programs guaranteeing that parallel
tasks remain oblivious to each other's allocations.
As demonstrated in the MaPLe compiler and run-time system, disentanglement
can be exploited for fast automatic memory management, especially
task-local garbage collection with no synchronization between
parallel tasks.
However, as a low-level property, disentanglement can be difficult to
reason about for programmers.
The only means of statically verifying disentanglement so far
has been DisLog,
an Iris-fueled variant of separation logic,
mechanized in the Rocq proof assistant.
DisLog is a fully-featured program logic, allowing for proof of functional
correctness as well as verification of disentanglement.
Yet its employment requires significant expertise and per-program proof effort.

This paper explores the route of automatic verification via a type system,
ensuring that any well-typed program is disentangled and
lifting the burden of carrying out manual proofs from the programmer.
It contributes TypeDis, a type system inspired by region types,
where each type is annotated with a timestamp,
identifying the task that allocated it.
TypeDis supports iso-recursive types as well as polymorphism over both
types and timestamps.
Crucially, timestamps are allowed to change
during type-checking, at join points as well as
via
a form
of subtyping, dubbed \emph{subtiming}.
The paper illustrates TypeDis and its features on a range of examples.
The soundness of TypeDis and the examples are mechanized
in the Rocq proof assistant, using an improved version of DisLog,
dubbed DisLog2.

\end{abstract}


\begin{DIFnomarkup}
\end{DIFnomarkup}


\maketitle

\section{Introduction}
\label{sec:intro}
A recent line of work has identified a key memory property of parallel
programs called \emph{disentanglement}~\cite{acar-et-al-15,raghunathan-et-al-16,guatto-et-al-18,westrick-et-al-20,arora-westrick-acar-21,westrick-arora-acar-22,arora-westrick-acar-23,moine-westrick-balzer-24,arora-muller-acar-24}.
Roughly speaking, disentanglement is the property that concurrent tasks remain
oblivious to each other's memory allocations.
As demonstrated by the MaPLe compiler~\cite{mpl-2020}, this property makes it
possible to perform task-local memory management (allocations and
garbage collection) independently, in parallel, without any synchronization
between concurrent tasks.
MaPLe in particular features a provably efficient memory management system
for a dialect of Parallel ML---a parallel functional programming
language---and offers competitive performance in practice relative to low-level
parallel code written in languages such as C/C++~\cite{arora-westrick-acar-23}.

This line of work aims to gain control over
the synchronization costs of parallel garbage collection by taking
advantage of structured fork-join parallelism.
The idea is to synchronize the garbage collector only at
application-level forks and joins, thereby making these synchronization costs
predictable at the source level, and
avoiding the need for any global synchronization of the garbage collector.
At each fork and join, the runtime system
performs $O(1)$ work to maintain a dynamic tree of heaps which mirrors
the parent/child relationships between tasks.
Each task thus has its own task-local heap, in which it allocates memory
objects and may perform garbage collection independently, in parallel.
The independence of these task-local garbage collections hinges upon
\emph{disentanglement}, which can be defined as a ``no cross-pointers''
invariant.
Specifically, disentanglement allows for \emph{up-pointers} from descendant
heaps to ancestors, as well as \emph{down-pointers} from ancestors to
descendants, but disallows \emph{cross-pointers} between concurrent
tasks (siblings, cousins, etc.).
The existence of cross-pointers is called \emph{entanglement}.
When two tasks become entangled with a cross-pointer, neither task can
perform garbage collection without synchronizing with the other.
These additional synchronizations lead to significant performance
degradations~\cite{arora-westrick-acar-23}, and in this sense,
entanglement is a performance hazard.

\begin{figure*}
\begin{minipage}{0.5\textwidth}
{\small
\begin{minted}[escapeinside=??]{ocaml}
let r = ref ""
let rec write_max x =
  let current = !r in
  if x <= current || compare_and_swap r current x
  then () else write_max x
let entangled =
  par(fun () -> write_max (Int.to_string 1234),
      fun () -> write_max (Int.to_string 5678))
\end{minted}
}
\end{minipage}
\hfill
\begin{minipage}{0.4\textwidth}
\includegraphics[width=\textwidth]{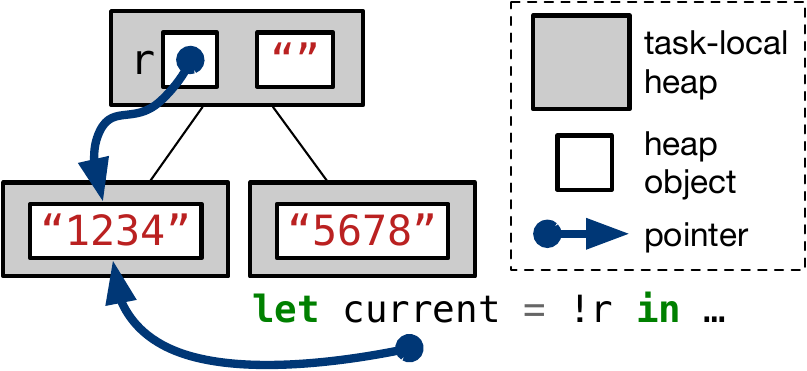}
\end{minipage}
\captionlabel{Entanglement example}{fig:entanglement}
\end{figure*}

Entanglement arises from a particular communication pattern,
where one task allocates a local heap object and then
another task (executing concurrently, relative to the first task)
acquires a pointer to the object.
An example is shown in \Cref{fig:entanglement}.
The example uses the fork-join primitive \codeinline{par(f1,f2)} to execute
two functions in parallel; these two function calls correspond to two child
tasks.
In the example, the two child tasks perform \codeinline{write_max} concurrently,
both attempting to update \codeinline{r} to point to a locally allocated
string.
After one task finishes, the other task reads the updated \codeinline{r} and
acquires a cross-pointer, which constitutes entanglement.

As shown below, this example could be rewritten to be disentangled by moving the allocations
of the strings ``up'' into the parent task (making all pointers involved up-pointers).
{\small
\begin{minted}[escapeinside=??]{ocaml}
    ... (* same definitions of r and write_max, as in ?\Cref{fig:entanglement}? *) ...
    let disentangled =
      let a = Int.to_string 1234 in let b = Int.to_string 5678 in
      par(fun () -> write_max a, fun () -> write_max b)
\end{minted}
}

\paragraph{Preventing entanglement}
One way to rule out entanglement is to disallow side effects entirely.
Indeed, the original study of disentanglement emerged out of an interest in
improving the performance of parallel functional programming techniques, which
naturally have a high rate of allocation and whose scalability and efficiency
is largely determined by the performance of automatic memory management.
In this setting, disentanglement is guaranteed by construction due to a lack
of side effects.
But the full power of disentanglement lies in its expressivity beyond
purely functional programming---in particular, disentanglement allows for
judicious utilization of side effects such as in-place updates
and irregular and/or data-dependent access patterns in shared memory.
These side effects are crucial for efficiency in state-of-the-art
implementations of parallel
algorithms, such as those in the PBBS Benchmark
Suite~\cite{DBLP:conf/spaa/0002PZWJ24,DBLP:conf/spaa/ShunBFGKST12,DBLP:conf/ppopp/AndersonBDD022},
which have been found to be
naturally disentangled~\cite{westrick-arora-acar-22}.

In these more general settings, where there are
numerous opportunities to efficiently utilize side effects, it
is easy for a programmer to accidentally entangle concurrent tasks.
Ideally, it would be evident at the source level where entanglement
may or may not occur.
However, in practice, this is not the case.
To reason about entanglement, the programmer effectively
has to know the memory allocation and access patterns of the entire program.
This makes it especially difficult to reason about higher-order functions,
because the memory effects of a function taken as argument are unknown.
Other high-level programming features also complicate the matter, such as
parametric polymorphism which allows for code to be specialized for both
``boxed'' (heap-allocated) and ``unboxed'' types, potentially resulting in
entanglement in one case but not the other.
These details can be formally considered using the
program logic DisLog~\cite{moine-westrick-balzer-24},
but verifying disentanglement
using DisLog requires
significant expertise and effort, even in small
examples.

An interesting question therefore is whether it is
possible to guarantee disentanglement statically through a type system.
This would have the advantage of being mostly automatic, requiring (ideally)
only a modest amount of type annotation.
Most importantly, a type system would raise the level of abstraction at which
the programmer can reason about disentanglement, clarifying how the property
interacts with high-level abstractions such as parametric
polymorphism, higher-order functions, algebraic datatypes, and other desirable
features.

\paragraph{A type system for disentanglement}
In this paper, we present \defn{\typedis}, the first static type system
for disentanglement.
We intend for \typedis to be the type system for a high-level ML-like language
with structured fork-join parallelism, in-place atomic operations on shared
memory, and disentangled parallel garbage collection.
The language features a single parallel construct, written $\epar{f_1}{f_2}$,
which calls $f_1()$ and $f_2()$ in parallel, waits for both to
complete, and returns their results as a pair.
Here, we think of the execution of the two function calls as two child
tasks, which themselves might execute \texttt{par}(...) recursively,
creating a dynamic tree (parent-child) relationship between tasks.

\typedis
identifies tasks with timestamp variables~$\delta$,
and annotates every value computed
during execution with the timestamp of the task that allocated that value.
This is tracked explicitly in the type of the value.
For example, $s: \textsf{string}@\delta$ indicates
that the value $s$ is a \textsf{string} that was allocated
by a task $\delta$.
The type system implicitly maintains a partial order over timestamps,
written $\prece{\delta'}{\delta}$,
intuitively
corresponding to the tree
relationship between tasks.
Crucially, \typedis guarantees an invariant that we call the
\textbf{up-pointer invariant}:
for every task running at timestamp~$\delta$,
every value accessed by this task must have a timestamp $\prece{\delta'}{\delta}$,
i.e., the value must have been allocated ``before'' the current timestamp.
%
%
In other words, the key insight in this paper is to \emph{restrict all memory
references to point backwards in time}, which is checked statically.
This restriction is a deep invariant over values: every data
structure will only contain values allocated at the same timestamp or a
preceding timestamp.
As a result, all loads in the language are guaranteed to be safe for
disentanglement.

The up-pointer invariant statically rules out one feature of disentanglement:
down-pointers.
This restriction is mild, however, because down-pointers are fairly rare.
Quantitatively, there has been at least one relevant study: in their work on
entanglement detection, \citet{westrick-arora-acar-22} observe in multiple
benchmarks that down-pointers do not arise at all, and more broadly they
measure that the number of objects containing down-pointers is small.
The creation of a down-pointer requires a combination of dynamic allocation and
pointer indirection, each of which is typically avoided in parallel
performance-sensitive code to reduce memory pressure and improve cache
efficiency.
In this paper, we have found the up-pointer invariant to be sufficiently
expressive to encode a number of interesting examples~(\cref{sec:evaluation}),
fully typed within TypeDis, and therefore guaranteed disentangled.
The up-pointer invariant is especially well-suited for immutable
data (which naturally adheres to the invariant), as well as parallel batch
processing of pre-allocated data.
The up-pointer invariant also allows for structure sharing, even in the
presence of mutable state.

\paragraph{Maintaining the up-pointer invariant}
To maintain the up-pointer invariant in the presence of mutable state,
\typedis places a restriction on writes (in-place updates), requiring that
the timestamp of the written value precedes the timestamp of
the reference pointing to it.
This restriction is implemented in the type system with a form of subtyping,
dubbed \textbf{subtiming}, which affects only the timestamps of values within
their types.
The idea is to allow for any value to be (conservatively) restamped with a
\emph{newer} timestamp.
Subtiming makes it possible to express the restriction on writes as a simple
unification over the type of the contents of a mutable reference or
array.

Restamping with an \emph{older} timestamp would be unsound in \typedis, as it
would allow for a child's (heap-allocated) value to be written into a parent's
container, potentially making that value accessible to a concurrent sibling.
This is prevented throughout the type system, except in one place: at the
join point of \eparname{}.
At this point, the two sub-tasks have completed and their parent inherits
the values they allocated.
To allow the parent task to access these values,
\typedis restamps the result of \eparname{} with the timestamp of
the parent.
We dub this operation \textbf{backtiming}.

\typedis features first-class function types
$(\alpha \to^\delta \beta)$, annotated
by a timestamp variable $\delta$, indicating which task the function may be
called by.
Timestamp variables can be universally quantified, 
effectively allowing for \textbf{timestamp polymorphism}.
For example, pure
functions that have no side-effects are type-able as
$(\forall \delta.\, \alpha \to^\delta \beta)$, indicating that the function
may be safely called by any task.
\typedis also allows for \textbf{constrained timestamp polymorphism}.
For example, a function of type
$(\forall \prece{\delta'}{\delta}.\ \textsf{string}@{\delta'} \to^\delta \vunit{})$
only accepts as argument strings timestamped at some $\delta'$
that precede the timestamp $\delta$ of the calling task.
Typically, such constraints arise from the use of closures, especially
those that close over mutable state.

\paragraph{Soundness}
The soundness of \typedis is verified in the Rocq prover (the new name of the Coq proof assistant)
on top of the Iris higher-order concurrent separation logic framework~\cite{iris}.
We use the approach of \emph{semantic typing}~\citep{ConstableBook1986, LoefARTICLE1982, timany-krebbers-dreyer-birkedal-24},
and define a logical relation
targeting a variation of \olddislog~\citep{moine-westrick-balzer-24},
from which we reuse the technical parts.
As illustrated by RustBelt~\citep{rustbelt-18},
semantic typing facilitates manual verification of programs that are correct (e.g. disentangled),
but ill-typed, by carrying out a logical relation inhabitation proof using the program logic---overcoming incompleteness inherent to any type system.
%
For example, in the case of \typedis, this allows the user to verify part of the code
that use down-pointers.

We note that, similar to many other type systems (such as those in OCaml and
Haskell), \typedis relies on dynamic checks in the operational semantics to
enforce memory safety for out-of-bounds (OOB) array accesses.
The formal statement of soundness (\cref{sec:soundness_statement}) therefore
explicitly distinguishes between three kinds of program states: those that
have terminated, those that can step, and those that are stuck due to
OOB.
The soundness theorem states that all executions of programs typed within
\typedis always remain disentangled throughout execution.

\paragraph{Contributions}
Our contributions include:
\begin{itemize}
\item \typedis, the first static type system for disentanglement.
It includes the notion of a timestamp
to track which object is accessible by which task.
\typedis offers \mbox{(iso-)}recursive types as well as
polymorphism over types and over timestamps.
Moreover, \typedis supports polymorphic recursion over timestamps,
and offers a relaxation of the value restriction.
\item Two mechanisms to update a timestamp annotation:
via
subtiming, a form of subtyping,
and specifically at join points via the new operation of backtiming.
\item A new model for disentanglement with cyclic computation graphs.
We prove this model equivalent to the standard one and explain
why it is more amenable to verification.
\item A soundness proof of \typedis mechanized
in the Rocq prover
using the Iris framework.
We use semantic typing~\citep{timany-krebbers-dreyer-birkedal-24}
and \dislog, an improved version of \olddislog.
\item A range of case studies,
including building and iterating over an immutable tree in parallel,
as well the challenging example of deduplication via concurrent hashing.
\end{itemize}

\section{Key Ideas}
\label{sec:key_ideas}
In this section, we cover the key ideas of our work.
We start by recalling the definition of disentanglement~(\cref{sec:background}).
We then present the main idea of \typedis: adding task identifiers,
specifically \emph{timestamp variables}, to types~(\cref{sec:typedis101}).
Based on examples, we then illustrate two core principles of \typedis, allowing for updating timestamps
within types: backtiming~(\cref{sec:backtiming}) and subtiming~(\cref{sec:keysubtiming}).

\subsection{Preliminaries}
\label{sec:background}

\paragraph{Nested fork-join parallelism and task trees}
We consider programs written in terms of a single parallel primitive:
$\epar{f_1}{f_2}$, which creates two new child tasks $f_1()$ and
$f_2()$ to execute in parallel, waits for both of the child tasks to complete, and then
returns the results of the two calls as an immutable pair.
Creating the two child tasks is called a \defn{fork}, and waiting for the
two children to complete is called a \defn{join}.
The behavior of the \eparname{} primitive guarantees that every fork has a
corresponding join.
Any task may (recursively) fork and join, facilitating \emph{nested} parallelism and
giving rise to a dynamic tree during execution
called the \defn{task tree}.
The nodes of the task tree correspond to (parent) tasks that are waiting for
their children to join, and the leaves of the task tree correspond to tasks
which may actively take a step.
Whenever two sibling tasks join, the children are removed
from the tree and the parent resumes as a leaf task.
The task tree therefore dynamically grows and shrinks as tasks fork and join.
In this paper, we will use the letter~$t$ to denote tasks (leaves of the
task tree), and will equivalently refer to these as \defn{timestamps}.

\paragraph{Computation graphs}
\label{sec:altapproach}
The evolution of the task tree over time can be recorded as a
\defn{computation graph}, where vertices correspond to tasks and edges
correspond to scheduling dependencies.
The computation graph records not just the current task tree, but also
the history of tasks that have joined.
When a task $\vertex$ forks into two children $\vertex_1$ and $\vertex_2$,
two edges
$(\vertex, \vertex_1)$ and $(\vertex, \vertex_2)$ are added to the graph;
later when $\vertex_1$ and
$\vertex_2$ join, two edges $(\vertex_1, \vertex)$ and $(\vertex_2, \vertex)$
are added to the graph.
We say that $\vertex$ \emph{precedes}~$\vertex'$
in graph~$\graph$ and write $\reachable\graph\vertex{\vertex'}$, when
there exists a sequence of edges from $\vertex$ to $\vertex'$.
Note that $\smallpreceq$ is reflexive.
Two tasks are \defn{concurrent} when neither precedes the other.

\paragraph{Cyclic versus standard computation graphs}
\label{sec:cyclic}

We contribute a new definition of computation graphs,
which we call the \defn{\newapproach},
that differs slightly from the standard presentation used in
prior work~\citep{dag-calculus,westrick-et-al-20,moine-westrick-balzer-24}.
The standard approach is to use a fresh task identifier at each join point,
effectively renaming the resumed parent task.
In the \newapproach, we instead use the same
task identifier after the join point.
\Cref{fig:comparison} illustrates the difference between the two approaches.
It presents two computation graphs representing the same computation:
\cref{fig:graphstandard} shows the standard approach, and \cref{fig:graphnew}
shows the (new) \newapproach.
The distinction occurs when two tasks join.
In \cref{fig:graphstandard} tasks $t_3$ and $t_4$ join and form a new task $t_2'$
whereas in \cref{fig:graphnew} the two tasks join by going back to task $t_2$.
This distinction occurs again when~$t_2'$ (resp. $t_2$) join to form $t_0'$ (resp. $t_0$).
\begin{figure}\centering\small
\renewcommand{\captionlabel}[2]{%
  \vspace{-0.3em}%
  \caption{#1}%
  \Description{#1}%
  \label{#2}
  \vspace{-0.4em}%
}
\begin{subfigure}{0.47\textwidth}\centering
\begin{tikzpicture}[x=0.5cm,y=0.5cm]
\node (n1) at (0,0) {$t_0$};
\node (n3) at (2,-2) {$t_1$};
\node (n2) at (2,2) {$t_2$};
\node (n4) at (4,3) {$t_3$};
\node (n5) at (4,1) {$t_4$};
\node (n6) at (6,2) {$t_2'$};
\node (n7) at (8,0) {$t_0'$};
\node (n8) at (10,1) {$t_5$};
\node (n9) at (10,-1) {$t_6$};

\draw[->] (n1) -- (n2);
\draw[->] (n2) -- (n4);
\draw[->] (n2) -- (n5);
\draw[->] (n4) -- (n6);
\draw[->] (n5) -- (n6);
\draw[->] (n1) -- (n3);
\draw[->] (n6) -- (n7);
\draw[->] (n3) -- (n7);

\draw[->] (n7) -- (n8);
\draw[->] (n7) -- (n9);
\end{tikzpicture}
\captionlabel{Standard computation graph}{fig:graphstandard}
\end{subfigure}
\begin{subfigure}{0.47\textwidth}\centering
\begin{tikzpicture}[x=0.5cm,y=0.5cm]
\node (n1) at (0,0) {$t_0$};
\node (n3) at (2,-2) {$t_1$};
\node (n2) at (2,2) {$t_2$};
\node (n4) at (4,3) {$t_3$};
\node (n5) at (4,1) {$t_4$};
\node (n8) at (7,1) {$t_5$};
\node (n9) at (7,-1) {$t_6$};

\draw[->] (n1) to[bend left=15] (n2);
\draw[->] (n2) to[bend left=15] (n1);
\draw[->] (n1) to[bend right=15] (n3);
\draw[->] (n3) to[bend right=15] (n1);

\draw[->] (n2) to[bend left=15] (n4);
\draw[->] (n4) to[bend left=15] (n2);
\draw[->] (n5) to[bend right=15] (n2);
\draw[->] (n2) to[bend right=15] (n5);

\draw[->] (n1) to[bend right=5] (n8);
\draw[->] (n1) to[bend left=5] (n9);
\end{tikzpicture}
\captionlabel{Equivalent ``cyclic'' approach used in this paper}{fig:graphnew}
\end{subfigure}
\captionlabel{Comparison of two computation graphs equivalent for disentanglement}{fig:comparison}%
\vspace{-0.5em}
\end{figure}

The \newapproach considerably reduces
the need to manipulate timestamps,
not only in our proofs (for example the soundness proof of backtiming), but
also in the design of the type system
itself as well as in the underlying program logic~(\cref{sec:dislog2}).
We prove the two approaches equivalent for the purpose of verifying disentanglement~\citep{mechanization}.
Intuitively, the two approaches are equivalent
because we never need to check
reachability between two tasks that have both completed.
We have formally proven this equivalence with a simulation
theorem: every reduction in the semantics with the standard approach
implies the existence of a reduction with the same scheduling reaching the same expression
in the semantics with the \newapproach, and vice-versa.
Moreover, if one state is disentangled in one semantics, so it is in the other.

\paragraph{Roots}
At any moment, every task has a set of task-local \defn{roots} which are the
memory locations directly mentioned within a subexpression of that task.
For example, the expression
`$\elet x {\eprod{\loc_1}{\loc_2}} {\kw{fst}(x)}$'
has roots $\{\loc_1,\loc_2\}$, where (formally) $\loc_1$ and $\loc_2$ are
locations within the memory store.
Note that the roots of a task change over time: for example, the above
expression eventually steps to $\loc_1$ at which point it only has one root,
$\{\loc_1\}$.
The set of roots can grow due to allocations and loads from memory.

\paragraph{Disentanglement}

Disentanglement restricts the set of possible task-local roots.
\textbf{A program state is disentangled if
each root of a task has been
allocated by some preceding
task}.
More precisely, a program state with a computation graph $G$
is disentangled if,
for a root~$\loc$ of a task~$\vertex$,~$\loc$
has been allocated by a task $\vertex'$ such that
$G \vdash \prece{\vertex'}{\vertex}$,
that is, such that $\vertex'$ precedes $\vertex$ in $G$.
Following the computation graph definition,
preceding tasks include $\vertex$ itself, parent tasks,
but also children tasks that have terminated.
The formal definition of disentanglement appears in \Cref{sec:defdis}.

\typedis, the type system we present, verifies
that a program is disentangled, that is,
every reachable program state is disentangled.

\newcommand{\parzero}{\textsf{par}_0\xspace}

\subsection{TypeDis 101: Timestamps in Types}
\label{sec:typedis101}
In order to keep track of which task allocated which location,
\typedis incorporates timestamps in types.
More precisely, every heap-allocated (``boxed'') type is annotated by a
\emph{timestamp variable}, written~$\delta$, which can be understood as the
timestamp of the task that allocated the underlying location.
For example, a reference allocated by task $\delta$
on an (unboxed) integer
has type $\tref{\tinteger}@\delta$.

\paragraph{Timestamp polymorphism}
Functions in \typedis are annotated by a timestamp variable, restricting
which task they may run on.
Such a variable can be universally quantified, allowing for functions to
be run by different tasks.
%
For example, consider the function \codeinline{fun x -> newref(x)} which
allocates a new mutable reference containing an integer \codeinline{x}.
This function can be given the type
$\forall \delta.\; \tinteger \rightarrow^\delta \tref{\tinteger}@\delta$.
The superscript $\delta$ on the arrow indicates that the function must run on a
task at timestamp $\delta$, and the result type $\tref{\tinteger}@\delta$
indicates that the resulting reference will be allocated
at the same timestamp~$\delta$.
By universally quantifying $\delta$, the function is permitted to run on
any task, with the type system tracking that the resulting reference will be
allocated at the same timestamp as the caller.

\paragraph{Type polymorphism}
\typedis allows type variables~$\tvar$ to be universally quantified.
%
Using type polymorphism, we can now give the function \codeinline{fun x -> newref(x)}
the more general type
$\forall \tvar.\forall \delta.\; \tvar \rightarrow^\delta \tref{\tvar}@\delta$,
indicating that it is polymorphic in the type $\tvar$ of the contents of
the mutable reference.
Corresponding \refget and \refset primitives for mutable references are then
typed as shown in~\cref{fig:reftype}, all of which are polymorphic in
the type variable $\tvar$.
For functions with multiple arguments, such as \codeinline{set},
we adopt the notational convention to only specify the timestamp variable on the last arrow.

\paragraph{The up-pointer invariant}
In \cref{fig:reftype}, the type of \refget is given as
$\forall\alpha.\,\forall \delta\,\delta'.\,\tref{\alpha}@\delta' \rightarrow^\delta \alpha$.
Note that this type
is parameterized over both a caller time $\delta$ as well as
a (potentially different) timestamp $\delta'$ associated with the input
reference.
Intuitively, this type specifies that \refget is safe to call at any moment,
by any task, with any reference given as argument.
The design of \typedis in general guarantees that all loads from memory,
both mutable and immutable, are always safe.
Specifically, this is guaranteed by enforcing an invariant that we call
the \defn{up-pointer invariant}: all data structures in the language may
only contain values allocated at the same timestamp or a preceding timestamp.
For example, given two non-equal timestamps $\delta_1$ and $\delta_2$ where
${\delta_1} \prec {\delta_2}$, the type $\tref{\tref{\tinteger}@\delta_1}@\delta_2$
is valid, but $\tref{\tref{\tinteger}@\delta_2}@\delta_1$ is not.

\paragraph{Closures}
In \typedis, functions are first-class values and may be passed as arguments
to other functions, or stored in data structures, etc.
Function values are implemented as heap-allocated closures~\cite{landin-64,appel-92}, and must be
given a timestamp indicating when they were allocated.
%
For example, consider the definition of function $f$ in~\cref{fig:closures}, which closes over a mutable reference $r$ and an
immutable string $w$, both allocated at timestamps $\delta_0$ which
(in this example) is the timestamp of the current task.
We can give $f$ the type $(\forall \delta.\, \tunit \to^{\delta} \tunit)@\delta_0$,
indicating that $f$ itself was allocated at timestamp $\delta_0$.
Additionally, the type of $f$ specifies that it may be freely called at any
timestamp; this is safe for disentanglement because $f$ preserves the
up-pointer invariant, regardless of \emph{when} it will be called.
Contrast this with the definition of function $g$, which (when called)
allocates a new string and writes this string into the reference $r$.
If $g$ were called at some timestamp $\delta_1$ where $\delta_0 \prec \delta_1$,
then this would violate the up-pointer invariant for $r$.
The function $g$ does however admit the type
$(\tinteger \to^{\delta_0} \tunit)@\delta_0$,
indicating that $g$ may be safely called only by tasks at time $\delta_0$
(the same timestamp as the reference $r$).

\begin{figure}\small
\begin{minipage}{0.42\textwidth}
\[\begin{array}{r@{\;\;}c@{\;\;}l}
\refref&:&\forall\alpha.\,\forall \delta.\; \alpha \rightarrow^\delta \tref{\alpha}@\delta\\
\refget&:&\forall\alpha.\,\forall \delta\,\delta'.\;\tref{\alpha}@\delta' \rightarrow^\delta \alpha\\
\refset&:&\forall\alpha.\,\forall \delta\,\delta'.\;\tref{\alpha}@\delta' \rightarrow \alpha \rightarrow^\delta \tunit
\end{array}\]

\vspace{-0.1in}
\captionlabel{Example: typing reference primitives}{fig:reftype}
\end{minipage}\hfill
\begin{minipage}{0.55\textwidth}
\begin{minted}[escapeinside=??]{ocaml}
let r = newref "hello"   ?${\color{blue}r\,:\,\tref{\textsf{string}@\delta_0}@\delta_0}$?
let w = "world"          ?${\color{blue}w\,:\,\textsf{string}@\delta_0}$?
let f () = set r w       ?${\color{blue}f\,:\,(\forall \delta.\,\tunit \rightarrow^{\delta} \tunit)@\delta_0 }$?
let g i =                ?${\color{blue}g\,:\,(\tinteger \rightarrow^{\delta_0} \tunit)@\delta_0 }$?
  set r (Int.to_string i)
\end{minted}
\vspace{-0.1in}
\captionlabel{Example: typing closures}{fig:closures}
\end{minipage}
\end{figure}

\newcommand{\itree}[1]{\textsf{tree}@#1}
\begin{figure}\small
\begin{minted}[escapeinside=??]{ocaml}
type ?$\itree\delta$? = ?$(\,\tinteger + \tproduct{\itree\delta}{\itree\delta}@\delta\,)@\delta$?
let leaf x = inj1 x                           ?${\color{blue}\textsf{leaf}\,:\,(\forall \delta.\,\tinteger \to^\delta \textsf{tree}@\delta)@\delta_0}$?
let node x y = inj2 (x,y)                     ?${\color{blue}\textsf{node}\,:\,(\forall \delta.\,\textsf{tree}@\delta \to \textsf{tree}@\delta \to^\delta \textsf{tree}@\delta)@\delta_0}$?
let rec build n x =                           ?${\color{blue}\textsf{build}\,:\,(\forall\delta.\,\tinteger \rightarrow \tinteger \rightarrow^\delta \textsf{tree}@\delta)@\delta_0}$?
  if n <= 0 then leaf x else
  let n' = n - 1 in
  let (l,r) = par (fun () -> build n' x) (fun () -> build n' (x + pow2 n')) in
  node l r
\end{minted}
\captionlabel{Example: building a tree in parallel}{fig:build}
\end{figure}

\subsection{Backtiming the Result of a \eparname{}}
\label{sec:backtiming}
As explained earlier~(\cref{sec:background}),
we consider in this paper the parallel primitive \eparname{}(...),
which executes two closures in parallel and returns their result as an immutable pair.
The \eparname{} primitive can be used to build data structures in parallel.
Consider the code presented in \Cref{fig:build}.
The recursive type \ocaml{?$\itree\delta$? = ?$(\,\tinteger + \tproduct{\itree\delta}{\itree\delta}@\delta\,)@\delta$?}
describes a binary tree with integer leaves.
It consists of an immutable sum of either an integer (a leaf)
or a product of two subtrees (a node).
All the parts of a tree are specified in the type to have been
allocated at the same timestamp~$\delta$.
A leaf is built with the first injection, and a node with the second injection.
The function \ocaml{build n x} builds in parallel
a binary tree of depth $n$, with leaves labeled
from $x$ to $x + 2^n - 1$ in left-to-right order.

\typedis type-checks
\ocaml{build} with the type $\forall\delta.\; \tinteger \rightarrow \tinteger \rightarrow^\delta \textsf{tree}@\delta$.
The reader may be surprised:
we announced that the type \ocaml{tree?@$\delta$?}
has all of its parts allocated at the same timestamp~$\delta$,
but we are showing a function that builds a tree in parallel,
hence with some parts allocated by different tasks at different
timestamps. What's the trick?

The key observation is that
\emph{we can pretend that the objects allocated by a completed sub-task were
instead allocated by its parent}.
Indeed, disentanglement prevents sharing of data allocated
in parallel, but as soon as the parallel phase has ended, there is
no restriction anymore!

In \typedis, the \eparname{} primitive implements
\textbf{backtiming}, meaning that it replaces
the timestamp of the child task by the timestamp of the parent task
in the return type of the closures executed in parallel.
Indeed, the \eparname{} primitive admits the following, specialized for  \ocaml{build}, type:
\[\forall\delta\,\delta_l\,\delta_r.\, (\forall\delta'.\,\tunit \rightarrow^{\delta'} \itree{\delta'})@\delta_l \rightarrow (\forall\delta'.\,\tunit \rightarrow^{\delta'} \itree{\delta'})@\delta_r \rightarrow^\delta (\itree\delta \times \itree\delta)@\delta\]
This type for \eparname{} does exactly what we need:
it returns the result of the two closures in a pair as if they were called
at time $\delta$.
Backtiming is a powerful feature: it reduces parallelism
to almost an implementation detail.
Indeed, the type of \ocaml{build} does not reveal its internal use of parallelism.

\subsection{Making Something New out of Something Old with Subtiming}
\label{sec:keysubtiming}
A common practice (especially in functional programming) is
data structural \emph{sharing}, where components of an old structure are
reused inside part of a new structure.
In the context of \typedis, data structural sharing is interesting in
that it mixes data of potentially different timestamps within the same
structure.
Here we consider one such example and describe a key feature of
\typedis which enables such ``mixing'' of timestamps.

\begin{figure}\small
\begin{minted}[escapeinside=??]{ocaml}
let rec selectmap p f t =               ?${\color{blue}\textsf{selectmap}\,:\,(\forall\delta\,\delta_p\,\delta_f\,\delta_t.\; (\forall\delta'.\; \tinteger \rightarrow^{\delta'} \tbool)@\delta_p}$?
  match t with                                              ?${\color{blue}\rightarrow  (\forall\delta'.\; \tinteger \rightarrow^{\delta'} \tinteger)@\delta_f}$?
  | inj1 x -> if p x then leaf (f x) else t                 ?${\color{blue}\rightarrow \itree{\delta_t} \rightarrow^{\delta} \itree{\delta})@\delta_0}$?
  | inj2 (l,r) ->
    let (nl,nr) = par (fun () -> selectmap p f l) (fun () -> selectmap p f r) in
    if nl == l && nr == r then t else node nl nr
\end{minted}

\vspace{-0.05in}
\captionlabel{Example: the selectmap function}{fig:selectmap}
\end{figure}

\Cref{fig:selectmap} presents the \ocaml{selectmap p f t} function,
which selectively applies the function \ocaml{f} to the leaves of the tree
\ocaml{t}, following a predicate \ocaml{p} on integers.
The \ocaml{selectmap} function traverses the tree in parallel and crucially
preserves sharing as much as possible.
Specifically, when none of the leaves of the tree satisfy the predicate, the
function returns the original input tree as-is, instead of building another
identical tree.
To type this function in \typedis,
it may not be immediately clear what the timestamp of the
resulting tree should be:
\ocaml{selectmap} might directly return the argument passed as argument
(potentially coming from an older task), or it might return a new tree.
\typedis type-checks \ocaml{selectmap} with the type
\[ \forall\delta\,\delta_p\,\delta_f\,\delta_t.\; (\forall\delta'.\; \tinteger \rightarrow^{\delta'} \tbool)@\delta_p \rightarrow  (\forall\delta'.\; \tinteger \rightarrow^{\delta'} \tinteger)@\delta_f \rightarrow \itree{\delta_t} \rightarrow^{\delta} \itree{\delta} \]
This type universally quantifies over $\delta$
(the timestamp at which \ocaml{selectmap} will run),
$\delta_p$ and $\delta_f$ (the timestamps of the two closure arguments),
and $\delta_t$ (the timestamp of the tree argument).
Crucially, the result is of type $\itree\delta$, as if the whole result tree was allocated by $\delta$. What's the trick?

\typedis supports \textbf{subtiming}, that is,
a way of ``advancing'' timestamps within a type, following the precedence.
The rules of subtiming are as follows.
For a mutable type (e.g. an array or a reference),
subtiming is \emph{shallow}:
the outermost timestamp can be updated, but not the inner timestamps;
this is due to well-known variance issues~\citep[\S15]{tapl}.
For an immutable type (e.g. products and sums),
subtiming is \emph{deep}:
any timestamp within the type can be advanced, as long as the up-pointer invariant is preserved.

For \ocaml{selectmap}, we need to use deep subtiming on the
recursive immutable type $\itree{\delta_t}$ in order to update it to $\itree{\delta}$.
How can we be sure that $\delta_t$, the timestamp of the tree, precedes $\delta$, the timestamp at which
we call \ocaml{selectmap}?
We unveil a key invariant of \typedis:
every timestamp of every memory location in scope precedes the ``current'' timestamp,
that is, the timestamp of the task executing the function.
In our case the current timestamp is precisely $\delta$.
We hence deduce that $\delta_t$ precedes $\delta$, allowing us to use subtiming
to ``restamp'' the value $t : \itree{\delta_t}$ as $t : \itree{\delta}$.

To allow the user to express additional knowledge
about the dependencies between timestamps,
\typedis annotates universal timestamp quantification with a set of \emph{constraints},
which are supposed to hold while typing the function body, and are verified
at call sites.
For example, the following function \ocaml{let par' f g = ignore (par f g)} that executes two
closures \ocaml{f} and {g} from unit to unit in parallel and ignores the result can be given the type:
\[\forall\delta\,\delta_1\,\delta_2.\, (\forall\delta'\, \prece{\delta}{\delta'}.\,\tunit \rightarrow^{\delta'} \tunit)@\delta_1 \rightarrow (\forall\delta'\, \prece{\delta}{\delta'}.\,\tunit \rightarrow^{\delta'} \tunit)@\delta_2 \rightarrow^\delta \tunit\]
This type says that, if \ocaml{par'} gets called at timestamp $\delta$ with arguments $f$ and $g$,
then $f$ and $g$ can assume that they will be called at timestamp $\delta'$ such that $\prece{\delta}{\delta'}$.
These constraints are discussed in \Cref{sec:typesystem}, and the
fully general type of \eparname{} is presented in \Cref{typepoly}.

\section{Syntax and Semantics}
\label{sec:syntax_and_semantics}
The formal language we study, dubbed \lang,
can be understood as an extension of \oldlang, the language studied by \citet{moine-westrick-balzer-24}.
\lang adds support for immutable pairs and sums, iso-recursive types,
and directly offers the \eparname{} primitive for fork-join parallelism.
We present the syntax of \lang~(\cref{sec:syntax}),
its semantics~(\cref{sec:semantics}), and the formal
definition of disentanglement~(\cref{sec:defdis}).

\subsection{Syntax}
\label{sec:syntax}
\begin{figure}
\centering\small
\newcommand{\commentary}[1]{ & \text{\small\it #1} \\}
\newcommand{\defineq}{\mathrel{::=}}
\[
\begin{array}{l@{\;\;}r@{\;\defineq\;}l}
\text{Values} &\val,\wal & \vunit \mid \vbool \in \{\vtrue,\vfalse\} \mid \vint \in \mathbb{Z} \mid {\color{blue}\loc \in \Loc} \mid {\color{blue}\vfold\val} \\

\text{Blocks} & \block &\wals \mid \eprod\val\val \mid \einj{\varoff \in \{0,1\}}{\val} \mid \bclo\funcname{\vec\var}\expr \\

\text{Primitives} & \primitive  & + \mid - \mid \times \mid \div \mid\textsf{mod} \mid \symbeq \mid \mathord{<} \mid \mathord{\leq} \mid \mathord{>} \mid \mathord{\geq} \mid \symbor \mid \symband \\

\text{Expressions} &\expr&
\begin{array}[t]{@{}l@{\;\;}c@{}}
\val \mid \var \in \Var \mid \elet{\var}{\expr}{\expr} \mid \eif{\expr}{\expr}{\expr} \mid \ecallprim{\expr}{\expr} &\\
\mid \eclo{\funcname}{\vec\var}{\expr} \mid \erunapp{\expr}{\vec\expr} & \textit{closures}\\
\mid \eprod\expr\expr \mid \eproj{\varoff \in \{1,2\}}{\expr} & \textit{pairs} \\
\mid \einj{i \in \{1,2\}}{\expr} \mid \ecase{\expr}{\var}{\expr}{\var}{\expr} &  \textit{sums} \\
\mid \ealloc{\expr}{\expr} \mid \eload{\expr}{\expr} \mid \estore{\expr}{\expr}{\expr} \mid \elength\expr & \textit{arrays} \\
\mid \efold{\expr} \mid \eunfold{\expr} & \textit{iso-recursive types}\\
\mid \epar{\expr}{\expr} \mid {\color{blue}\erunpar\expr\expr} \mid \ecas\expr\expr\expr\expr & \textit{parallelism and concurrency}\\
\end{array}\\
\text{Contexts} & \ectx &
\begin{array}[t]{@{}lllll}
\elet{\var}\khole\expr &\mid \eif\khole\expr\expr &\mid
\ealloc\khole\expr &\mid \ealloc\val\khole &\mid \elength\khole\\
\mid \eload\khole\expr &\mid \eload\val\khole &\mid \estore\khole\expr\expr &\mid \estore\val\khole\expr &\mid
\estore\val\val\khole \\
\mid \ecallprim\khole\expr &\mid \ecallprim\val\khole &\mid \erunapp\khole{\vec\expr}
&\mid \erunapp\val{(\vals \dplus \khole\dplus \vec\expr)} &\mid \efold{\khole}\\
\mid \eunfold\khole &\mid \eprod\khole\expr  &\mid \eprod\val\khole &\mid \eproj{\varoff}{\khole} &\mid \einj{i}\khole\\
\multicolumn{3}{l}{\kern-0.55em\mid \ecase\khole\var\expr\var\expr} &\mid \epar{\khole}{\expr} &\mid \epar{\val}{\khole} \\
\mid\ecas\khole\expr\expr\expr &\mid \ecas\val\khole\expr\expr &\mid \ecas\val\val\khole\expr &\mid \ecas\val\val\val\khole\\
\end{array}
\end{array}\]
\captionlabel{Syntax of \lang. Constructs \textcolor{blue}{in blue} are runtime-level.}{fig:syntax}
\end{figure}

The syntax of \lang appears in \Cref{fig:syntax}.
The constructs \textcolor{blue}{in blue}
are forbidden in the source program and
occur only at runtime.

A value $\val \in \Values$ can be the unit
value $\vunit$, a boolean $\vbool \in \{\vtrue,\vfalse\}$, an idealized integer $\vint \in \mathbb{Z}$, a
\emph{memory location} $\loc \in \Loc$, where $\Loc$ is an
infinite set of locations, or a folded value~$\vfold{\val}$,
witnessing our use of iso-recursive types~\citep[\S20]{tapl}.

A \emph{block} describes the contents of a heap cell,
amounting to either an array of values,
written~$\wals$,
an immutable pair~$\eprod{\val}{\val}$,
the first injection~$\einl{\val}$ or the second injection~$\einr{\val}$ of an immutable sum,
or a \emph{$\lambda$-abstraction}~$\bclo\funcname\argsname\body$.
Lambdas can close over free variables,
compilers of functional languages usually implement them as \emph{closures}~\citep{landin-64,appel-92}.
A closure is a heap-allocated object carrying a code pointer as well as
an environment, recording the values of the free variable.
Thus, acquiring a closure can create entanglement.
Moreover, because functions and tuples are heap allocated,
currying and uncurrying---that is,
converting a function taking multiple arguments to a function taking
a tuple of arguments and vice-versa---does not come for free.
Hence, we chose to support a version of the language were every
function takes possibly multiple arguments. 
Closure allocation is written
$\eclo{\funcname}{\vec\var}{\expr}$.
This notation binds a recursive name $\funcname$,
argument names $\vec\var$
in the expression $\expr$.
A function call is written
$\erunapp{\expr}{\vec\expr}$.

In~\lang, fork-join parallelism is available via the parallel
primitive~$\epar{\expr_1}{\expr_2}$,
which reduces~$\expr_1$ and~$\expr_2$ to closures,
calls them in parallel,
and returns their result as an immutable pair.
This parallel computation is represented by the
\emph{active} parallel pair~$\erunpar{\expr_1}{\expr_2}$,
appearing only at runtime.
\lang supports a compare-and-swap
instruction~$\ecas\expr\expr\expr\expr$, which targets an array, and is parameterized by 4
arguments: the
location of the array, the index in the array, the old value and the new value.
A (sequential) evaluation context~$\ectx$ describes a term with a hole,
written~$\khole$.
The syntax of evaluation contexts dictates a left-to-right
\emph{call-by-value} evaluation strategy.
Note that evaluation contexts $\ectx$ in this presentation are sequential.
Specifically, we intentionally excluded active parallel pairs $(\erunpar - -)$
from the grammar of $\ectx$.
The evaluation strategy for active parallel pairs allows for interleaving
of small steps, which is handled separately by a ``scheduler reduction''
relation in the operational semantics (\cref{sec:semantics}).

\subsection{Operational Semantics}
\label{sec:semantics}

\paragraph{Head reduction relation}
\begin{figure}
\centering\small
\begin{mathpar}
\inferrule[HeadAlloc]
{0 < \allocsize \\ \loc \notin \dom(\store) \\ \loc \notin \dom(\amap)}
{\headstep\graph\vertex\store\amap{\ealloc\allocsize\val}{\blockupd\store\loc{\blockrepeat\allocsize\val}}{\blockupd\amap\loc\vertex}{\loc}}
%

\inferrule[HeadLoad]
{ \store(\loc) = \wals\\ 0 \leq \ofs < \length\wals \\
  \wals(\ofs) = \val
}
{\headstep\graph\vertex\store\amap{\eload\loc\ofs}\store\amap\val}

\inferrule[HeadStore]
{ \store(\loc) = \wals \\ 0 \leq \ofs < \length\wals}
{\headstep\graph\vertex\store\amap{\estore\loc\ofs\val}{\blockupd\store\loc{\blockupd\wals\ofs\val}}\amap\vunit}

\inferrule[HeadLetVal]{}
{\structheadstep\graph\vertex\store\amap{\elet\var\val\expr}{\subst\var\val\expr}}
%


\inferrule[HeadIfTrue]{}
{\structheadstep\graph\vertex\store\amap{\eif\vtrue{\expr_1}{\expr_2}}{\expr_1}}

\inferrule[HeadIfFalse]{}
{\structheadstep\graph\vertex\store\amap{\eif\vfalse{\expr_1}{\expr_2}}{\expr_2}}

\inferrule[HeadClosure]
{\loc \notin \dom(\store) \\ \loc \notin \dom(\amap)}
{\headstep\graph\vertex\store\amap{\bclo\funcname\argsname\expr}{\blockupd\store\loc{\bclo\funcname\argsname\expr}}{\blockupd\amap\loc\vertex}{\loc}}

\inferrule[HeadCall]
{ \store(\loc) = \bclo\funcname\argsname\body \\
  \length\argsname = \length\wals
}
{ \structheadstep\graph\vertex\store\amap{\erunapp\loc\wals}{\subst\funcname\loc{\subst\argsname\wals\body}}}

\inferrule[HeadCallPrim]
{\purestep{\ecallprim{\val_1}{\val_2}}\val}
{\structheadstep\graph\vertex\store\amap{\ecallprim{\val_1}{\val_2}}\val}

\inferrule[HeadPair]
{\loc \notin \dom(\store) \\ \loc \notin \dom(\amap)}
{\headstep\graph\vertex\store\amap{\eprod{\val_\offleft}{\val_\offright}}{\blockupd\store\loc{\eprod{\val_\offleft}{\val_\offright}}}{\blockupd\amap\loc\vertex}{\loc}}

\inferrule[HeadProj]
{\store(\loc) = \eprod{\val_\offleft}{\val_\offright}
}
{\structheadstep\graph\vertex\store\amap{\eproj\varoff\loc}{\val_\varoff}}

\inferrule[HeadInj]
{\loc \notin \dom(\store) \\ \loc \notin \dom(\amap)}
{\headstep\graph\vertex\store\amap{\einj{i}{\val}}{\blockupd\store\loc{\einj{i}{\val}}}{\blockupd\amap\loc\vertex}{\loc}}

\inferrule*[left=HeadCase]
{\store(\loc) = \einj\varoff{\val}
}
{\structheadstep\graph\vertex\store\amap{(\ecase{\loc}{x_\offleft}{\expr_\offleft}{x_\offright}{\expr_\offright})}{\subst{x_\varoff}{\val}{\expr_\varoff}}}

\inferrule[HeadFold]{}
{\structheadstep\graph\vertex\store\amap{\efold\val}{\vfold{\val}}}

\inferrule[HeadUnfold]{}
{\structheadstep\graph\vertex\store\amap{\eunfold{(\vfold\val)}}{\val}}
\end{mathpar}
\captionlabel{Head reduction (selected rules)}{fig:head_semantics}
\end{figure}

A \emph{head configuration} $\hconfig{\store}{\amap}{\expr}$ is composed
of a store~$\store$, an allocation map~$\amap$, and an expression~$\expr$.
The store~$\store$ represents the heap and
consists of a finite map of locations to blocks.
The allocation map~$\amap$ is a finite map of locations to
timestamps, recording the timestamps at which locations were
allocated.
\Cref{fig:head_semantics} presents parts of the definition
of the head reduction relation
between two head configurations
$\headstep\graph\vertex\store\amap\expr{\store'}{\amap'}{\expr'}$
occurring at the (local) task of timestamp~$\vertex$
in the (global) computation graph~$\graph$.
A head configuration consists of the expression $\expr$ being evaluated,
the store~$\store$, and an allocation map~$\amap$.
\Cref{fig:head_semantics} omits rules for the length array primitive as well
as the atomic compare-and-swap on arrays.

We write~$\store(\loc)$ to denote the block stored at the location~$\loc$ in the store~$\store$.
We write~$\blockupd\store\loc\block$ for the insertion of block $\block$ at location $\loc$ in $\store$.
Note that only arrays can be updated; closures, pairs and sums are immutable.
We write~$\wals(\ofs)$ to refer to the index~$\ofs$ of an array~$\wals$.
We write $\blockupd\wals\ofs\val$ for an update to an array,
and we similarly write~$\blockupd\amap\loc\vertex$ for an insertion in
the allocation map.
We write $\val^\allocsize$ for an array of length~$\allocsize$,
where each element of the array is initialized with the value~$\val$.

\RULE{HeadAlloc} allocates an array, extending the store and the allocation map.
\RULE{HeadLoad} acquires the value $\val$ from an index of an array.
\RULE{HeadStore}, \RULE{HeadLetVal}, \RULE{HeadIfTrue} and \RULE{HeadIfFalse} are standard.
\RULE{HeadClosure} allocates a closure
and \RULE{HeadCall} calls a closure.
\RULE{HeadCallPrim} calls a primitive, whose result is computed at the meta-level by
the $\smash{\puresteparrow}$ relation.
\RULE{HeadPair} and \RULE{HeadProj} allocate and project immutable pairs, respectively.
\RULE{HeadInj} and \RULE{HeadCase} allocate and case over immutable sums, respectively.
\RULE{HeadFold} and \RULE{HeadUnfold} handle iso-recursive types in
a standard way.

\paragraph{Scheduler reduction relation}
\renewcommand{\state}{S}
\newcommand{\smallconfig}[3]{#1\sepconfig#2\sepconfig#3}
\newcommand{\phruple}[3]{#1\sepconfig#2\sepconfig#3}
\newcommand{\thruple}[3]{(#1,#2,#3)}
\newcommand{\tasktree}{T}
\newcommand{\tleaf}[1]{#1}
\newcommand{\tpar}[3]{#2\otimes_{#1} #3}
\newcommand{\littlenode}[3]{\tpar{#1}{\tleaf{#2}}{\tleaf{#3}}}

\newcommand{\structparastep}[8]{\parastep{\smallconfig{\phruple{#1}{#2}{#3}}{#4}{#5}}{\smallconfig{\phruple{#1}{#2}{#6}}{#7}{#8}}}

\begin{figure}
\centering\small
\begin{mathpar}
\inferrule[SchedHead]
{\headstep\graph\vertex\store\amap\expr{\store'}{\amap'}{\expr'}}
{\parastep{\smallconfig{\phruple\store\amap\graph}{\tleaf\vertex}\expr}{\smallconfig{\phruple{\store'}{\amap'}\graph}{\tleaf\vertex}{\expr'}}}

\inferrule[SchedFork]
{\vertex_1,\vertex_2 \notin \vertices\graph \\\\
  \graph' = \graph \cup \{\graphedge{\vertex}{\vertex_1},\graphedge{\vertex}{\vertex_2} \} \\
  \expr' = \erunpar{\ecall{\val_1}{[\vunit]}}{\ecall{\val_2}{[\vunit]}}
}
{\structparastep\store\amap\graph{\tleaf{\vertex}}{\epar{\val_1}{\val_2}}{\graph'}
{\littlenode{\vertex}{\vertex_1}{\vertex_2}}{\expr'}}

\inferrule[SchedJoin]
  { \loc \notin \dom(\store) \\
    \loc \notin \dom(\amap) \\
   \graph' = \graph \cup \{\graphedge{\vertex_1}{\vertex},\graphedge{\vertex_2}{\vertex} \}
  }
  {\parastep
    {\smallconfig{\phruple\store\amap\graph}{\littlenode{\vertex}{\vertex_1}{\vertex_2}}{\erunpar{\val_1}{\val_2}}}
    {\smallconfig{\phruple{\blockupd\store\loc{\eprod{\val_1}{\val_2}}}{\blockupd\amap\loc{\vertex}}{\graph'}}{\tleaf{\vertex}}{\loc}}}

\inferrule[StepSched]
{\parastep
  {\smallconfig{\phruple\store\amap\graph}\tasktree\expr}
  {\smallconfig{\phruple{\store'}{\amap'}{\graph'}}{\tasktree'}{\expr'}}}
{\step
  {\smallconfig{\thruple\store\amap\graph}\tasktree\expr}
  {\smallconfig{\thruple{\store'}{\amap'}{\graph'}}{\tasktree'}{\expr'}}}

\inferrule[StepBind]
{\step
  {\smallconfig\state\tasktree\expr}
  {\smallconfig{\state'}{\tasktree'}{\expr'}}}
{\step
  {\smallconfig\state\tasktree{\efillctx\ectx\expr}}
  {\smallconfig{\state'}{\tasktree'}{\efillctx\ectx{\expr'}}}}

\inferrule[StepParL]
{\step
  {\smallconfig\state{\tasktree_1}{\expr_1}}
  {\smallconfig{\state'}{\tasktree_1'}{\expr_1'}}}
{\step
  {\smallconfig\state{\tpar\vertex{\tasktree_1}{\tasktree_2}}{\erunpar{\expr_1}{\expr_2}}}
  {\smallconfig{\state'}{\tpar\vertex{\tasktree_1'}{\tasktree_2}}{\erunpar{\expr_1'}{\expr_2}}}}

\inferrule[StepParR]
{\step
  {\smallconfig\state{\tasktree_2}{\expr_2}}
  {\smallconfig{\state'}{\tasktree_2'}{\expr_2'}}}
{\step
  {\smallconfig\state{\tpar\vertex{\tasktree_1}{\tasktree_2}}{\erunpar{\expr_1}{\expr_2}}}
  {\smallconfig{\state'}{\tpar{\vertex}{\tasktree_1}{\tasktree_2'}}{\erunpar{\expr_1}{\expr_2'}}}}
\end{mathpar}
\captionlabel{Reduction under a context and parallelism}{fig:semantics}
\end{figure}

In order to keep track of the timestamp of each task
and whether the task is activated or suspended,
we follow \citet{westrick-et-al-20}
and enrich the semantics with an
auxiliary structure called a \emph{task tree}, written~$\tasktree$,
of the following formal grammar: ${\tasktree \;\eqdef\; \vertex \in \Vertices \mid \tpar\vertex\tasktree\tasktree}$.
A leaf $\vertex$ indicates an active task denoted by its timestamp.
A node~$\tpar\vertex{\tasktree_1}{\tasktree_2}$ represents a suspended task $\vertex$
that has forked
two parallel computations, recursively described by the task trees~$\tasktree_1$
and $\tasktree_2$.

\cref{fig:semantics} presents
the scheduling reduction relation $\smash{\parastepinline{\smallconfig{\phruple\store\amap\graph}\tasktree\expr}{\smallconfig{\phruple{\store'}{\amap'}{\graph'}}{\tasktree'}{\expr'}}}$ as either a head step, a fork, or a join.
In this reduction relation,~$\store$ is a store,~$\amap$ an allocation map,~$\graph$ a computation graph,~$\tasktree$ a task tree, and~$\expr$ an expression.
\RULE{SchedHead} reduction describes a head reduction.
\RULE{SchedFork} reduction describes a fork:
the task tree consists of a leaf~$\vertex$ and the expression~$\epar{\val_1}{\val_2}$, where both $\val_1$ and $\val_2$
are closures to be executed in parallel.
The reduction generates two fresh timestamps~$\vertex_1$ and~$\vertex_2$,
adds the corresponding edges to the computation graph, and
updates the task tree to comprise the node with two leaves~$\littlenode{\vertex}{\vertex_1}{\vertex_2}$.
The reduction then updates the expression to the active
parallel pair $\erunpar{\ecall{\val_1}{[\vunit]}}{\ecall{\val_2}{[\vunit]}}$,
reflecting the parallel call of the two closures $\val_1$ and $\val_2$, each one called with a single argument, the unit value $\vunit$.
\RULE{SchedJoin} reduction describes a join
and differs from prior semantics for disentanglement~\citep{westrick-arora-acar-22,moine-westrick-balzer-24}
because it reuses a timestamp~(\cref{sec:altapproach}).
The task tree is at a node $\vertex$ with two
leaves~$\littlenode\vertex{\vertex_1}{\vertex_2}$, and both leaves reached a value.
The reduction adds edges $(\vertex_1,\vertex)$ and $(\vertex_2,\vertex)$ to the computation graph,
and allocates a memory cell to store the result of
the (active) parallel pair.
It then updates the task tree to the leaf~$\vertex$.

\paragraph{Parallelism and reduction under a context}
The lower part of \cref{fig:semantics} presents
the main reduction relation~$\smash{\stepinline
{\smallconfig\state\tasktree\expr}
{\smallconfig{\state'}{\tasktree'}{\expr'}}}$, which describes a scheduling
reduction inside the whole parallel program~\citep{moine-westrick-balzer-24}.
A \emph{configuration} $\smallconfig\state\tasktree\expr$ consists of the program state~$\state$, the task tree~$\tasktree$, and an expression~$\expr$.
This expression $\expr$ can consist of multiple tasks,
governed by the nesting of active parallel pairs $(\erunpar{e_1}{e_2})$.
The corresponding timestamps of these tasks are given by the accompanying task
tree $T$.
A state~$\state$ consists of the tuple $(\store, \amap, \graph)$, denoting a store~$\store$, an allocation map~$\amap$, and a computation graph~$\graph$.
\RULE{StepSched} reduction describes a scheduling step.
The other reductions describe \emph{where} the scheduling reduction
takes place in the whole parallel program.
\RULE{StepBind} reduction describes a reduction under an evaluation
context.
\RULE{StepParL} and~\RULE{StepParR} reductions are non-deterministic:
if a node of the task tree is encountered facing an active parallel pair,
the left side or the right side can reduce.

\subsection{Definition of Disentanglement}
\label{sec:defdis}
\newcommand{\disentangled}[3]{\textsf{Disentangled}\;\phruple{#1}{#2}{#3}}

\begin{figure}\centering\small
\begin{mathpar}
\inferrule[DELeaf]
{\forall\loc.\,\loc \in \locs\expr \implies \reachable\graph{\amap(\loc)}{\vertex}}
{\disentangled{\thruple{\_}{\amap}{\graph}}{\vertex}{\expr}}

\inferrule[DEPar]
{\disentangled{\state}{\tasktree_1}{\expr_1} \\
\disentangled{\state}{\tasktree_2}{\expr_2}
}
{\disentangled{\state}{\tpar\vertex{\tasktree_1}{\tasktree_2}}{\erunpar{\expr_1}{\expr_2}}}

\inferrule[DEBind]
{\state = \thruple{\_}{\amap}{\graph} \\
\disentangled\state{\tpar\vertex{\tasktree_1}{\tasktree_2}}\expr\\\\
\forall \loc.\,\loc \in \locs{\ectx} \implies  \forall\vertex'.\,\vertex' \in \leaves{\tasktree_1} \cup \leaves{\tasktree_2} \implies \reachable\graph{\amap(\loc)}{\vertex'}}
{\disentangled\state{\tpar\vertex{\tasktree_1}{\tasktree_2}}{\efillctx\ectx\expr}}
\end{mathpar}
\captionlabel{Definition of Disentanglement}{fig:defdis}
\end{figure}

The property $\disentangled{\state}{\tasktree}{\expr}$
asserts that,
given a program state~$\state$ and a task tree~$\tasktree$,
the expression~$\expr$ is disentangled---that is,
the roots of each task in $\expr$ were allocated by preceding tasks.
\Cref{fig:defdis} gives the inductive definition of
$\disentangled{\state}{\tasktree}{\expr}$.
If the program state has an allocation map $\amap$
and a computation graph $\graph$,
and if the task tree is a leaf~$\vertex$, \RULE{DELeaf}
requires for every location $\loc$ in $\locs\expr$,
that is, the set of locations syntactically occurring in $\expr$,
that the location~$\loc$ has been allocated by a task $\amap(\vertex)$ preceding
$\vertex$ in $\graph$.
If the task tree is a node $\tpar\vertex{\tasktree_1}{\tasktree_2}$,
there are two cases.
In the first case, if the expression is an active parallel pair,
\RULE{DEPar} requires that the two sub-expressions are disentangled.
Otherwise, the expression must be of the form~$\efillctx{\ectx}{\expr}$,
and then \RULE{DEBind} requires that $\expr$ itself is disentangled
and that for every location~$\loc$ occurring in the evaluation context $\ectx$,
the location $\loc$ has been allocated
before every leaf~$\vertex'$
of $\tasktree_1$ and $\tasktree_2$.

\paragraph{Difference with Previous Semantics for Disentanglement}
Inspired by \citet{westrick-arora-acar-22},
we equip \lang with a mostly standard semantics,
instrumented with a computation graph and an allocation map.
We then distinguish disentangled states using the
$\textsf{Disentangled}$ property,
resembling the ``$\textsf{rootsde}$'' invariant
proposed by \citet{westrick-arora-acar-22}.
The novelty of our approach resides in the instrumentation
with the, more amenable to verification,
cyclic computation graph~(\cref{sec:cyclic}).

DisLog~\citep{moine-westrick-balzer-24} chooses a
slightly different formalization in which
the semantics gets stuck if entanglements is detected.
Each time a task acquires a location from
the heap, the semantics performs a check to verify that
the location was allocated by a preceding task.
Intuitively, this check ensures by construction that a
program's evaluation reaches only states satisfying the
$\textsf{Disentangled}$ property.
Conversely, guaranteeing the
$\textsf{Disentangled}$ property at every step ensures that
a disentanglement check cannot fail.

\section{Type System}
\label{sec:typesystem}
In this section, we describe \typedis in depth.
First, we present the formal syntax of types~(\cref{sec:typesyntax})
as well as the typing judgment~(\cref{sec:typing_judgment}).
We then comment on typing rules for mutable heap blocks~(\cref{sec:heapblocks}),
which enforce disentanglement.
Next, we present the rules for creating and calling closures~(\cref{sec:abstractions}),
which are crucial for understanding our approach for
typing the \eparname{} primitive~(\cref{sec:parrule}).
We then focus on advanced features of \typedis :
general recursive types and type polymorphism~(\cref{sec:evolved}).
We conclude by presenting subtiming~(\cref{sec:subtiming}).

\subsection{Syntax of Types}
\label{sec:typesyntax}
\begin{figure}\small
\[\begin{array}{r@{\;\;}rcl}
\text{Timestamp variables} & \delta \\
\text{Type variables} & \tvar \\
\text{Logical graphs} & \Delta & \eqdef & \Delta,\,\prece{\delta}{\delta} \mid \emptyset\\
\text{Kinds} & \kind & \eqdef & \oldstar \mid  \kindsucc\kind \\
\text{Unboxed types} & \tau & \eqdef& \tunit \mid \tbool\mid\tinteger\\
\text{Boxed types} & \sigma & \eqdef&\tarray{\rho} \mid \tproduct\rho\rho \mid \tsum\rho\rho \mid
\tfun{\vec\delta}{\Delta}{\vec{\rho}}\delta\rho  \\
\text{Types} & \rho & \eqdef& \tau  \mid \tlambda\delta\rho \mid \tapp\rho\delta \mid \tforall\tvar\kind\rho \mid \trec\tvar\sigma\delta \mid \tvar \mid \sigma@\delta \\
\text{Environments} & \Gamma  & \eqdef& x:\rho,\Gamma \mid \alpha :: \kind, \Gamma\, \mid \emptyset
\end{array}\]
\captionlabel{Syntax of types}{fig:typesyntax}
\end{figure}

To reason statically about the runtime notions
of timestamps $\vertex$ and computation graphs $\graph$~(\cref{sec:semantics}),
we introduce their corresponding static notions:
\emph{timestamp variables} $\delta$ and
\emph{logical graphs}~$\Delta$, respectively.
A~logical graph $\Delta$ is a set of pairs $\prece{\delta_1}{\delta_2}$,
asserting that the timestamp $\delta_1$
\emph{precedes} the timestamp~$\delta_2$, that is,
everything allocated by the task at $\delta_1$
is safe to acquire for the task at $\delta_2$.
\Cref{fig:typesyntax} summarizes these notions together with the syntax of types.

A powerful feature of our type system is its support for \emph{timestamp polymorphism},
facilitated through \emph{higher-order types}.
This higher-order feature is instrumental in typing the \eparname{} primitive~(\cref{sec:parrule}),
and thus supporting the \newapproach detailed in~\cref{sec:altapproach}.
Because our system is higher-order, we introduce \emph{kinds},
written~$\kind$, which capture the number of timestamps a type
expects as arguments.
The ground kind, written~$\oldstar$, indicates that the type does not
take a timestamp argument.
The successor kind, written~$\skindsucc{\kind}$, indicates
that the type expects $\kind + 1$ timestamp arguments.

A base type $\tau$ describes an \emph{unboxed} value,
that is, a value that is not allocated on the heap.
Base types include the unit type, Booleans, and integers.

The syntax of types~$\rho$ is mutually inductive with the syntax
of boxed types~$\sigma$.
A type~$\rho$ is either a base type~$\tau$,
a type taking a timestamp argument~$\tlambda\delta\rho$,
an application of a type to a timestamp $\tapp{\rho}{\delta}$,
a universal quantification of a type variable with some kind $\tforall{\tvar}{\kind}{\rho}$,
a recursive type $\trec{\alpha}{\sigma}{\delta}$,
a type variable $\tvar$,
or a boxed type annotated with a timestamp~$\sigma@\delta$.
When the timestamp $\delta$ does not matter, we write $\sigma@\_$.
A boxed type $\sigma$ is either
an array~$\tarray{\rho}$, an immutable pair~$\tproduct{\rho}{\rho}$,
an immutable sum $\tsum\rho\rho$, or a function~$\tfun{\svec\delta}{\Delta}{\vec{\rho}}\delta\rho$.

Types support $\alpha$-equivalence for both type and timestamp variables,
as well as $\beta$-reduction.

\subsection{The Typing Judgment}
\label{sec:typing_judgment}
\begin{figure}\small
\begin{mathpar}
\inferrule[T-Var]
{\Gamma(x) = \rho}
{\judg\delta\Delta\Gamma{x}{\rho}}

\inferrule[T-Unit]{}
{\judg\delta\Delta\Gamma{\vunit}{\textsf{unit}}}

\inferrule[T-Int]{}
{\judg\delta\Delta\Gamma{\vint}{\textsf{int}}}

\inferrule[T-Bool]{}
{\judg\delta\Delta\Gamma{\vbool}{\textsf{bool}}}

\inferrule[T-Let]
{\judg\delta\Delta\Gamma{\expr_1}{\rho'} \\
\judg\delta\Delta{x : \rho', \Gamma}{\expr_2}{\rho}
}
{\judg\delta\Delta\Gamma{\elet{x}{\expr_1}{\expr_2}}{\rho}}

\inferrule[T-If]
{\judg\delta\Delta\Gamma{\expr_1}{\tbool} \\\\
\judg\delta\Delta\Gamma{\expr_2}{\rho} \\
\judg\delta\Delta\Gamma{\expr_3}{\rho}
}
{\judg\delta\Delta\Gamma{\eif{\expr_1}{\expr_2}{\expr_3}}{\rho}}

\inferrule[T-Pair]
{\judg\delta\Delta\Gamma{\expr_\offleft}{\rho_\offleft} \\
\judg\delta\Delta\Gamma{\expr_\offright}{\rho_\offright}
}
{\judg\delta\Delta\Gamma{\eprod{\expr_\offleft}{\expr_\offright}}{\tproduct{\rho_\offleft}{\rho_\offright}@\delta} }

\inferrule[T-Proj]
{\judg\delta\Delta\Gamma{\expr}{\tproduct{\rho_\offleft}{\rho_\offright}@\_}
}
{\judg\delta\Delta\Gamma{\eproj{\varoff}\expr}{\rho_i}}

\inferrule[T-Inj]
{\judg\delta\Delta\Gamma{\expr}{\rho_i}}
{\judg\delta\Delta\Gamma{\einj{\varoff}\expr}{\tsum{\rho_\offleft}{\rho_\offright}@\delta}}

\inferrule[T-Case]
{\judg\delta\Delta\Gamma{\expr}{\tsum{\rho_\offleft}{\rho_\offright}@\_}\\\\
\judg\delta\Delta{x_\offleft : \rho_\offleft, \Gamma}{\expr_\offleft}{\rho}\\
\judg\delta\Delta{x_\offright : \rho_\offright, \Gamma}{\expr_\offright}{\rho}}
{\judg\delta\Delta\Gamma{\ecase\expr{x_\offleft}{\expr_\offleft}{x_\offright}{\expr_\offright}}{\rho}}

\inferrule[T-Array]
{\judg\delta\Delta\Gamma{\expr_1}{\tinteger} \\
\judg\delta\Delta\Gamma{\expr_2}{\rho}}
{\judg\delta\Delta\Gamma{\ealloc{\expr_1}{\expr_2}}{\tarray{\rho}@\delta}}

\inferrule[T-Store]
{\judg\delta\Delta\Gamma{\expr_1}{\tarray{\rho}@\_} \\\\
\judg\delta\Delta\Gamma{\expr_2}{\tinteger} \\
\judg\delta\Delta\Gamma{\expr_3}{\rho}
}
{\judg\delta\Delta\Gamma{\estore{\expr_1}{\expr_2}{\expr_3}}{\tunit}}

\inferrule[T-Load]
{\judg\delta\Delta\Gamma{\expr_1}{\tarray{\rho}@{\_}} \\\\
\judg\delta\Delta\Gamma{\expr_2}{\tinteger}
}
{\judg\delta\Delta\Gamma{\eload{\expr_1}{\expr_2}}{\rho}}

\inferrule[T-Abs]
{
\judg{\delta_f}{\Delta,\Delta_1,\prece{\delta}{\delta_f}}{f : (\tfun{\vec{\delta_1}}{\Delta_1}{\vec{\rho_1}}{\delta_f}{\rho_2})@\delta,(\vec{x} : \vec{\rho_1}),\Gamma}\expr{\rho_2}
}
{\judg\delta\Delta\Gamma{\eclo{f}{\vec{x}}{\expr}}{(\tfun{\vec{\delta_1}}{\Delta_1}{\vec{\rho_1}}{\delta_f}\rho_2)@\delta}}

\inferrule[T-App]
{
\delta = [\vec{\delta_1'}/\vec{\delta_1}]\delta_f \\
\vec\rho_1' = [\vec{\delta_1'}/\vec{\delta_1}]\vec\rho_1 \\
\rho_2' = [\vec{\delta_1'}/\vec{\delta_1}]\rho_2 \\
\Delta_1' = [\vec{\delta_1'}/\vec{\delta_1}]\Delta_1
\\\\
\judg\delta\Delta\Gamma{e}{(\tfun{\vec{\delta_1}}{\Delta_1}{\vec{\rho_1}}{\delta_f}\rho_2)@\_} \\
\Delta \vdash \Delta_1' \\
\judg\delta\Delta\Gamma{\vec{e'}}{\vec{\rho_1'}}
}
{\judg\delta\Delta\Gamma{\erunapp{\expr}{\vec{\expr'}}}{\rho_2'}}

\inferrule[T-Par]
{ \wellkinded{\Gamma}{\varphi_1}{\kindsucc{\oldstar}} \\
\wellkinded{\Gamma}{\varphi_2}{\kindsucc{\oldstar}} \\\\
\judg\delta{\Delta}\Gamma{e_1}{(\tfun{\delta'}{\prece{\delta}{\delta'}}{\vunit}{\delta'}{\varphi_1\,\delta'})@\_} \\
\judg\delta{\Delta}\Gamma{e_2}{(\tfun{\delta'}{\prece{\delta}{\delta'}}{\vunit}{\delta'}{\varphi_2\,\delta'})@\_}}
{\judg\delta\Delta\Gamma{\epar{e_1}{e_2}}{(\varphi_1\,\delta \times \varphi_2\,\delta)@\delta}}

\inferrule[T-Fold]
{\wellkinded{\Gamma}{\trec{\tvar}{\sigma}{\delta}}{\oldstar} \\\\
\judg\delta\Delta\Gamma{\expr}{(\subst{\tvar}{\trec{\tvar}{\sigma}{\delta}}{\sigma})@\delta}}
{\judg\delta\Delta\Gamma{\efold\expr}{\trec\tvar\sigma\delta}}

\inferrule[T-Unfold]
{\wellkinded{\Gamma}{\trec{\tvar}{\sigma}{\delta}}{\oldstar} \\\\
\judg\delta\Delta\Gamma{\expr}{\trec\tvar\sigma\delta}}
{\judg\delta\Delta\Gamma{\eunfold\expr}{(\subst{\tvar}{\trec{\tvar}{\sigma}{\delta}}{\sigma})@\delta}}

\inferrule[T-TAbs]
{\judg\delta\Delta{\alpha :: \kind,\,\Gamma}{e}{\rho} \\
\verypure{\expr}
}
{\judg\delta\Delta\Gamma\expr{\tforall\alpha\kind\rho}}

\inferrule[T-TApp]
{\wellkinded{\Gamma}{\rho'}{\kind} \\
\judg\delta\Delta\Gamma{e}{\forall\alpha :: \kind.\,\rho}}
{\judg\delta\Delta\Gamma{\expr}{[\rho'/\alpha]\rho}}

\inferrule[T-GetRoot]
{\Gamma(x) = \tarray{\sigma'}@\delta' \;\lor\; \Gamma(x) = \trec{\tvar}{\sigma'}{\delta'} \\\\
\judg\delta{\Delta,\,\prece{\delta'}{\delta}}\Gamma{\expr}{\rho}}
{\judg\delta\Delta\Gamma{\expr}{\rho}}

\inferrule[T-Subtiming]
{ \judg\delta\Delta\Gamma{\expr}{\rho} \\
\subtime{\Delta}\delta{\rho}{\rho'}}
{\judg\delta\Delta\Gamma{\expr}{\rho'}}
\end{mathpar}
\captionlabel{The type system (selected rules)}{fig:typesystem}
\end{figure}

A typing environment $\Gamma$ is a map from free program variables to types, and
from free type variables to kinds.
The general form of the typing judgment of \typedis is:
\[\judg\delta\Delta\Gamma\expr\rho\]
where $\Delta$ is a logical graph, $\Gamma$ a typing environment,
$\expr$ the expression being type-checked at type $\rho$ and at current timestamp $\delta$.

Selected rules of the type system appear in \cref{fig:typesystem}.
The rules adopt Barendregt's convention~\citep{barendregt},
assuming bound variables to be distinct from already existing free variables in scope.
The reader might notice that several rules
(for example, \RULE{T-Abs} or \RULE{T-TAbs})
require the user to manually
decide where to apply these rules and with which arguments.
We leave to future work the design of syntactic features
together with a type inference mechanism
for simplifying this process.
Various rules are standard:
\RULE{T-Var} type-checks variables and
\RULE{T-Unit}, \RULE{T-Int}, and \RULE{T-Bool}
type-check base types.
The structural rules \RULE{T-Let} and \RULE{T-If} are also standard,
and type-check let bindings and if statements,
respectively.
In the remainder, we discuss the rules that deserve special attention with regard to disentanglement.

\subsection{Typing Rules for Heap Blocks}
\label{sec:heapblocks}

Heap blocks must be handled with care to guarantee disentanglement:
every time the program acquires a location---that is,
the address of a heap block---we must ensure that
this location has been allocated by a preceding task.
Otherwise, this newly created root would break the disentanglement
invariant~(\cref{sec:defdis}).
Because load operations are so common in programming languages,
we chose to enforce the following invariant on
the typing judgment~$\judg\delta\Delta\Gamma\expr\rho\,$:
every location that can be acquired from $\Gamma$
was allocated before the current timestamp $\delta$~(\cref{sec:keysubtiming}).
Hence, load operations (from immutable blocks and from mutable blocks)
do not have any timestamp check.

Operations on immutable blocks are type-checked by
\RULE{T-Pair} and \RULE{T-Proj}, for pairs, and
by \RULE{T-Inj} and \RULE{T-Case}, for sums.
In particular, \RULE{T-Pair} and \RULE{T-Inj} reflect that
pair creation and injection \emph{allocate heap blocks},
hence, the resulting type is annotated with $@\delta$,
denoting the allocating timestamp.

Operations on mutable blocks are type-checked by \RULE{T-Array},
\RULE{T-Store}, and
\RULE{T-Load}.

\subsection{Abstractions and Timestamp Polymorphism}
\label{sec:abstractions}
A function can be seen as a \emph{delayed computation}.
In our case, this notion of ``delay'' plays an interesting role:
a function can run on a task distinct from the one that allocated it.
Hence, functions in \typedis have three non-standard features related to timestamps,
roughly describing the status of the computation graph when the function will run.
First, a function takes timestamp parameters, which are universally quantified.
Second, a function takes a constraint over these timestamps, as a logical graph.
Third, a function is annotated with a timestamp
representing the task it will run on.

Let us focus on the abstraction rule \RULE{T-Abs}.
This rule type-checks a function definition of the form
$\smash{\eclo{f}{\vec{x}}{\expr}}$,
and requires the user to provide
timestamp parameters~$\svec{\delta_1}$, logical graph~$\Delta_1$,
and a running timestamp $\delta_f$.
The current timestamp is $\delta$ and
the type associated to the function
is $\smash{(\tfun{\svec{\delta_1}}{\Delta_1}{\vec{\rho_1}}{\delta_f}\rho_2)@\delta}$.
This type asserts that, if
\textit{(i)} there is some instantiation of $\smash{\vec{\delta_1}}$ satisfying $\Delta_1$,
\textit{(ii)} there are some arguments of type $\smash{\vec{\rho_1}}$, and
\textit{(iii)} the timestamp of the calling task is $\delta_f$,
then the function will produce a result of type $\rho_2$.
This type also reminds us that a function is a heap-allocated object,
and is hence annotated with the task that allocated it, here~$\delta$.
The premise of \RULE{T-Abs} changes the current timestamp to be~$\delta_f$, the timestamp of the invoking task,
and requires the body $\expr$ to be of type $\rho_2$.
\RULE{T-Abs} is in fact the sole rule of the system ``changing'' the current timestamp
while type-checking.
The logical graph is augmented with $\Delta_1$ plus the knowledge that $\delta$ precedes $\delta_f$,
conveying the fact that a function can only be called at a subsequent timestamp.
The environment $\Gamma$ is extended with the parameters $(\vec{x} : \vec{\rho_1})$ as well as
the recursive name $f$. Note that timestamp parameters~$\svec{\delta_1}$ and logical graph~$\Delta_1$
are before the arguments $\vec{x}$.
This means that the body~$\expr$ will be able to recursively call~$f$
with different timestamp arguments (potentially including a different~$\delta_f$),
for example after it forked.
We already saw such an example for the \ocaml{build} and \ocaml{selectmap} functions in \Cref{sec:backtiming,sec:keysubtiming}.

Let us now focus on \RULE{T-App}, type-checking a function application.
The conclusion type-checks
the expression $\erunapp{\expr}{\vec{\expr'}}$
to be of type~$\rho'_2$ at the current timestamp $\delta$.
The premise of \RULE{T-App} requires~$\expr$ to be a function of type
$\tfun{\svec{\delta_1}}{\Delta_1}{\vec{\rho_1}}{\delta_f}\rho_2$,
allocated by some irrelevant task.
The premise then substitutes in all the relevant parts the user-supplied
timestamps~$\svec{\delta_1'}$ in place of~$\svec{\delta_1}$.
Hence, the result type $\rho'_2$ is equal
to $[\svec{\delta_1'}/\svec{\delta_1}]\rho_2$.
In particular, the premise $\delta = [\svec{\delta_1'}/\svec{\delta_1}]\delta_f$
requires that the running timestamp $\delta_f$ to be equal to $\delta$, the current timestamp.
The premise also requires the logical graph $\Delta_1'$ to be a subgraph of
the logical graph $\Delta$, written $\Delta \vdash \Delta_1'$, meaning that every pair of vertices reachable in
$\Delta_1'$ must also be reachable in $\Delta$.
This property is formally defined in \citeappendix{sec:reach}.
Finally, the premise requires the arguments $\vec{\expr}$
to be of the correct type $\smash{\vec{\rho'_1}}$.

\subsection{The Par Rule}
\label{sec:parrule}
The typing rule for the \eparname{} primitive
is at the core of \typedis.
\RULE{T-Par} type-checks $\epar{\expr_1}{\expr_2}$
at current timestamp $\delta$.
Recall~(\cref{sec:semantics}) that the results of $\expr_1$ and $\expr_2$
must be closures; these closures are then called in parallel and their
results are returned as an immutable pair.
%
To preserve disentanglement, the two closures must not communicate
allocations they make with each other.
Hence, the premise of \RULE{T-Par} requires the two expressions $\expr_1$
and $\expr_2$ to be of type
$\tfun{\delta'}{\prece{\delta}{\delta'}}{\vunit}{\delta'}{\dots}$,
signaling that they must be closures that are expected to run on
a task $\delta'$, universally quantified, and subsequent to~$\delta$.
Because of this universal quantification over the running timestamp $\delta'$
and because
the rules allocating blocks (\RULE{T-Array}, \RULE{T-Pair}, \RULE{T-Proj} and \RULE{T-Abs})
always tag the value they allocate with the running timestamp,
the tasks will not be able to communicate allocations they make.

After these two closure calls terminate, and their underlying tasks join,
the parent task gains access to everything the two children allocated.
In fact, from the point of view of disentanglement,
we can even pretend that the parent task itself allocated these locations!
\RULE{T-Par} does more than pretending and
\textbf{backtimes} the return types of the two closures,
by substituting the
running timestamp of the children~$\delta'$ by
the running timestamp of the parent~$\delta$.
Indeed, the return types of the closures,
$\varphi_1\,\delta'$ for $\expr_1$ and $\varphi_2\,\delta'$
for $\expr_2$,
signal that these two closures will return some type,
parametrized by the running timestamp $\delta'$.
This formulation allows the rule to type-check
the original $\epar{\expr_1}{\expr_2}$
as $(\varphi_1\,\delta \times \varphi_2\,\delta)@\delta$,
that is, a pair of the two types returned by the
closures, \emph{but where the running timestamp of the child $\delta'$ was replaced
by the running timestamp of the parent $\delta$}.

\subsection{Recursive Types and Type Polymorphism}
\label{sec:evolved}

\paragraph{Recursive types}
\typedis supports iso-recursive types~\citep[\S20.2]{tapl}.
In \typedis,
a recursive type takes the form $\trec{\alpha}{\sigma}{\delta}$,
binding the recursive name $\alpha$ in the boxed type $\sigma$
which must have been allocated at $\delta$.
This syntax ensures that types are well-formed,
and forbids meaningless types $\mu \alpha.\,\alpha$ as well as
useless types $\mu\alpha.\,\mu\beta.\,\rho$.
\RULE{T-Fold} and \RULE{T-Unfold}
allow for going from $\trec\tvar\sigma\delta$ to $(\subst{\tvar}{\trec{\tvar}{\sigma}{\delta}}{\sigma})@\delta$
and vice-versa.
Note that this approach requires that the recursive occurrences of $\alpha$
are all allocated at the same timestamp;
all the nodes of the recursive data structures must have been allocated at the same timestamp.
This may seem restrictive, but subtiming will relax this requirement~(\cref{sec:subtiming}).

Let us give an example.
The type of lists allocated at timestamp $\delta$
containing integers is:
\[\trec{\alpha}{\tsum{\vunit}{\tproduct{\tinteger}{\alpha}@\delta}}{\delta}\]
This type describes that a list of integers is either
the unit value (describing the nil case),
or the pair of an integer and a list of integers (describing the cons case).

\paragraph{Type polymorphism}
\label{typepoly}
\typedis supports type polymorphism,
through type abstraction \RULE{T-TAbs} and type application \RULE{T-TApp}.
Whereas the former is standard, the latter has an unusual premise $\verypure{\expr}$,
our variant of the value restriction.

The \emph{value restriction}~\citep{wright-restriction-95}
is a simple syntactic restriction
guaranteeing soundness of polymorphism in the presence of mutable state---a
combination that is well known to be unsound if unrestricted.
In particular, the value restriction permits only \emph{values} to be polymorphic.
However, \lang has an unusual aspect: functions are \emph{not} values,
they are allocated on the heap~(\cref{sec:syntax}).
Hence, the value restriction is not applicable as-is, yet it is crucial to allow
universal type quantification in front of functions.
We contribute a variant of the value restriction,
that allows type quantification in front of any
\emph{pure expression that does not call a function, project a pair, case over a sum, or fork new tasks}.
This includes function allocation, pair allocation,
sums injection, as well as other control-flow constructs.
This syntactic check is ensured by the predicate~$\verypure{\expr}$
that appears as a premise of the type abstraction rule \RULE{T-TAbs}.
The predicate $\verypure\expr$ is defined in \citeappendix{sec:verypure}.
It can be seen as an alternative to the
solution proposed by~\citet{thesis_paulo},
in which every arrow has a \emph{purity attribute},
indicating if the function interacts with
the store.
Contrary to \citeauthor{thesis_paulo}'s solution,
we support some benign interactions with the store:
the allocation of immutable data structures.

\typedis supports higher-kind type polymorphism.
For example, reminding of the typing rule \RULE{T-Par},
one could present \eparname{} as a higher-order function of the
following type
\begin{align*}
&\textsf{par}\;:\;\,\forall(\varphi_1::\kindsucc{\basekind})\,(\varphi_2::\kindsucc{\basekind}).\\
&\quad\forall\delta\,\delta_1\,\delta_2.\, (\forall\delta'\, \prece{\delta}{\delta'}.\,\tunit \rightarrow^{\delta'} \varphi_1\,\delta')@\delta_1 \rightarrow (\forall\delta'\, \prece{\delta}{\delta'}.\,\tunit \rightarrow^{\delta'} \varphi_2\,\delta')@\delta_2 \rightarrow^\delta (\varphi_1\,\delta \times \varphi_2\,\delta)@\delta
\end{align*}
Taking $\varphi_1 = \varphi_2 = \tlambda{\delta}{\itree\delta}$ and
doing $\beta$-reduction matches the type presented in \Cref{sec:backtiming}.

\subsection{Subtiming}
\label{sec:subtiming}
\begin{figure}\small
\begin{mathpar}
\inferrule[S-Refl]{}
{\subtime{\Delta}\delta{\rho}{\rho}}

\inferrule[S-ReflAt]{}
{\subtime{\Delta}\delta{\sigma}{\sigma}}

\inferrule[S-TAbs]
{\subtime{\Delta}\delta{\rho_1}{\rho_2}}
{\subtime{\Delta}\delta{\tforall\alpha\kind{\rho_1}}{\tforall\alpha\kind{\rho_2}}}

\inferrule[S-Pair]
{\subtime{\Delta}\delta{\rho l_1}{\rho l_2} \\
\subtime{\Delta}\delta{\rho r_1}{\rho r_2}
}
{\subtime{\Delta}\delta{\tproduct{\rho l_{1}}{\rho r_{1}}}{\tproduct{\rho l_{2}}{\rho r_{2}}}}

\inferrule*[Left=S-Sum]
{\subtime{\Delta}\delta{\rho l_1}{\rho l_2} \\
\subtime{\Delta}\delta{\rho r_1}{\rho r_2}
}
{\subtime{\Delta}\delta{\tsum{\rho l_{1}}{\rho r_{1}}}{\tsum{\rho l_{2}}{\rho r_{2}}}}

\inferrule*[Left=S-At]
{\reachable\Delta{\delta_1}{\delta_2} \\
(\delta_1 \neq \delta_2 \implies \reachable\Delta{\delta_2}{\delta}) \\\\
\subtime{\Delta}{\delta_2}{\sigma_1}{\sigma_2}}
{\subtime{\Delta}\delta{\sigma_1@\delta_1}{\sigma_2@\delta_2}}

\inferrule[S-Rec]
{\reachable\Delta{\delta_1}{\delta_2} \\
(\delta_1 \neq \delta_2 \implies \reachable\Delta{\delta_2}{\delta}) \\\\
\subtime{\Delta}{\delta_2}{\sigma_1}{\sigma_2} \\
\validat\Delta\alpha{\delta_2}{\delta_2}{\sigma_2}
}
{\subtime{\Delta}\delta{\trec{\tvar}{\sigma_1}{\delta_1}}{\trec{\alpha}{\sigma_2}{\delta_2}}}

\inferrule[S-Abs]
{\Delta' = \Delta \cup \Delta_2 \\ \Delta' \vdash \Delta_1 \\\\
\subtime{\Delta'}{\delta_f}{\vec{\rho s_2}}{\vec{\rho s_1}} \\
\subtime{\Delta'}{\delta_f}{\rho_1}{\rho_2}
}
{\subtime{\Delta}{\delta}{\tfun{\vec{\delta s}}{\Delta_1}{\vec{\rho s_1}}{\delta_f}{\rho_1}}{\tfun{\vec{\delta s}}{\Delta_2}{\vec{\rho s_2}}{\delta_f}{\rho_2}}}

\inferrule*[Left=S-Inst]
{\delta_2 = [\delta_y / \delta_x]\delta_1 \\
\vec{\rho s_2} = [\delta_y / \delta_x]\vec{\rho s_1}  \\
\Delta_2 = [\delta_y / \delta_x]\Delta_1}
{\subtime{\Delta}{\delta}
{\tfun{(\vec{\delta s_l} \dapp [\delta_x] \dapp \vec{\delta s_r})}
{\Delta_1}{\vec{\rho s_1}}{\delta_1}{\rho_1}}
{\tfun{(\vec{\delta s_l} \dapp \vec{\delta s_r})}{\Delta_2}{\vec{\rho s_2}}{\delta_2}{\rho_2}}}
\end{mathpar}
\vspace{-0.8em}%
\captionlabel{The subtiming judgment}{fig:subtiming}
\vspace{-0.5em}%
\end{figure}

As presented so far, backtiming---that is, substituting the
timestamp of a child task by the one of its parent task at the join point---is
the only way of changing a timestamp inside a type~(\cref{sec:parrule}).
We propose here another mechanism that we dub \emph{subtiming}.
As the name suggests, subtiming is a form of
subtyping~\citep[\S15]{tapl} for timestamps.

At a high-level, subtiming allows for ``advancing'' a timestamp within
a type, as long as this update makes sense.
This notion of ``advancing'' relates to the notion of \emph{precedence},
describing reachability between two timestamps.
We write $\reachable{\Delta}{\delta_1}{\delta_2}$ to
describe that $\delta_1$ can reach $\delta_2$ in~$\Delta$~(\citeappendix{sec:reach}).
Equipped with this reachability predicate,
we make a first attempt at capturing the idea of subtiming as follows:
\begin{mathpar}
\inferrule[Specialized-Subtiming]
{\judg\delta\Delta\Gamma{\expr}{\sigma@\delta_1} \\ \Delta \vdash \prece{\delta_1}{\delta_2} \\ \Delta \vdash \prece{\delta_2}{\delta}}
{\judg\delta\Delta\Gamma{\expr}{\sigma@\delta_2}}
\end{mathpar}

\RULE{Specialized-Subtiming} asserts that
an expression of type $\sigma @ \delta_1$ can be viewed
as an expression of type $\sigma @ \delta_2$
as long as $\delta_1$ precedes $\delta_2$
and $\delta_2$ is not ahead of time,
that is $\delta_2$ precedes the current timestamp $\delta$.
Indeed, \typedis enforces that, if $\judg\delta\Delta\Gamma{\expr}{\sigma@\delta'}$
holds, then $\delta'$ precedes~$\delta$.

While \RULE{Specialized-Subtiming} is admissible in \typedis,
it is not general enough,
as it only considers the timestamp at the root of a type.
This motivates rule \RULE{T-Subtiming} in~\cref{fig:typesystem},
which relies on the subtiming judgment $\subtime{\Delta}{\delta}{\rho}{\rho'}$,
given in \Cref{fig:subtiming},
and acts as a subsumption rule.
Intuitively, the judgment $\subtime{\Delta}{\delta}{\rho}{\rho'}$ captures the fact
the timestamps in $\rho$ precede the timestamps in $\rho'$
under logical graph $\Delta$,
knowing that every timestamp occurring in~$\rho'$ must precede $\delta$.
The definition of the judgment now allows changing the timestamps inside immutable types.
Because of variance issues (see \citep[\S15.5]{tapl}), however,
subtiming for mutable types is only shallow:
a timestamp can be changed only at the root of an $\textsf{array}$ type.

The subtiming judgment $\subtime{\Delta}{\delta}{\rho}{\rho'}$ assumes that types are in $\beta$-normal form.
\RULE{S-Refl} and \RULE{S-ReflAt} assert that the subtiming judgment
is reflexive.
\RULE{S-TAbs} asserts that subtiming goes below type quantifiers
(which are irrelevant here, the subtiming judgment tolerates open terms).
\RULE{S-Pair} and \RULE{S-Sum} reflect that subtiming for
immutable types is deep.

\RULE{S-At} illustrates the case presented in \RULE{Specialized-Subtiming}.
This rule asserts that, with logical graph $\Delta$ and
maximum allowed timestamp $\delta$,
the boxed type $\sigma_1@\delta_1$
is a subtype of $\sigma_2@\delta_2$ if three conditions are met.
First, $\delta_1$ must precede $\delta_2$.
Second, if subtiming is applied here, that is, if $\delta_1 \neq \delta_2$,
then $\delta_2$ must precede $\delta$, the maximum timestamp allowed.
Third, $\sigma_1$ must recursively be a subtype of $\sigma_2$,
with maximum timestamp allowed $\delta_2$.
Indeed, recall that \typedis allows only for up-pointers:
every timestamp in $\sigma_2$ must precede $\delta_2$.

\RULE{S-Rec} allows subtiming for recursive types
$\trec{\alpha}{\sigma_1}{\delta_1}$ and $\trec{\alpha}{\sigma_2}{\delta_2}$.
The first three premises (in the left-to-right, top-to-bottom order)
are the same as for \RULE{S-At}.
The fourth premise $\validat\Delta\alpha{\delta_2}{\delta_2}{\sigma_2}$
requires explanations.
This predicate, dubbed the ``valid variable'' judgment
and formally defined in \citeappendix{sec:validat},
ensures two properties.
First, that $\alpha$
does not appear in an $\textsf{array}$ type (because subtiming is not allowed at this position)
or in an arrow or another recursive type (for simplicity).
Second, that if $\alpha$ appears under a timestamp $\delta$, then $\delta_2$ must
precede $\delta$.

\RULE{S-Abs} allows subtiming for function types
$\tfun{\vec{\delta s}}{\Delta_1}{\vec{\rho s_1}}{\delta_f}{\rho_1}$
and $\tfun{\vec{\delta s}}{\Delta_2}{\vec{\rho s_2}}{\delta_f}{\rho_2}$.
The quantified timestamps $\vec{\delta s}$ and
the calling timestamp $\delta_f$ must be the same.
The extended logical graph $\Delta'$, equal to $\Delta \cup \Delta_2$,
must subsume $\Delta_1$.
Moreover, the arguments $\vec{\rho s_2}$ must subtime $\vec{\rho s_1}$ (note the polarity inversion).
The return type $\rho_1$ must subtime $\rho_2$.

\RULE{S-Inst} allows for specializing a universally-quantified timestamp
and has a more standard subtyping flavor.
In this rule, the quantified timestamp $\delta_x$ is being instantiated
with $\delta_y$ (similarly to the instantiation occurring in \RULE{T-App}).

Before using subtiming, information about precedence
may be needed.
\typedis guarantees a strong invariant: every timestamp occurring in the typing environment
comes before the current timestamp.
Such an invariant is illustrated by \RULE{T-GetRoot},
which allows adding to the logical graph~$\Delta$
an edge $(\delta',\delta)$,
where $\delta'$ is a timestamp in the environment and $\delta$ the current timestamp.

\section{Soundness}
\label{sec:soundness}
In this section, we state the soundness of \typedis
and give an intuition for its proof, which takes the form
of a logical relation in Iris and is mechanized in Rocq~\citep{mechanization}.
We first enunciate the soundness theorem~(\cref{sec:soundness_statement}).
We then recall the concepts of Iris we need~(\cref{sec:irisprimer})
and present \dislog~(\cref{sec:dislog2}), the verification logic we use.
We then devote our attention to the formal proof,
by presenting the high-level ideas of the logical relation~(\cref{sec:interptypes}) we developed
and its fundamental theorem~(\cref{sec:interpjudg}).
We conclude by assembling all the building blocks we presented
and sketching the soundness proof of \typedis~(\cref{sec:proofskecth}).

\subsection{Soundness Statement of TypeDis}
\label{sec:soundness_statement}
\newcommand{\safename}{\textsf{Safe}\xspace}
\newcommand{\safe}[3]{\safename\;\phruple{#1}{#2}{#3}}
\begin{figure}\small
\begin{mathpar}
\inferrule[OOB-Alloc]
{\vint < 0}
{ \oob\store{(\ealloc{\vint}{\val})}}

\inferrule[OOB-Load]
{\store(\loc) = \vec\val \\\\ \vint < 0 \lor \vint \geq \length{\vec\val} }
{\oob\store{(\eload{\loc}{\vint})}}

\inferrule[OOB-Store]
{\store(\loc) = \vec\val \\\\ \vint < 0 \lor \vint \geq \length{\vec\val} }
{\oob\store{(\estore{\loc}{\vint}{\wal})}}

\inferrule[OOB-CAS]
{\store(\loc) = \vec\val \\\\ \vint < 0 \lor \vint \geq \length{\vec\val} }
{ \oob\store{(\ecas{\loc}{\vint}{\wal_1}{\wal_2})}}
\\
\inferrule[Red-Sched]
  {\parastep
  {\smallconfig\state\tasktree\expr}
  {\smallconfig{\state'}{\tasktree'}{\expr'}}}
  {\reducible\state\tasktree\expr}

\inferrule[Red-OOB]
{\oob\store\expr}
{\reducible{\thruple\store\amap\graph}\vertex\expr}

\inferrule[Red-Ctx]
{\reducible\state\tasktree\expr}
{\reducible\state\tasktree{\efillctx\ectx\expr}}

\inferrule[Red-Par]
{(\expr_1 \notin \Values \;\lor\; \expr_2 \notin \Values ) \\\\
(\expr_1 \notin \Values \implies \reducible\state{\tasktree_1}{\expr_1}) \\\\
(\expr_2 \notin \Values \implies \reducible\state{\tasktree_2}{\expr_2})}
{\reducible\state{\tpar\vertex{\tasktree_1}{\tasktree_2}}{\erunpar{\expr_1}{\expr_2}}}

\inferrule[Safe-Final]{}
{\safe\state\vertex\val}

\inferrule[Safe-NonFinal]
{\reducible\state\tasktree\expr}
{\safe\state\tasktree\expr}
\end{mathpar}
\captionlabel{The \oobname, \reduciblename and \safename predicates}{fig:reducible}
\end{figure}

Our soundness statement adapts \citet{milner-78}' slogan
``well-typed programs cannot go wrong''
by proving that the reduction of a well-typed program
reaches only configurations that are \emph{safe} and \emph{disentangled}.

We already formally defined the concept of disentanglement~(\cref{sec:defdis}).
What about safety?
Intuitively, a configuration
is safe if \emph{all tasks can take a step} or,
conversely, \emph{no task is stuck}.
However, this property is too strong for our type system
due to reasons unrelated to disentanglement.
Being purposefully designed for disentanglement, our type system is not capable of verifying arbitrary functional correctness conditions.
In particular, while the semantics of \lang ensures that accesses to arrays by load and store
operations are within bounds and thus cannot cause a task to get stuck, our type system does not enforce that.
This restriction comes at the advantage of freeing programmers from carrying out correctness proofs themselves,
which are carried out by the type-checker instead.
Intuitively, we say that a configuration
is safe if it is final, or
each task can either take a step
or encounters a load or a store operation
out-of-bounds.
We formalize these properties in \cref{fig:reducible}.
The property $\oob\store\expr$ asserts that the
expression $\expr$ faces an out-of-bounds operation:
that is, an allocation, a load, a store, or a CAS outside the bounds.
The property $\reducible{\state}{\tasktree}{\expr}$
asserts that, within the configuration of the
program state~$\state$, the task tree~$\tasktree$
and the expression
$\expr$,
every task of the task tree
can either take a step or faces an out-of-bounds operation.
\RULE{Red-Sched} asserts that the configuration can take a scheduling step
(that is, either a head step, a fork, or a join).
\RULE{Red-OOB} asserts that the configuration is at a leaf
and faces an out-of-bounds operation.
\RULE{Red-Ctx} asserts that an expression under evaluation
is reducible if this very expression is reducible.
\RULE{Red-Par} asserts that an active parallel pair
$\erunpar{e_1}{e_2}$ is reducible if
at least one of its components $e_i$ is not a value and
any $e_i$ that is not a value is reducible.
(If both expressions are values, a join is possible).
The property $\safe{\state}{\tasktree}{\expr}$ asserts that
the configuration $\phruple{\state}{\tasktree}{\expr}$
is either final (\RULE{Safe-Final}), that is, the task tree is at a leaf
and the expression is a value, or that every task of the task tree
can either take a step or faces an out-of-bounds operation~(\RULE{Safe-NonFinal}).

An expression~$\expr$ is \defn{\almostsafe} if
$\smash{\rtcstep{\smallconfig{\thruple\emptyset\emptyset{\singleton{(\vertex_0,\vertex_0)}}}{\tleaf\vertex}\expr}{\smallconfig{\state'}{\tasktree'}{\expr'}}}$
implies
that $\safe{\state'}{\tasktree'}{\expr'}$ and
$\disentangled{\state'}{\tasktree'}{\expr'}$ hold,
for some initial timestamp $\vertex_0$.

\begin{theorem}[Soundness of \typedis]
\label{thm:soundness}
If $\,\judg{\delta}{\emptyset}{\emptyset}{\expr}{\rho}$ then $\expr$ is \almostsafe.
\end{theorem}
\begin{proof}
We prove this theorem using a logical relation~\citep{timany-krebbers-dreyer-birkedal-24},
which makes use of \dislog, a variation of \olddislog~\citep{moine-westrick-balzer-24}.
We present the proof sketch in \Cref{sec:proofskecth}.
\end{proof}

\subsection{Iris Primer}
\label{sec:irisprimer}

We set up our proofs in Iris~\citep{iris}, and recall here
the base notations.
Iris' assertions are of type \iProp.
We write $\pre$ for an assertion,
$\pure{P}$ for an assertion of the meta-logic (that is, Rocq),
$\pre_1 \star \pre_2$ for a separating conjunction, and
$\pre_1 \wand \pre_2$ for a separating implication.
We write a postcondition---that is, a predicate over values---using $\post$.

One of the most important features of Iris are \emph{invariants}.
An invariant assertion~$\pre$, written $\iInv{\pre}$, holds true in-between every computation step.
(Formally, invariants are annotated with so-called masks~\citep[\S2.2]{iris},
we omit them for brevity.)
Invariants, as well as other logical resources in Iris,
are implemented using \emph{ghost state}.
We write $\pre_1 \vs \pre_2$ to denote a \emph{ghost update}---that is,
an update of the ghost state between $\pre_1$ and $\pre_2$.

Iris features a variety of modalities.
In this work we use two of them extensively.
First, the \emph{persistence} modality, written $\always\pre$,
asserts that the assertion $\pre$ is persistent,
meaning in particular that $\always\pre$ is duplicable.
Second, the \emph{later} modality, written $\later\pre$,
asserts that $\pre$ holds ``one step of computation later''.

We write $\loc \mapsto \vec\val$ to denote that
$\loc$ points-to an array with contents $\vec\val$.
We write $\loc \mapstopersist \block$,
with a \emph{discarded fraction}~\citep{vindum-birkedal-21},
to denote that
$\loc$ points-to an immutable block~$\block$
(that is, either a closure, an immutable pair, or an immutable sum).
This latter assertion is persistent.

\subsection{Taking Advantage of the Cyclic Approach with \dislog}
\label{sec:dislog2}
\begin{figure}\small\morespacingaroundstar
\renewcommand{\land}{\;\wedge\;}
\newcommand{\varray}[1]{\begin{array}{@{}l@{\;\;}l@{}}#1\end{array}}
\begin{mathpar}
\inferrule[D-Load]
{\pure{0 \leq \ofs < \length\wals \land
  \wals(i) = \val} \\
  \loc \mapsto_{\qp} \wals \\
  \abef{\val}{\vertex}}
{\wp
  {\vertex}
  {\eload{\loc}{\ofs}}
  {\val'}
  {\pure{\val'=\val} \star \loc \mapsto_{\qp}\wals }}

\inferrule[D-LoadOOB]
{\pure{\vint < 0 \lor \vint \geq \length{\wals}} \\
  \loc \mapsto_{\qp} \wals}
{\wp
  {\vertex}
  {\eload{\loc}{\ofs}}
  {\_}
  {\bot}}

\inferrule[D-Par]{\renewcommand{\star}{\,\ast\,}
\varray{\forall \vertex_1\,\vertex_2.\,
&\prece{\vertex}{\vertex_1} \star \prece{\vertex}{\vertex_2} \vs
\exists \post_1\,\post_2.\;\; \basewp{\vertex_1}{\ecall{\loc_1}{[\vunit]}}{\post_1} \star
\basewp{\vertex_2}{\ecall{\loc_2}{[\vunit]}}{\post_2} \,\star\\
&\quad\big(\forall\val_1\,\val_2\,\loc.\;\; \post_1\,\val_1 \star \post_2\,\val_2 \star
\prece{\vertex_1}{\vertex} \star \prece{\vertex_2}{\vertex}
\star  \loc \mapsto \eprod{\val_1}{\val_2} \wand \post\,\loc\big)
  }}
{\basewp{\vertex}{\epar{\loc_1}{\loc_2}}{\post}}

\inferrule[D-ClockMono]
{\abef{\val}{\vertex_1} \\ \prece{\vertex_1}{\vertex_2}}
{\abef{\val}{\vertex_2}}
\end{mathpar}
\captionlabel{Selected rules of \dislog}{fig:dislog}
\end{figure}

\citet{moine-westrick-balzer-24} contributed \olddislog,
the first program logic for verifying disentanglement.
\olddislog depends on the very definition of disentanglement,
and uses the standard approach presented in \Cref{sec:altapproach}:
when two tasks join,
they form a new task with a fresh timestamp.
This choice impacts the logic:
the \emph{weakest precondition} (WP) modality of \olddislog
takes the form
$\oldwp{\vertex}{\expr}{\vertex'}{\val}{\pre}$
and asserts that the expression $\expr$ running on \emph{current timestamp}~$\vertex$
is disentangled, and if the evaluation of $\expr$ terminates,
it does so on the \emph{end timestamp}
$\vertex'$,
with final value $\val$ and satisfying the assertion $\pre$.
In particular,~$\vertex$ and~$\vertex'$ may not be the same, for example
if $\expr$ contains a call to \eparname{}.

To take advantage of the
\newapproach for disentanglement (\cref{sec:altapproach}),
we had to develop a new version of \olddislog,
yielding the logic \dislog.
\dislog allows reusing the timestamp of the forking task for the child tasks upon join.
As a result, the current timestamp and end timestamp of an expression always coincide,
allowing us to simplify the WP of \olddislog
by simply removing the end timestamp parameter of the postcondition.
Formally, the WP of \dislog then takes the form
\[\wp{\vertex}{\expr}{\val}{\pre}\]
and asserts that the expression $\expr$ running on timestamp~$\vertex$
is disentangled, and if the evaluation of $\expr$ terminates,
it does so with final value $\val$ and satisfying the assertion $\pre$.
In contrast to \olddislog, \dislog tolerate out-of-bounds accesses
to cater to the \typedis type system which only enforces disentanglement.
(In practice, \dislog is parameterized by
a boolean flag which can be used to enable or disable inboundedness proof obligations;
when such obligations are enabled, \dislog has the same expressive power as \olddislog.)
Otherwise, \dislog adapts all the ideas of \olddislog.
In particular, the logic features two persistent assertions related to timestamps.
First, the clock assertion $\abef\loc\vertex$
asserts that location $\loc$ was allocated by a task that precedes $\vertex$.
Similarly, $\abef\val\vertex$ has the same meaning, if $\val$ is a location $\loc$,
or otherwise denotes $\itrue$.
Second, the precedence assertion $\prece{\vertex_1}{\vertex_2}$
asserts that task $\vertex_1$ precedes task $\vertex_2$
in the underlying computation graph.
The precedence assertion forms a pre-order: it is reflexive and transitive.
Crucially, the clock assertion
\emph{is monotonic with respect to the precedence pre-order}~\citep{moine-westrick-balzer-24}.

In the remainder of the paper, we write $\tequiv{\vertex_1}{\vertex_2}$
to denote that $\vertex_1$ and $\vertex_2$ are \emph{equivalent},
that is, both $\prece{\vertex_1}{\vertex_2}$ and $\prece{\vertex_2}{\vertex_1}$ hold.

\paragraph{Selected rules of \dislog}
\Cref{fig:dislog} presents four key rules of \dislog.
The premise of these rules are implicitly separated by a separating conjunction $\ast$.

\RULE{D-Load}, targeting a load operation on the array $\loc$ at offset $\ofs$ on task $\vertex$,
ensures disentanglement.
Indeed, the rule requires that
$\loc$ points-to the array $\wals$
and that the offset $\ofs$ in $\wals$ corresponds to the value $\val$.
It also requires the assertion $\abef\val\vertex$,
witnessing that if $\val$ is a location, then this location must have been allocated by
a preceding task.
\RULE{D-LoadOOB} is unusual for a program logic
and reflects that we purposefully allow for OOB accesses in verified programs,
because our type system does.
Because an OOB access results in a crash, the postcondition
of the WP is $\pure{False}$, allowing the user to conclude anything.
\RULE{D-Par} is at the heart of \dislog
and allows verifying a parallel call to two closures $\loc_1$ and $\loc_2$
at timestamp $\vertex$.
The premise universally quantifies over $\vertex_1$ and $\vertex_2$,
the two timestamps of the forked tasks,
that are both preceded by $\vertex$.
Then, the user must provide two postconditions, $\post_1$ and $\post_2$
for the two tasks,
and verify that the closure call $\ecall{\loc_1}{[\vunit]}$ (resp. $\ecall{\loc_2}{[\vunit]}$)
is safe at timestamp $\vertex_1$ (resp. $\vertex_2$) with postcondition $\post_1$ (resp. $\post_2$).
The second line of the premise requires the
user to prove that,
after the two tasks terminated and joined,
the initial postcondition $\post\,\loc$ must hold, for some location~$\loc$ pointing
to the pair $(\val_1,\val_2)$ where $\val_1$ is the final result of $\vertex_1$
and$\val_2$ of $\vertex_2$.
\RULE{D-ClockMono} formalizes monotonicity of the clock assertion.

\paragraph{The adequacy theorem of \dislog}
The adequacy theorem of \dislog asserts that
if $\expr$ can be verified using the program logic,
then $\expr$ is \almostsafe.

\begin{theorem}[Adequacy of \dislog]\label{thm:dislog}
If $\basewp\vertex\expr{\post}$ holds then $\expr$ is
always safe and dis-entangled.
\end{theorem}
\begin{proof}
Similar to the adequacy proof of \olddislog;
see our mechanization~\citep{mechanization}.
\end{proof}

\subsection{A Logical Relation}
\label{sec:interptypes}

The very heart of the soundness proof of \typedis is a \emph{logical relation},
set up in Iris using \dislog.
\emph{Logical relations} \citep{girard-72, PittsStarkHOOTS1998, PlotkinTR1973, StatmanARTICLE1985, TaitARTICLE1967}
are a technique that allows one to prescribe properties of valid programs in terms of their
\emph{behavior}, as opposed to solely their static properties.
We adopt the \emph{semantic approach} \citep{ConstableBook1986, LoefARTICLE1982, timany-krebbers-dreyer-birkedal-24},
which admits terms that are not necessarily (syntactically) well-typed to be an inhabitant of the logical relation
and has been successfully deployed in the RustBelt project \citep{rustbelt-18}, for example.
Our logical relation is presented in \citeappendix{sec:logrel};
we comment next on the high-level ideas.

As usual, our (unary) logical relation gives the \emph{interpretation}
of a type $\rho$ with kind $\basekind$
as a predicate on values.
Values satisfying the predicate are said to inhabit the relation.
Because our types have higher kinds,
our logical relation includes predicates on timestamps.
In particular, the interpretation of a type $\rho$ with kind $\kind$
is a function taking $\kind$ timestamp arguments
(where $\basekind$ indicates zero timestamps
and $\skindsucc{\kappa}$ indicates $\kappa + 1$ timestamps)
and producing a predicate over values.

The presented relations involve two sorts of closing substitutions for variables occurring in types.
First, a \emph{timestamp substitution}, written~$\envd$,
which is a finite map from timestamp variables~$\delta$ to concrete timestamps~$\vertex$.
Second, a \emph{type substitution}, written~$\envt$,
which is a finite map from type variables to tuples
of a kind~$\kind$ and a tuple of two functions depending on $\kind$.
The first function takes~$\kind$ timestamps and produces a predicate over values;
it represents the semantic interpretation of the type by which the variable will be instantiated.
The second function takes~$\kind$ timestamps and produces a timestamp;
its result corresponds to the root timestamp of the type by which the variable will be instantiated.
The interpretation of a type guarantees that the type
only contains \emph{up-pointers}~(\cref{sec:typedis101}), that is,
the interpretation of $\sigma@\delta$ ensures that if $\delta'$
appears in $\sigma$, then $\delta'$ precedes $\delta$.
To enforce this invariant, our approach makes use of transitivity:
the interpretation of~$\sigma@\delta$ ensures that,
for each outermost $\rho$ encountered in $\sigma$,
\emph{the root timestamp of}~$\rho$---conceptually, the outermost timestamp in $\rho$---precedes
$\delta$.
Because this invariant is enforced at each stage of the type interpretation,
and because precedence is transitive, we guarantee that there are only up-pointers.
\Citeappendix{sec:root} presents a function that computes the root timestamp of a type
and defines the assertion
$\abeftyp{\envd}{\project\envt}{\rho}{\delta}$,
asserting that the root timestamp of $\rho$ comes before $\delta$
with the mappings $\envd$ and $\envt$.

The main relation
is the
\emph{type relation}~$\rhointerp{\envd}{\envt}{\kind}{\rho}$.
It produces a predicate waiting
for $\kind$ timestamps,
a value~$\val$,
and captures that $\val$ is of type $\rho$,
given the timestamp mapping $\envd$ and type mapping $\envt$.
%

Apart from timestamps,
the seasoned reader of logical relations in Iris
will not be surprised by our approach,
as it follows the standard recipe~\citep{timany-krebbers-dreyer-birkedal-24}:
a recursive type is interpreted using a guarded fixpoint,
universal type quantification is interpreted as a universal quantification in the logic,
an array is interpreted using an invariant,
and an arrow using~WP.
Moreover, every predicate is designed such that it is persistent.

\subsection{Interpretation of Typing Judgments}
\label{sec:interpjudg}
\newcommand{\judginterp}[1]{\llbracket\,#1\,\rrbracket}

\newcommand{\projone}[1]{#1}
\newcommand{\projtwo}[1]{#1}

\renewcommand{\dom}[1]{\textsf{dom}\,#1}

\newcommand{\allproper}[1]{\textsf{all\_proper}\,#1}
\newcommand{\allregular}[1]{\textsf{all\_regular}\,#1}
\newcommand{\envv}{u}

\begin{figure}\small\morespacingaroundstar
\[\begin{array}{r@{\;\;}c@{\;\;}l}
\judginterp{\judg\delta\Delta\Gamma{\expr}{\rho}} &\eqdef& \always\forall \envd\,\envt\,\envv.\; \pure{\dom\Gamma = \dom\envv} \wand\\
&& \pure{\forall \alpha\,\kappa\,\Psi\,\rfunc.\; \envt(\alpha) = (\kappa,(\Psi,\rfunc)) \implies \proper{\kappa}{\Psi} \;\wedge\; \regular{\kappa}{\rfunc}}\wand \\
&& \bigast{(\vertex_1,\vertex_2) \in \Delta}{\prece{\envd(\vertex_1)}{\envd(\vertex_2)}} \wand\\
&& \bigast{(\var,\rho) \in \Gamma,\, (\var,\val) \in \envv}{(\abeftyp{\envd}{\project\envt}{\rho}{\delta} \star \rhointerp{\envd}{\envt}\basekind{\rho}\,\val)} \wand\\
&& \forall\vertex.\; \tequiv{\vertex}{\envd(\delta)} \;\wand\; \wp{\vertex}{[\envv/]\expr}{\val}{\abeftyp{\envd}{\project\envt}{\rho}{\delta} \star \rhointerp{\envd}{\envt}\basekind{\rho}\,\val}
\end{array}\]
\captionlabel{The interpretation of typing judgments}{fig:interpjudg}
\end{figure}

We now focus on the interpretation
of the \typedis typing judgment,
paving our way to state the fundamental theorem of the logical relation.
\Cref{fig:interpjudg} gives its interpretation,
appealing to the WP of \dislog.
For the judgment with logical graph~$\Delta$,
type environment $\Gamma$, and expression $\expr$ with type $\rho$
at timestamp~$\delta$,
the interpretation starts by quantifying over three closing substitutions:
the timestamp mapping~$\envd$, the type mapping~$\envt$,
as well as a \emph{variable mapping}~$\envv$,
a map from term variables to values.
The variable mapping must have the same domain as the environment $\Gamma$.
The type mapping~$\envt$ is restricted such that
type variables are given only a proper interpretation
(via the $\proper{\kappa}{\Psi}$ property)
and a regular root function (via the $\regular{\kappa}{\rfunc}$ property).
The property~$\proper{\kind}{\Psi$} captures that
any timestamp parameter of $\Psi$ can be replaced by an equivalent one.
The property~$\regular{\kind}{\rfunc}$ captures that
the function $\rfunc$ either ignores all its arguments or returns one of them.
These two properties are needed in order to prove
the correctness of \RULE{T-Par}.

Then, the interpretation requires that $\Delta$ is a valid logical graph,
that is, each edge between two timestamp variables in $\Delta$
corresponds to an edge between their mapping.
The interpretation also requires that,
for every variable $x$ that has type $\rho$ in $\Gamma$ and is associated
to value $\val$ in $\envv$, the root timestamp of $\val$ precedes the interpretation of $\delta$
and $\val$ inhabits the interpretation of $\rho$.
Next, the definition quantifies over a timestamp $\vertex$,
equivalent to the interpretation of $\delta$,
and asserts WP
at timestamp $\vertex$ of the expression $\expr$
in which variables are substituted by values following the variable mapping $\envv$.
The postcondition asserts that
the root timestamp of $\val$ precedes the interpretation of $\delta$
and that the returned value $\val$ inhabits the interpretation of type $\rho$.

Having the interpretation of typing judgment
defined, we can state the \emph{fundamental theorem},
ensuring that syntactically well-typed terms~(\cref{sec:typesystem}) inhabit the logical relation.

\begin{theorem}[Fundamental]\label{thm:fundamental}
If $\,\judg\delta\Delta\Gamma{\expr}{\rho}$ holds,
then $\judginterp{\judg\delta\Delta\Gamma{\expr}{\rho}}$ holds too.
\end{theorem}
\begin{proof}
By induction over the typing derivation;
see our mechanization~\citep{mechanization}.
\end{proof}

\subsection{Putting Pieces Together: The Soundness Proof of \typedis}
\label{sec:proofskecth}
We can finally unveil the proof of the soundness
\cref{thm:soundness} of \typedis, which we formally establish in Rocq.
Let us suppose
that $\judg{\delta}{\emptyset}{\emptyset}{\expr}{\rho}$ holds.
Making use of the fundamental \cref{thm:fundamental},
we deduce that $\judginterp{\judg\delta\emptyset\emptyset{\expr}{\rho}}$ holds
too.
Unfolding the definition~(\cref{fig:interpjudg}),
instantiating the timestamp mapping~$\envd$ with the singleton
map $\singletonmap\delta{\vertex_0}$---for some initial timestamp $\vertex_0$---and
the type mapping~$\envt$ and
the variable mapping~$\envv$ with empty maps,
and simplifying trivial premises concerning these mappings,
we are left with the statement
$\forall\vertex.\; \tequiv{\vertex}{\vertex_0} \;\wand\; \wp{\vertex}{\expr}{\val}{\abeftyp{\envd}{\project\envt}{\rho}{\delta} \star \rhointerp{\singletonmap\delta{\vertex_0}}{\emptyset}\basekind{\rho}\,\val}$.
Instantiating $\vertex$ with $\vertex_0$,
we deduce that $\wp{\vertex_0}{\expr}{\val}{\abeftyp{\envd}{\project\envt}{\rho}{\delta} \star \rhointerp{\singletonmap\delta{\vertex_0}}{\emptyset}\basekind{\rho}\,\val}$ holds.
We finally use the adequacy \cref{thm:dislog} of \dislog
and deduce that $\expr$ is \almostsafe.

\section{Case Studies}
\label{sec:evaluation}
We evaluate the usefulness of \typedis by type-checking
several case studies in Rocq using the rules presented in \Cref{sec:typesystem}.

We verify the examples presented in the ``Key Ideas'' \Cref{sec:key_ideas}.
These examples illustrate: simple mechanics of the type system~(\cref{sec:typedis101}),
backtiming~(\cref{sec:backtiming}), and subtiming~(\cref{sec:keysubtiming}).
In particular, the last two examples,
\ocaml{build} and \ocaml{selectmap},
illustrate the use of \typedis
with higher-order functions and a recursive immutable type (a binary tree with integer leaves).
For more details than the intuitions we already gave for these examples,
we refer the reader
to our formalization~\citep{mechanization}.

Our largest case study consists
of the typing of a parallel deduplication algorithm via concurrent hashing.
This example is a case study of DisLog~\citep[\S6.3]{moine-westrick-balzer-24}.
Deduplication consists of removing duplicates
from an array---something that can
be done efficiently in a parallel, disentangled setting~\citep{westrick-thesis-2022}.
The algorithm relies on a folklore~\citep{verify-this-22} concurrent, lock-free,
fixed-capacity hash set using \emph{open addressing} and
\emph{linear probing} to handle collisions~\citep{knuth-98}.
The main deduplication function allocates a new hash set,
inserts in parallel every element into the hash set
using a parallel for loop,
and finally returns the elements of the set.
We first comment on the parallel for loop~(\cref{sec:parfor})
and then on the main deduplication algorithm~(\cref{sec:dedup}).

\subsection{The Parallel For Loop}
\label{sec:parfor}
\begin{figure}
\centering\small\morespacingaroundstar
\begin{align*}
&\parforname\;\eqdef\;\mu\,\selfname.\,\lambda [\lowbound;\highbound;\parforarg].\\
&\quad\kw{if}\; \highbound \leq \lowbound \;\kw{then}\;\vunit\;\kw{else\;if}\;\eeq\diffname\one\;\kw{then}\;\parforarg\;[\lowbound]\\
&\quad\kw{else}\;\kw{let}\;\midname\;=\; \lowbound + (\diffname/2)\;\kw{in}\;\epar{\selfname\;[\lowbound;\midname;\parforarg]}{\selfname\;[\midname;\highbound;\parforarg]}\\[0.5em]
&\parforname\;:\;\forall\delta\,\delta_\parforarg.\;\tinteger \rightarrow \tinteger \rightarrow (\forall\delta'\,\prece{\delta}{\delta'}.\, \tinteger \rightarrow^{\delta'} \tunit)@\delta_f \rightarrow^{\delta} \tunit
\end{align*}
\vspace{-1em}
\captionlabel{Implementation and type of the parfor primitive}{fig:parfor}
\end{figure}

Our implementation of the parallel for loop appears
in the upper part of \cref{fig:parfor}
and is a direct translation of MaPLe standard library's implementation~\citep{mpl-2020},
The function $\parforname$ takes three arguments:
a lower bound $\lowbound$, a higher bound $\highbound$,
and a closure $\parforarg$ to execute at each index between these bounds.
The function $\parforname$ is defined recursively:
it returns immediately if $\highbound \leq \lowbound$,
executes the closure $\parforarg\,[\lowbound]$ if $\highbound - \lowbound = 1$,
and otherwise calls itself recursively \emph{in parallel},
splitting the range in two.
The type we give to $\parforname$ appears
in the lower part of \cref{fig:parfor},
and is as one could expect.
Indeed, the type quantifies over two timestamps $\delta$,
at which $\parforname$ will be called, and $\delta_\parforarg$,
the (irrelevant) timestamp of the closure.
The function then requires two integers, and a closure
that will be called at some timestamp $\delta'$ that succeeds $\delta$.

The type-checking of $\parforname$ is non-trivial because it
involves \emph{polymorphic recursion}.
Indeed, $\parforname$'s type universally quantifies
over the calling timestamp, but calls itself recursively
after a par---that is, at another (subsequent) timestamp.
\typedis supports natively such a pattern thanks to \RULE{T-Abs}.
Interestingly, polymorphic recursion introduce a need for subtiming.
Indeed, while type-checking $\parforname$'s body at current timestamp $\delta$,
the closure $\parforarg$
has type $(\forall\delta'\,\prece{\delta}{\delta'}.\, \tinteger \rightarrow^{\delta'} \tunit)$.
However, the recursive call happens after a par, hence at a new
current timestamp $\delta_1$ such that $\prece{\delta}{\delta_1}$.
But in order to type-check the recursive call,
the user has to give to $\parforarg$ the type $(\forall\delta'\,\prece{\delta_1}{\delta'}.\, \tinteger \rightarrow^{\delta'} \tunit)$---notice
the difference between the precedence information on $\delta'$.
This is a typical use of subtiming, and because $\prece{\delta}{\delta_1}$,
we conclude using \RULE{S-Abs}.

\subsection{Internals of the Deduplication Case Study}
\label{sec:dedup}
\begin{figure}\centering\small\morespacingaroundstar
\begin{minipage}{0.45\textwidth}
\begin{align*}
&\haddname\eqdefspace \lambda[\hashname;\htable;\emptyelem;\helem].\\
&\quad\letinline{put}{\mu\,\selfname.\lambda[\ofs].}\\
&\qquad\kw{if}\;(\ecas\htable\ofs\emptyelem\helem\;\symbor\; \eeq{\eload\htable\ofs}\helem) \;\kw{then}\;\vunit\\
&\qquad\kw{else}\;\ecall\selfname{[(\ofs + 1) \bmod \htsize]}\;\kw{in}\\
&\quad\ecall{put}{[\ecall\hashname{[\helem]}\bmod\htsize]}
\end{align*}
\end{minipage}
\begin{minipage}{0.45\textwidth}
\begin{align*}
&\dedupname\eqdefspace \lambda[\hashname;\emptyelem;\loc].\\
&\quad\letinline\htable{\ealloc\htsize\emptyelem}\;\kw{in}\\
&\quad\letinline\taskname{\lambda[\ofs].\;\hadd\hashname\htable\emptyelem{\eload\loc\ofs}}\;\kw{in}\\
&\quad\parfor\zero{\elength\loc}\taskname\;\sequence\\
&\quad\filtercompact\htable\emptyelem
\end{align*}
\end{minipage}
\vspace{0.5em}
\[\begin{array}{r@{\;\;:\;\;}l}
\haddname & \forall\alpha::\basekind.\,
\forall\delta\,\delta_1\,\delta_2.\,(\alpha \rightarrow^\delta \tinteger)@\delta_1 \rightarrow \tarray{\alpha}@\delta_2  \rightarrow \alpha \rightarrow \alpha \rightarrow^\delta \tunit \\
\dedupname & \forall\alpha::\basekind.\,\forall\delta\,\delta_1\,\delta_2.\,(\forall\delta'\,\prece{\delta}{\delta'}.\, \alpha \rightarrow^{\delta'} \tinteger)@\delta_1 \rightarrow \alpha \rightarrow \tarray{\alpha}@\delta_2 \rightarrow^\delta \tarray{\alpha}@\delta
\end{array}\]
\captionlabel{Case study: deduplication of an array by concurrent hashing}{fig:dedup}
\end{figure}

Let us now focus on the code for our deduplication
algorithm, which appears in the upper part of \cref{fig:dedup}.
This code assumes a maximum size $\htsize$ for the underlying hash set.
The function $\dedupname$ takes three arguments:
a hashing function $\hashname$,
a dummy element $\emptyelem$ in order to populate the result array,
and the array to deduplicate~$\loc$.
The function first allocates the hash set~$\htable$
and then calls in parallel the $\haddname$ function for every index in~$\loc$.
The function $\haddname$ consists of a CAS loop, that tries
to insert the element in the first available slot.
Finally, $\dedupname$ filters the remaining dummy elements~$\emptyelem$
using an omitted function $\filtercompactname$.
Because it involves no fork or join,
the function $\haddname$ admits a simple polymorphic type,
shown in the lower part of \cref{fig:dedup},
quantifying over the type~$\alpha$ of the elements
of the array to deduplicate,
the calling timestamp $\delta$ and
two timestamps $\delta_1$ and $\delta_2$.
The first argument is a closure of a hashing function on $\alpha$
that will be called at timestamp $\delta$.
The second argument is the hash set, a $\tarray{\alpha}$.
The third and fourth arguments are of type $\alpha$
and correspond to the dummy element and the element to insert,
respectively.

Using this type for $\haddname$,
we are able to type-check $\dedupname$
with the type shown in the lower part of \cref{fig:dedup}.
This type quantifies again over the type $\alpha$
of the elements of the array to deduplicate,
and then quantifies over $\delta$, the calling timestamp, and $\delta_1$ and $\delta_2$, the (irrelevant) allocation timestamps
of the first and third argument, respectively.
The first argument is a closure of a hashing function on $\alpha$,
that will be called at subsequent tasks~$\delta'$.
The second argument is a dummy element.
The third argument is the array to deduplicate.
Again, type-checking $\dedupname$ requires subtiming:
the $\haddname$ function expects
a hashing function at its calling timestamp~$\delta$,
whereas the supplied $\hashname$ is more general,
because it is polymorphic with respect to its calling timestamp.
We use subtiming (\RULE{S-Abs} and \RULE{S-Inst}) to convert the latter into the former.

\section{Related Work}
\label{sec:related_work}
\paragraph{Disentanglement}
%
%
The specific property we consider in this paper is based on the definition
by \citet{westrick-et-al-20} which was later formalized by
\citet{moine-westrick-balzer-24}.
Most of the existing work on disentanglement considers structured fork-join
parallel code, as we do in this paper.
More recently, \citet{arora-muller-acar-24} showed that disentanglement is
applicable in a more general setting involving
parallel futures, and specifically prove deadlock-freedom in this setting.
%
%
We plan to investigate
whether \typedis could be extended to
support futures.

\paragraph{Verification of Disentanglement}
Two approaches to check for and/or verify disentanglement have been proposed
prior to \typedis.
First, as currently implemented in the MaPLe compiler, the programmer can rely
on a runtime \emph{entanglement detector}~\cite{westrick-arora-acar-22}.
This approach is similar in principle to dynamic race
detection~\cite{FlanaganFreundPLDI2009}.
In the case of entanglement, dynamic detection has been shown to have
low overhead, making it suitable for automatic
run-time management of entanglement~\cite{arora-westrick-acar-23}.
However, run-time detection cannot guarantee disentanglement due to the inherent
non-determinism of entanglement, which typically arises due to race
conditions and may or may not occur in individual executions.
The second approach, as developed by \citet{moine-westrick-balzer-24}, is
full-blown static verification of disentanglement using a separation logic
called DisLog, proven sound in Rocq.
This approach can be used to statically verify disentanglement
for a wide variety of programs---for example,
even for non-deterministic programs that utilize
intricate lock-free data structures in shared memory.
However, static verification with DisLog is difficult, requiring significant
effort even to verify small examples.
%

\paragraph{Region-based Systems}
\typedis associates timestamp variables with values in their types.
%
%
Immediately, we note similarities with region-based type and effect
systems~\cite{tofte-talpin-97,tofte-retro-04,grossman-et-al-02}
which have also recently received attention in supporting
parallelism~\cite{elsman-henriksen-23}.
The timestamps in our setting are somewhat analogous to regions, with
parent-child relationships between timestamps and the up-pointer invariant
of \typedis bearing resemblance to the stack discipline
of region-based memory management systems.
%
%
However, there are a number of key differences.
In region-based systems, allocations may occur within any region, and
all values within a region are all deallocated at the same moment; one
challenge in such systems is statically predicting or conservatively bounding
the lifetime of every value.
In contrast, in \typedis, allocations only ever occur at the ``current''
timestamp, and timestamps tell you nothing about
deallocation---every value in our approach is dynamically garbage collected.
Each timestamp in \typedis is associated with a task within a
nested fork-join task structure, and values with the same timestamp are
all allocated by the same task (or one of its subtasks).

\paragraph{Possible Worlds Type Systems}
Our type system falls into what can broadly be categorized as a \emph{possible worlds} type system.
These type systems augment the typing judgment with world modalities
(in our case timestamp variables~$\delta$)
that occur as syntactic objects in propositions (a.k.a. types),
and typing is then carried out relative to an accessibility relation
(in our case the logical graph~$\Delta$).
While our work is the first to contribute a possible worlds type system for disentanglement,
world modalities have been successfully used for other purposes.
In the context of fork-join parallelism,
\citet{MullerPLDI2017} employed world modalities to track priorities of tasks
and guarantee absence of priority inversions, ensuring responsiveness and interactivity.
While \citet{MullerPLDI2017} also require their priorities to be partially ordered,
as we require timestamps to be partially ordered,
their priorities are fixed, whereas ours are not.
In the context of message-passing concurrency,
world modalities have been employed to verify
deadlock-freedom \citep{BalzerESOP2019},
domain accessibility \citep{CairesCONCUR2019}, and
information flow control \citep{DerakhshanLICS2021,DerakhshanECOOP2024}.
This line of work not only differs in underlying computation model,
considering a process calculus,
but also adopts linear typing to control data races and non-determinism.
While disentanglement does not forbid races,
adopting some form of linear typing may be an interesting avenue for future work,
to admit even more disentangled programs as well-typed,
e.g. those with down-pointers.

\paragraph{Information Flow Control Type Systems}

Information flow type systems
\citep{VolpanoARTICLE1996,SmithVolpanoPOPL1998,SabelfeldIEE2003}
can also be viewed as representatives of possible worlds type systems,
where modalities capture confidentiality (or integrity) and pc labels, and the accessibility relation is a lattice.
Typically, modalities can change by typing.
For example, when type-checking the branches of an if statement
the pc label is increased to the join of its current value and
the confidentiality label of the branching condition.
A similar phenomenon happens in \typedis upon type-checking a fork,
where the sibling threads are type-checked at a later timestamp.
Besides these similarities in techniques employed,
the fundamental invariants preserved by type-checking are different.
In our setting it is the ``no cross-pointers invariant'',
whereas it is noninterference for IFC type systems.
As a result, the metatheory employed also differs:
whereas we use a unary logical relation,
noninterference demands a binary logical relation.
Such a binary logical relation for termination-insensitive noninterference
in the context of a sequential, higher-order language with higher-order store,
for example, has been contributed by \citet{OddershedeGregersenPOPL2021}.
The authors develop an IFC type system, in the spirit of
Flow Caml~\citep{SimonetWorkshop2003,PottierSimonetARTICLE2003},
with label polymorphism, akin to our timestamp polymorphism.
Like our work, the authors use the semantic typing approach
supported by the Iris separation logic framework.
Similarly, the authors support subtyping on labels, allowing a label to be raised in accordance with the lattice,
akin to our subtiming, in accordance with the precedence relation.

\paragraph{Type Systems for Parallelism and Concurrency}
There has been significant work on developing static techniques, especially
type systems, to guarantee correctness and safety properties
(such as race-freedom, deadlock-freedom, determinism, etc.) for
parallel and concurrent programs.
%
For example, the idea of ownership
\citep{NobleECOOP1998,ClarkeOOPSLA1998,MullerPhD2002,DietlMullerARTICLE2005}
has been exploited to rule out races and deadlocks among threads
\citep{BoyapatiRinardOOPSLA2001,BoyapatiOOPSLA2002,BoyapatiPOPL2003}.
Ownership is also enforced by linear type systems \citep{WadlerIFIP1990},
which rule out races by construction and
have been successfully employed in message-passing concurrency \citep{CairesCONCUR2019,WadlerICFP2012}.
The approach has then been popularized by Rust~\cite{klabnik2023rust},
in particular, focusing on statically restricting aliasing and
mutability~\cite{rustbelt-18}, which in Rust takes the form of
ownership and borrowing as well as reference-counted mutexes for maximal
flexibility.
Recently flexible mode-based systems have been explored,
as present in the work on DRFCaml~\cite{georges-et-al-25}, which exploits modes
(extending \citet{lorenzen-et-al-24}) to distinguish
values that can and cannot be safely shared
between threads.
Other systems leverage region-based techniques to restrict concurrent
threads, ensuring safe disjoint access to the
heap with minimal annotations~\cite{DBLP:conf/pldi/MilanoTM22}, or
leveraging explicit annotations to limit the set of permissible effects on
shared parts of the heap~\cite{bocchino-et-al-09}.
%

Much of these related works focus on the hazards of
concurrency: data races, race conditions,
non-determinism, and similar issues.
Disentanglement (and by extension, \typedis) focuses on an equally important
but different issue, namely, the \emph{performance} of parallel programs.
\typedis in particular is designed to allow for
unrestricted sharing of immutable data
(as demonstrated in \Cref{sec:keysubtiming})
mixed with disentangled sharing of mutable
data (for example, in \Cref{sec:evaluation}).
This support for data sharing is motivated by the implementation of efficient
parallel algorithms, many of which rely upon access to shared memory with
irregular and/or data-dependent access patterns,
which are difficult to statically analyze for safety.
For example, \citet{DBLP:conf/spaa/0002PZWJ24} find that many standard
implementations of parallel algorithms are rejected by the Rust type system,
yet these same implementations have
been shown to be
disentangled~\cite{westrick-arora-acar-22}.
We consider one such implementation as a case study in
\Cref{sec:evaluation} and confirm that it is typeable under \typedis.

\section{Conclusion and Future Work}
\label{sec:conclusion}
Disentanglement is an important property of parallel programs,
which can in particular serve for improving performance.
This paper introduces \typedis,
a static type system that
proves disentanglement.
\typedis annotates types with timestamps,
recording for each object the task that allocated it.
Moreover, \typedis supports iso-recursive types, as well as type and timestamp
polymorphism.
\typedis allows restamping the timestamps in types
using a particular form of subtyping we dub \emph{subtiming}.

This paper focuses on \emph{type-checking}, that is,
given a program annotated with types, checking if these types are valid.
We are currently working on a prototype type-checker, written in OCaml.
An immediate direction for future work is \emph{type inference},
that is, generating a valid type for a program.
For future work, we plan to use the framework of \citet{odersky-sulzmann-wehr-99},
which adapts Hindley-Milner to a system with constrained universal
quantification.
We believe subtiming and backtiming will be inferrable.
One challenging case will be mixing polymorphic recursion with \eparname{},
which might require annotations in order to remain decidable
(this is a known problem in region-based type systems~\cite{DBLP:journals/toplas/TofteB98}).

\begin{acks}
  We thank Umut A. Acar for sharing his insights during early design
  discussions and helping us shape the context for this work.
  We also thank Kashish Raimalani for reviewing an
  initial draft, and we thank the anonymous reviewers for their
  helpful comments.
\end{acks}

\ifappendix
\appendix
\section{Appendix}
\crefalias{section}{appendix}

\subsection{The Kinding Judgment}
\label[appendix]{sec:wellkinded}
\begin{figure}[!h]\small
\begin{mathpar}
\inferrule[K-Var]
{x :: \kind \in \Gamma}
{\wellkinded{\Gamma}{x}{\kind}}

\inferrule[K-Unboxed]{}
{\wellkinded{\Gamma}{\tau}{\oldstar} }

\inferrule[K-At]
{\wellkinded{\Gamma}{\sigma}{\oldstar}}
{\wellkinded{\Gamma}{\sigma@\delta}{\oldstar}}

\inferrule[K-Lam]
{\wellkinded{\Gamma}{\rho}{\kind}}
{\wellkinded{\Gamma}{\tlambda{\delta}{\rho}}{\kindsucc\kind}}

\inferrule[K-App]
{\wellkinded{\Gamma}{\rho}{\kindsucc\kind}}
{\wellkinded{\Gamma}{\tapp\rho\delta}{\kind}}

\inferrule[K-TAbs]
{\wellkinded{\tvar :: \kind,\Gamma}{\rho}{\oldstar}}
{\wellkinded{\Gamma}{\tforall{\tvar}{\kind}{\rho}}{\oldstar}}

\inferrule[K-Rec]
{\wellkinded{\tvar :: \oldstar,\Gamma}{\sigma}{\oldstar}}
{\wellkinded{\Gamma}{\trec{\tvar}{\sigma}\delta}{\oldstar}}

\inferrule[K-Array]
{\wellkinded{\Gamma}{\rho}{\oldstar}}
{\wellkinded{\Gamma}{\tarray{\rho}}{\oldstar}}

\inferrule[K-Pair]
{\wellkinded{\Gamma}{\rho_1}{\oldstar} \\
\wellkinded{\Gamma}{\rho_2}{\oldstar}
}
{\wellkinded{\Gamma}{\tproduct{\rho_1}{\rho_2}}{\oldstar}}

\inferrule[K-Sum]
{\wellkinded{\Gamma}{\rho_1}{\oldstar} \\
\wellkinded{\Gamma}{\rho_2}{\oldstar}
}
{\wellkinded{\Gamma}{\tsum{\rho_1}{\rho_2}}{\oldstar}}

\inferrule[K-Arrow]
{( \forall \rho.\;\rho \in \vec{\rho_1} \implies \wellkinded{\Gamma}{\rho}{\oldstar}) \\
\wellkinded{\Gamma}{\rho_2}{\oldstar}
}
{\wellkinded{\Gamma}{\tfun{\vec{\delta_1}}{\Delta}{\vec{\rho_1}}{\delta_2}{\rho_2}}{\oldstar}}

\end{mathpar}
\captionlabel{Kinding judgment}{fig:wellkinded}
\end{figure}

\Cref{fig:wellkinded} presents the kinding judgment
$\wellkinded{\Gamma}{\rho}{\kind}$,
asserting that type $\rho$ has kind $\kind$
considering the environment $\Gamma$.

\subsection{The veryPure Predicate}
\label[appendix]{sec:verypure}
\begin{figure}[!h]\small
\begin{mathpar}
\inferrule[VP-Val]{}
{\verypure{\val}}

\inferrule[VP-Abs]{}
{\verypure{(\eclo{f}{\vec{x}}{\expr})}}

\inferrule[VP-Var]{}
{\verypure{\var}}

\inferrule[VP-Prim]
{\verypure{{\expr_1}} \\ \verypure{\expr_2}}
{\ecallprim{\expr_1}{\expr_2}}

\inferrule[VP-Let]
{\verypure{{\expr_1}} \\ \verypure{\expr_2}}
{\elet{\var}{\expr_1}{\expr_2}}

\inferrule[VP-Pair]
{\verypure{{\expr_1}} \\ \verypure{\expr_2}}
{\eprod{\expr_1}{\expr_2}}

%
\inferrule[VP-Fold]
{\verypure{\expr}}
{\verypure{\efold\expr}}

\inferrule[VP-Fold]
{\verypure{\expr}}
{\verypure{\eunfold\expr}}

\inferrule[VP-If]
{\verypure{{\expr_1}} \\ \verypure{\expr_2} \\ \verypure{\expr_3} }
{\eif{\expr_1}{\expr_2}{\expr_3}}

\end{mathpar}
\captionlabel{The veryPure predicate}{fig:verypure}
\end{figure}

\Cref{fig:verypure} presents the \verypurename predicate over an expression.
This predicate ensures that the expression does not contain
any array allocation, load, store, par, projection, case, or function call.

\subsection{Reachability Predicates}
\label[appendix]{sec:reach}
\begin{figure}[!h]\small
\begin{mathpar}
\inferrule[R-Refl]{}
{\reachable{\Delta}{\delta}{\delta}}

\inferrule[R-Cons]
{(\delta_1,\delta_2) \in \Delta \\ \reachable{\Delta}{\delta_2}{\delta_3}}
{\reachable{\Delta}{\delta_1}{\delta_3}}

\inferrule[R-Logical]
{ \forall\delta_1\,\delta_2.\; (\delta_1,\delta_2) \in \Delta' \implies \reachable{\Delta}{\delta_1}{\delta_2}}
{\Delta \vdash \Delta'}
\end{mathpar}
\captionlabel{The reachability predicates}{fig:reach}
\end{figure}

\Cref{fig:reach} presents the reachability predicates that appear
in \RULE{T-App} and in \Cref{fig:subtiming}.

\RULE{R-Refl} asserts that a timestamp can always reach itself.
\RULE{R-Cons} asserts that
if there is an edge between $\delta_1$ and $\delta_2$
and if $\delta_2$ can reach $\delta_3$,
then $\delta_1$ can reach $\delta_3$.

\RULE{R-Logical} asserts that a logical graph $\Delta$ subsumes
a logical graph $\Delta'$ if every edge between $\delta_1$ and $\delta_2$ in
$\Delta'$ can be simulated in $\Delta$.

\subsection{The ``Valid Variable'' Judgment}
\label[appendix]{sec:validat}
\begin{figure}[!h]\small
\begin{mathpar}
\inferrule[VA-Var]
{\alpha=\alpha' \implies \reachable\Delta{\delta_1}{\delta_2}}
{\validat\Delta\alpha{\delta_1}{\delta_2}{\alpha'}}

\inferrule[VA-Base]{}
{\validat\Delta\alpha{\delta_1}{\delta_2}{\tau}}

\inferrule[VA-At]
{\validat\Delta\alpha{\delta_1}{\delta}{\sigma}}
{\validat\Delta\alpha{\delta_1}{\delta_2}{\sigma@\delta}}

\inferrule[VA-TAbs]
{\alpha\neq\alpha' \implies \validat\Delta\alpha{\delta_1}{\delta_2}{\rho}}
{\validat\Delta\alpha{\delta_1}{\delta_2}{\tforall{\alpha'}{\kind}{\rho}}}

\inferrule[VA-Pair]
{\validat\Delta\alpha{\delta_1}{\delta_2}{\rho_1} \\
\validat\Delta\alpha{\delta_1}{\delta_2}{\rho_2}
}
{\validat\Delta\alpha{\delta_1}{\delta_2}{\tproduct{\rho_1}{\rho_2}}}

\inferrule[VA-Sum]
{\validat\Delta\alpha{\delta_1}{\delta_2}{\rho_1} \\
\validat\Delta\alpha{\delta_1}{\delta_2}{\rho_2}
}
{\validat\Delta\alpha{\delta_1}{\delta_2}{\tsum{\rho_1}{\rho_2}}}

\inferrule[VA-TRec]
{\alpha \notin \fvtyp\rho \setminus \singleton{\alpha'}}
{\validat\Delta\alpha{\delta_1}{\delta_2}{\trec{\alpha'}{\sigma}{\delta}}}

\inferrule[VA-Array]
{ \alpha \notin \fvtyp{\rho} }
{\validat\Delta\alpha{\delta_1}{\delta_2}{\tarray{\rho}}}

\inferrule[VA-Abs]
{\alpha \notin \fvtyp{\vec{\rho'}} \cup \fvtyp{\rho''}}
{\validat\Delta\alpha{\delta_1}{\delta_2}{\tfun{\vec{\delta'}}{\Delta'}{\vec{\rho'}}{\delta_f}{\rho''}}}
\end{mathpar}
\captionlabel{The ``valid variable'' judgment}{fig:validat}
\end{figure}

\Cref{fig:validat} presents the ``valid variable'' judgment
that is used for subtiming recursive types~(\RULE{S-Rec}).

\subsection{The Root Functions and Assertions}
\label[appendix]{sec:root}
\begin{figure}[!h]\small
\[\begin{array}{r@{\;\;}c@{\;\;}l}
\answer&\eqdef&\correct\delta \mid \unboxed \mid \nonsense\\[1em]

\troot\envd\envt\kind\rho&:&\ikind{\kind}{\answer}\\
\troot\envd\envt\kind{\tvar}&\eqdef&\eif{\envt(\tvar)= (\kind, (\_,\rfunc))}{\rfunc}{\nonsense_\kind}\\
\troot\envd\envt\basekind\tau&\eqdef&\unboxed\\
\troot\envd\envt\basekind{(\sigma@\delta)}&\eqdef&\correct{\delta}\\
\troot\envd\envt{(\skindsucc\kind)}{(\tlambda{\delta}{\rho})}&\eqdef&\lambda\vertex.\,\troot{(\blockupd\envd\delta\vertex)}\envt\kind\rho\\
\troot\envd\envt{\kind}{(\tapp{\rho}{\delta})}&\eqdef&(\troot\envd\envt{(\skindsucc\kind)}\rho)\,(\envd(\delta))\\
\troot\envd\envt\basekind{(\tforall{\tvar}{\kind}{\rho})}&\eqdef&\troot\envd{(\blockupd\envt\tvar{\nonsense_\kind})}\basekind\rho\\
\troot\envd\envt\basekind{(\trec\tvar\sigma\delta)}&\eqdef&\correct{\delta}\\[1em]

\abeftyp{\envd}{\envt}{\rho}{\delta} &\eqdef&
\textsf{match}\,(\troot{\envd}{\envt}\basekind{\rho})\,\textsf{with}\\
&&\begin{array}{@{|\,}l@{\;\Rightarrow\;}l}
\unboxed & \itrue\\
\nonsense & \pure{False}\\
\correct{\delta'} & \prece{\envd(\delta')}{\envd(\delta)}
\end{array}
\end{array}\]
\captionlabel{Root-related functions and assertions}{fig:troot}
\end{figure}

\Cref{fig:troot} presents the $\trootname$ function,
expecting a timestamp mapping~$\envd$,
a type mapping~$\envt$,
a kind $\kind$ and a type $\rho$,
and produces a function expecting $\kind$ timestamps
and returning an ``answer'', representing the root timestamp of $\rho$.
An answer is either a timestamp, $\textsf{Unboxed}$ to indicate an unboxed type,
of $\textsf{Nonsense}$ if the type has no sensible root timestamp (for example, $\tforall{\alpha}{\kind}{\alpha}$).
In this definition, we write $\textsf{Nonsense}_\kind$ the function expecting $\kind$ arguments
and returning $\textsf{Nonsense}$.

The assertion $\abeftyp{\envd}{\envt}{\rho}{\delta}$,
also presented in \Cref{fig:troot}
matches root timestamp of $\rho$.
If it is unboxed, the assertion is true,
if it is nonsensical, the assertion is false,
and if it is a regular timestamp $\delta'$,
then $\envd{(\delta')}$ must precede $\envd{(\delta)}$.

\subsection{Definition of our Logical Relations}
\label[appendix]{sec:logrel}
\newlength{\seplength}
\setlength{\seplength}{10pt}

\begin{figure}[!h]\small\morespacingaroundstar
\[\begin{array}{r@{}c@{}l}
\ikind\basekind{a} &\spaceeqdef&A\\
\ikind{(\skindsucc{\kind})}{A} &\spaceeqdef& \Vertices \rightarrow \ikind\kind{A}\\[\seplength]
\rhointerp{\envd}{\envt}\kind{\rho} &:& \tinterp{\kind}\\
\rhointerp{\envd}{\envt}\kind{\tvar} &\spaceeqdef& \begin{cases}
\Psi \,\ast_\kind\, \dotabef\kind\rfunc & \text{if }m(\tvar) = (\kind,(\Psi,\rfunc))\\
\bot_\kind& \text{else}
\end{cases}\\
\rhointerp{\envd}{\envt}{\basekind}{\tau} &\spaceeqdef& \lambda\val.\; \pure{\tau=\tunit \wedge \val=\vunit} \;\lor\; \pure{\tau=\tbool \wedge \val \in \{\vtrue,\vfalse\}} \;\lor\; \pure{\tau = \tinteger \wedge \val \in \mathbb{Z}} \\
\rhointerp{\envd}{\envt}{\basekind}{\sigma@\delta}&\spaceeqdef& \lambda\val.\; \abef{\val}{\envd(\delta)} \star \sigmainterp{\envd}{\envt}{\sigma}\,\delta\,\val \\
\rhointerp{\envd}{\envt}{\basekind}{\trec{\tvar}{\sigma}{\delta}}&\spaceeqdef& \mu(\Psi:\tinterpzero).\\
&&\lambda\val.\;\exists\wal.\; \pure{\val=\vfold{\wal}} \star \abef{\wal}{\envd(\delta)} \star \later\, \sigmainterp{\envd}{\tupd{\tvar}{\basekind}{\Psi}{\envd(\delta)}{\envt}}{\sigma}\,\delta\,\val \\
\rhointerp{\envd}{\envt}{\basekind}{\tforall{\tvar}{\kind}{\rho}}&\spaceeqdef& \lambda\val.\; \always \forall(\Psi:\tinterp{\kind})\,(\rfunc : \trfunc\kind).\\
&&\pure{\proper{\kind}{\Psi} \,\wedge\, \regular\kind\rfunc}  \wand \rhointerp{\envd}{\tupd{\tvar}{\kind}{\Psi}{\rfunc}{\envt}}{\basekind}{\rho}\,\val\\
\rhointerp{\envd}{\envt}{(\skindsucc\kind)}{\tlambda{\delta}{\rho}}&\spaceeqdef& \lambda\vertex.\; \rhointerp{\eupd\delta{\vertex}\envd}{\envt}{\kind}{\rho} \\
\rhointerp{\envd}{\envt}{\kind}{\tapp{\rho}{\delta}}&\spaceeqdef& \rhointerp{\envd}{\envt}{(\skindsucc{\kind})}{\rho}\,\envd(\delta) \\[\seplength]

\sinterp{\envd}{\envt}{\rho}&\spaceeqdef&\lambda\delta\,\val.\; \abeftyp{\envd}{\project\envt}{\rho}{\delta} \star \rhointerp{\envd}{\envt}{\basekind}{\rho}\,\val\\[\seplength]

\sigmainterp{\envd}{\envt}{\sigma} &:& \Vertices \rightarrow \Values \rightarrow \iProp \\
\sigmainterp{\envd}{\envt}{\tarray{\rho}}&\spaceeqdef&\lambda\delta\,\val.\;\exists\loc.\; \pure{\val = \loc} \star \iInv{\exists\vec\wal.\;\loc \mapsto \vec\wal \star \bigast{\val' \in \vec\wal}{(\sinterp{\envd}{\envt}{\rho}\,\delta\,\val')} } \\
\sigmainterp{\envd}{\envt}{\tproduct{\rho_1}{\rho_2}}&\spaceeqdef&\lambda\delta\,\val.\;\exists\loc\,\val_1\,\val_2.\;\pure{\val = \loc} \star \loc \mapstopersist (\val_1,\val_2) \star \sinterp{\envd}{\envt}{\rho_1}\,\delta\,\val_1 \star \sinterp{\envd}{\envt}{\rho_2}\,\delta\,\val_2\\
\sigmainterp\envd\envt{\tsum{\rho_1}{\rho_2}}&\spaceeqdef&\lambda\delta\,\val.\;\exists\loc\,\val'.\;\pure{\val=\loc} \star\\
&&(\loc \mapstopersist \einj{1}{\val'} \star \sinterp{\envd}{\envt}{\rho_1}\,\delta\,\val')  \;\lor\; (\loc \mapstopersist \einj{2}{\val'} \star \sinterp{\envd}{\envt}{\rho_2}\,\delta\,\val')\\
\sigmainterp{\envd}{\envt}{\tfun{\vec{\delta_1}}{\Delta_1}{\vec{\rho_1}}{\delta_f}{\rho_2}}
&\spaceeqdef&\lambda\delta\,\val.\; \always\forall\vec\vertex\,\vec{\wal}.\; \pure{\length{\vec{\delta_1}} = \length{\vec\vertex} \;\wedge\; \length{\vec{\rho_1}} = \length{\vec\wal}} \wand\\
&&\textsf{let}\,\envd' = \eupd{\vec{\delta_1}}{\vec{\vertex}}\envd\,\textsf{in}\\
&&\prece{\envd(\delta)}{\envd'(\delta_f)} \wand \edinterp{\envd'}{\Delta_1} \wand \bigast{\rho \in \vec{\rho_1},\,\val' \in \vec\wal}{(\sinterp{\envd'}{\envt}{\rho}\,\val')} \wand\\
&&\forall\vertex'.\; \tequiv{\vertex'}{\envd'(\delta_f)} \wand \wp{\vertex'}{\ecall{\val}{\vec\wal}}{\val'}{\sinterp{\envd'}{\envt}{\rho_2}\,\val'}
\end{array}\]
\captionlabel{The interpretation of types}{fig:interp}
\end{figure}

\Cref{fig:interp} presents the logical relations we define.
Formally, we define the (meta-type-level) function
$\ikind\kind{A}$ producing a function waiting for $\kind$ timestamps
and returning something of type $A$.
This function is defined by induction over the kind $\kind$.

More precisely, \Cref{fig:interp} presents the type relation and the boxed type relation.
The \emph{type relation}~$\rhointerp{\envd}{\envt}{\kind}{\rho}$
produces a function waiting
for $\kind$ timestamps,
a value~$\val$,
and captures that $\val$ is of type $\rho$,
within the timestamp mapping $\envd$ and type mapping $\envt$.
The \emph{boxed type relation} $\sigmainterp{\envd}{\envt}{\sigma}$
produces a predicate over a timestamp~$\delta$ and a value~$\val$,
capturing that $\val$ is of type $\sigma$ allocated at $\delta$,
within the timestamp mapping $\envd$ and type mapping $\envt$.
The type relation and the boxed type relation are defined
by mutual induction over their type argument.
The omitted cases are all sent to $\bot_\kind$,
the always false predicate ignoring its $\kind$ arguments.

\paragraph{Interpretation of types}
Let us first present the relation
$\rhointerp{\envd}{\envt}{\kind}{\rho}$.

If $\rho$ is a variable $\tvar$, then
$\tvar$ must have kind $\kind$ in the type mapping $\envt$,
linked with predicate $\Psi$ and timestamp function $\rfunc$.
The relation returns the predicate $\Psi \;\ast_\kind\; \dotabef\kind\rfunc$.
The operator $\ast_\kind$ lifts the separating conjunction to predicated
in $\tinterp{\kind}$
by distributing $\kind$ timestamp arguments and a value to $\Psi$
and $\dotabef\kind\rfunc$.
The predicate $\dotabef\kind\rfunc$ is of type $\tinterp{\kind}$;
it feeds~$\kind$ timestamps to $\rfunc$,
and asserts that the value argument was allocated before the result of $\rfunc$.
For example,
in the particular case of $\kind = \,\kindsucc\basekind$,
we have that $\Psi \;\ast_\kind\; \dotabef\kind\rfunc \;=\; \lambda\delta\,\val.\, \Psi\,\delta\,\val \;\star\; \abef{\val}{(\rfunc\,\delta)}$.

If $\rho$ is a base type $\tau$, the kind must be the base kind $\basekind$,
and the relation binds a value
which must correspond to the particular base type under consideration.

If $\rho$ is a boxed type $\sigma @ \rho$,
the kind must be $\basekind$,
and the relation binds a value~$\val$
which must have been allocated before~$\envd(\delta)$,
the timestamp associated to $\delta$ in $\envd$.
The relation also ensures that~$\val$
recursively satisfies the interpretation of $\sigma$.

If $\rho$ is a recursive type $\trec{\tvar}{\sigma}{\delta}$,
the kind must be $\basekind$, and
the relation is expressed as a \emph{guarded fixed-point}.
Intuitively, the predicate $\Psi$, from a value to \iProp,
captures the interpretation of
the recursive type itself.
The interpretation binds a value $\val$ and asserts that
it is of the form $\vfold\wal$.
The interpretation is then similar to a boxed type:
$\wal$ must have been allocated before $\envd(\delta)$
and be in relation with the interpretation of $\sigma$,
with a type environment updated to bind $\tvar$ to $\Psi$
as well as the root timestamp $\envd(\delta)$.

If $\rho$ is a type abstraction $\tforall{\tvar}{\kind}{\rho}$,
the kind must be $\basekind$,
and the relation binds a value~$\val$.
The relation universally quantifies over the predicate $\Psi$
and the timestamp function $\rfunc$, which
will be instantiated during the \RULE{T-TApp} rule.
Both $\Psi$ and $\rfunc$ are constrained.
The property~$\proper{\kind}{\Psi$} captures that
any timestamp parameter of $\Psi$ can be replaced by an equivalent one.
The property~$\regular{\kind}{\rfunc}$ captures that
the function $\rfunc$ either ignores all its arguments or returns one of them.
These two properties are needed in order to prove that \RULE{T-Par}
is sound.
The relation then calls itself recursively on $\rho$,
augmenting the type mapping by associating $\tvar$
to its kind $\kappa$ and the pair of $\Psi$ and $\rfunc$.

If $\rho$ is a timestamp abstraction $\tlambda{\delta}{\rho}$,
the kind must be of the form $\kindsucc{\kind}$,
and the relation expands to a function waiting for a timestamp $\delta$
and adding it to the timestamp mapping~$\envd$.

If $\rho$ is a timestamp application $\tapp{\rho}{\delta}$
at some kind $\kind$, then the relation applies the timestamp $\envd(\delta)$
to the interpretation
of $\rho$ at kind $\kindsucc{\kind}$.

\paragraph{Interpretation of boxed types}
The \emph{enriched type interpretation} $\sinterp{\envd}{\envt}{\rho}$,
defined next in \cref{fig:interp},
is a predicate over a timestamp $\delta$
and a value $\val$.
It asserts that the root timestamp of $\rho$
comes before~$\delta$ and that~$\val$
is in relation with
the interpretation of $\rho$.
This wrapper in used for the interpretation of boxed types,
which we present next.
The interpretation of boxed types is written
$\sigmainterp{\envd}{\envt}{\sigma}$
and is a predicate over a timestamp variable~$\delta$
and a value~$\val$.

If $\sigma$ is an array $\tarray{\rho}$,
then $\val$ must be a location~$\loc$, such that
$\loc$ points-to an array $\wals$ and that for each value
$\val'$ in $\wals$ is in relation with the enriched interpretation
of $\rho$.
The points-to assertion and the relation on the values of the array
appears inside an invariant, ensuring their persistence.

If $\sigma$ is a pair $\tproduct{\rho_1}{\rho_2}$,
then $\val$ must be a location $\loc$ pointing to
a pair of values $(\val_1,\val_2)$ such that~$\val_1$
(resp. $\val_2$) is in relation with the enriched
interpretation of $\rho_1$ (resp. $\rho_2$).
The sum case is similar.

If $\sigma$ is an arrow $\tfun{\vec{\delta_1}}{\Delta_1}{\vec{\rho_1}}{\delta_f}{\rho_2}$,
then the interpretation quantifies
over the list of timestamp arguments $\vec\vertex$
and the list of arguments of the function $\wals$, which must both have the correct length.
The relation then defines $\envd'$, the new timestamp environment,
being $\envd$ where $\vec{\delta_1}$ are instantiated with $\vec\vertex$.
The relation next requires that
the allocation timestamp $\envd(\delta)$ precedes the
timestamp of the caller $\envd'(\delta_f)$, and that every
value in $\wals$ is of the correct type.
Last, the relation requires that for any timestamp equivalent to
$\envd'(\delta_f)$, the WP of the function call holds, and that the returned value
is in relation with the enriched interpretation of the return type $\rho_2$.

\fi

\bibliography{english,local}


\begin{thebibliography}{68}


\ifx \showCODEN    \undefined \def \showCODEN     #1{\unskip}     \fi
\ifx \showISBNx    \undefined \def \showISBNx     #1{\unskip}     \fi
\ifx \showISBNxiii \undefined \def \showISBNxiii  #1{\unskip}     \fi
\ifx \showISSN     \undefined \def \showISSN      #1{\unskip}     \fi
\ifx \showLCCN     \undefined \def \showLCCN      #1{\unskip}     \fi
\ifx \shownote     \undefined \def \shownote      #1{#1}          \fi
\ifx \showarticletitle \undefined \def \showarticletitle #1{#1}   \fi
\ifx \showURL      \undefined \def \showURL       {\relax}        \fi
\providecommand\bibfield[2]{#2}
\providecommand\bibinfo[2]{#2}
\providecommand\natexlab[1]{#1}
\providecommand\showeprint[2][]{arXiv:#2}

\bibitem[Abdi et~al\mbox{.}(2024)]%
        {DBLP:conf/spaa/0002PZWJ24}
\bibfield{author}{\bibinfo{person}{Javad Abdi}, \bibinfo{person}{Gilead Posluns}, \bibinfo{person}{Guozheng Zhang}, \bibinfo{person}{Boxuan Wang}, {and} \bibinfo{person}{Mark~C. Jeffrey}.} \bibinfo{year}{2024}\natexlab{}.
\newblock \showarticletitle{When Is Parallelism Fearless and Zero-Cost with Rust?}. In \bibinfo{booktitle}{\emph{Proceedings of the 36th {ACM} Symposium on Parallelism in Algorithms and Architectures, {SPAA} 2024, Nantes, France, June 17-21, 2024}}, \bibfield{editor}{\bibinfo{person}{Kunal Agrawal} {and} \bibinfo{person}{Erez Petrank}} (Eds.). \bibinfo{publisher}{{ACM}}, \bibinfo{pages}{27--40}.
\newblock
\href{https://doi.org/10.1145/3626183.3659966}{doi:\nolinkurl{10.1145/3626183.3659966}}


\bibitem[Acar et~al\mbox{.}(2020)]%
        {mpl-2020}
\bibfield{author}{\bibinfo{person}{Umut~A. Acar}, \bibinfo{person}{Jatin Arora}, \bibinfo{person}{Matthew Fluet}, \bibinfo{person}{Ram Raghunathan}, \bibinfo{person}{Sam Westrick}, {and} \bibinfo{person}{Rohan Yadav}.} \bibinfo{year}{2020}\natexlab{}.
\newblock \bibinfo{title}{{MPL}: A high-performance compiler for Parallel ML}.
\newblock
\urldef\tempurl%
\url{https://github.com/MPLLang/mpl}
\showURL{%
\tempurl}


\bibitem[Acar et~al\mbox{.}(2015)]%
        {acar-et-al-15}
\bibfield{author}{\bibinfo{person}{Umut~A. Acar}, \bibinfo{person}{Guy~E. Blelloch}, \bibinfo{person}{Matthew Fluet}, \bibinfo{person}{Stefan~K. Muller}, {and} \bibinfo{person}{Ram Raghunathan}.} \bibinfo{year}{2015}\natexlab{}.
\newblock \showarticletitle{Coupling Memory and Computation for Locality Management}. In \bibinfo{booktitle}{\emph{1st Summit on Advances in Programming Languages, {SNAPL} 2015, May 3-6, 2015, Asilomar, California, {USA}}} \emph{(\bibinfo{series}{LIPIcs}, Vol.~\bibinfo{volume}{32})}, \bibfield{editor}{\bibinfo{person}{Thomas Ball}, \bibinfo{person}{Rastislav Bod{\'{\i}}k}, \bibinfo{person}{Shriram Krishnamurthi}, \bibinfo{person}{Benjamin~S. Lerner}, {and} \bibinfo{person}{Greg Morrisett}} (Eds.). \bibinfo{publisher}{Schloss Dagstuhl - Leibniz-Zentrum f{\"{u}}r Informatik}, \bibinfo{pages}{1--14}.
\newblock
\href{https://doi.org/10.4230/LIPICS.SNAPL.2015.1}{doi:\nolinkurl{10.4230/LIPICS.SNAPL.2015.1}}


\bibitem[Acar et~al\mbox{.}(2016)]%
        {dag-calculus}
\bibfield{author}{\bibinfo{person}{Umut~A. Acar}, \bibinfo{person}{Arthur Charguéraud}, \bibinfo{person}{Mike Rainey}, {and} \bibinfo{person}{Filip Sieczkowski}.} \bibinfo{year}{2016}\natexlab{}.
\newblock \showarticletitle{Dag-calculus: a calculus for parallel computation}. In \bibinfo{booktitle}{\emph{International Conference on Functional Programming (ICFP)}}. \bibinfo{pages}{18--32}.
\newblock
\urldef\tempurl%
\url{https://doi.org/10.1145/2951913.2951946}
\showURL{%
\tempurl}


\bibitem[Anderson et~al\mbox{.}(2022)]%
        {DBLP:conf/ppopp/AndersonBDD022}
\bibfield{author}{\bibinfo{person}{Daniel Anderson}, \bibinfo{person}{Guy~E. Blelloch}, \bibinfo{person}{Laxman Dhulipala}, \bibinfo{person}{Magdalen Dobson}, {and} \bibinfo{person}{Yihan Sun}.} \bibinfo{year}{2022}\natexlab{}.
\newblock \showarticletitle{The problem-based benchmark suite (PBBS), {V2}}. In \bibinfo{booktitle}{\emph{PPoPP '22: 27th {ACM} {SIGPLAN} Symposium on Principles and Practice of Parallel Programming, Seoul, Republic of Korea, April 2 - 6, 2022}}, \bibfield{editor}{\bibinfo{person}{Jaejin Lee}, \bibinfo{person}{Kunal Agrawal}, {and} \bibinfo{person}{Michael~F. Spear}} (Eds.). \bibinfo{publisher}{{ACM}}, \bibinfo{pages}{445--447}.
\newblock
\showISBNx{978-1-4503-9204-4}
\href{https://doi.org/10.1145/3503221.3508422}{doi:\nolinkurl{10.1145/3503221.3508422}}


\bibitem[Appel(1992)]%
        {appel-92}
\bibfield{author}{\bibinfo{person}{Andrew~W. Appel}.} \bibinfo{year}{1992}\natexlab{}.
\newblock \bibinfo{booktitle}{\emph{Compiling with Continuations}}.
\newblock \bibinfo{publisher}{Cambridge University Press}.
\newblock
\urldef\tempurl%
\url{http://www.cambridge.org/9780521033114}
\showURL{%
\tempurl}


\bibitem[Arora et~al\mbox{.}(2024)]%
        {arora-muller-acar-24}
\bibfield{author}{\bibinfo{person}{Jatin Arora}, \bibinfo{person}{Stefan~K. Muller}, {and} \bibinfo{person}{Umut~A. Acar}.} \bibinfo{year}{2024}\natexlab{}.
\newblock \showarticletitle{Disentanglement with Futures, State, and Interaction}.
\newblock \bibinfo{journal}{\emph{Proc. {ACM} Program. Lang.}} \bibinfo{volume}{8}, \bibinfo{number}{{POPL}} (\bibinfo{year}{2024}), \bibinfo{pages}{1569--1599}.
\newblock
\href{https://doi.org/10.1145/3632895}{doi:\nolinkurl{10.1145/3632895}}


\bibitem[Arora et~al\mbox{.}(2021)]%
        {arora-westrick-acar-21}
\bibfield{author}{\bibinfo{person}{Jatin Arora}, \bibinfo{person}{Sam Westrick}, {and} \bibinfo{person}{Umut~A. Acar}.} \bibinfo{year}{2021}\natexlab{}.
\newblock \showarticletitle{Provably space-efficient parallel functional programming}.
\newblock \bibinfo{journal}{\emph{Proc. {ACM} Program. Lang.}} \bibinfo{volume}{5}, \bibinfo{number}{{POPL}} (\bibinfo{year}{2021}), \bibinfo{pages}{1--33}.
\newblock
\href{https://doi.org/10.1145/3434299}{doi:\nolinkurl{10.1145/3434299}}


\bibitem[Arora et~al\mbox{.}(2023)]%
        {arora-westrick-acar-23}
\bibfield{author}{\bibinfo{person}{Jatin Arora}, \bibinfo{person}{Sam Westrick}, {and} \bibinfo{person}{Umut~A. Acar}.} \bibinfo{year}{2023}\natexlab{}.
\newblock \showarticletitle{Efficient Parallel Functional Programming with Effects}.
\newblock \bibinfo{journal}{\emph{Proc. {ACM} Program. Lang.}} \bibinfo{volume}{7}, \bibinfo{number}{{PLDI}} (\bibinfo{year}{2023}), \bibinfo{pages}{1558--1583}.
\newblock
\href{https://doi.org/10.1145/3591284}{doi:\nolinkurl{10.1145/3591284}}


\bibitem[Balzer et~al\mbox{.}(2019)]%
        {BalzerESOP2019}
\bibfield{author}{\bibinfo{person}{Stephanie Balzer}, \bibinfo{person}{Bernardo Toninho}, {and} \bibinfo{person}{Frank Pfenning}.} \bibinfo{year}{2019}\natexlab{}.
\newblock \showarticletitle{Manifest Deadlock-Freedom for Shared Session Types}. In \bibinfo{booktitle}{\emph{28th European Symposium on Programming ({ESOP})}} \emph{(\bibinfo{series}{Lecture Notes in Computer Science}, Vol.~\bibinfo{volume}{11423})}. \bibinfo{publisher}{Springer}, \bibinfo{pages}{611--639}.
\newblock
\href{https://doi.org/10.1007/978-3-030-17184-1\_22}{doi:\nolinkurl{10.1007/978-3-030-17184-1\_22}}


\bibitem[Barendregt(1984)]%
        {barendregt}
\bibfield{author}{\bibinfo{person}{Henk~P. Barendregt}.} \bibinfo{year}{1984}\natexlab{}.
\newblock \bibinfo{booktitle}{\emph{The Lambda Calculus, Its Syntax and Semantics}}.
\newblock \bibinfo{publisher}{Elsevier}.
\newblock
\urldef\tempurl%
\url{http://www.elsevier.com/wps/find/bookdescription.cws_home/501727/description}
\showURL{%
\tempurl}


\bibitem[{Bocchino Jr.} et~al\mbox{.}(2009)]%
        {bocchino-et-al-09}
\bibfield{author}{\bibinfo{person}{Robert~L. {Bocchino Jr.}}, \bibinfo{person}{Vikram~S. Adve}, \bibinfo{person}{Danny Dig}, \bibinfo{person}{Sarita~V. Adve}, \bibinfo{person}{Stephen Heumann}, \bibinfo{person}{Rakesh Komuravelli}, \bibinfo{person}{Jeffrey Overbey}, \bibinfo{person}{Patrick Simmons}, \bibinfo{person}{Hyojin Sung}, {and} \bibinfo{person}{Mohsen Vakilian}.} \bibinfo{year}{2009}\natexlab{}.
\newblock \showarticletitle{A type and effect system for deterministic parallel Java}. In \bibinfo{booktitle}{\emph{Proceedings of the 24th Annual {ACM} {SIGPLAN} Conference on Object-Oriented Programming, Systems, Languages, and Applications, {OOPSLA} 2009, October 25-29, 2009, Orlando, Florida, {USA}}}, \bibfield{editor}{\bibinfo{person}{Shail Arora} {and} \bibinfo{person}{Gary~T. Leavens}} (Eds.). \bibinfo{publisher}{{ACM}}, \bibinfo{pages}{97--116}.
\newblock
\href{https://doi.org/10.1145/1640089.1640097}{doi:\nolinkurl{10.1145/1640089.1640097}}


\bibitem[Boyapati et~al\mbox{.}(2002)]%
        {BoyapatiOOPSLA2002}
\bibfield{author}{\bibinfo{person}{Chandrasekhar Boyapati}, \bibinfo{person}{Robert Lee}, {and} \bibinfo{person}{Martin~C. Rinard}.} \bibinfo{year}{2002}\natexlab{}.
\newblock \showarticletitle{Ownership types for safe programming: preventing data races and deadlocks}. In \bibinfo{booktitle}{\emph{{ACM} {SIGPLAN} Conference on Object-Oriented Programming Systems, Languages and Applications ({OOPSLA})}}. \bibinfo{publisher}{{ACM}}, \bibinfo{pages}{211--230}.
\newblock
\href{https://doi.org/10.1145/582419.582440}{doi:\nolinkurl{10.1145/582419.582440}}


\bibitem[Boyapati et~al\mbox{.}(2003)]%
        {BoyapatiPOPL2003}
\bibfield{author}{\bibinfo{person}{Chandrasekhar Boyapati}, \bibinfo{person}{Barbara Liskov}, {and} \bibinfo{person}{Liuba Shrira}.} \bibinfo{year}{2003}\natexlab{}.
\newblock \showarticletitle{Ownership types for object encapsulation}. In \bibinfo{booktitle}{\emph{30th {SIGPLAN-SIGACT} Symposium on Principles of Programming Languages ({POPL})}}. \bibinfo{publisher}{{ACM}}, \bibinfo{pages}{213--223}.
\newblock
\href{https://doi.org/10.1145/604131.604156}{doi:\nolinkurl{10.1145/604131.604156}}


\bibitem[Boyapati and Rinard(2001)]%
        {BoyapatiRinardOOPSLA2001}
\bibfield{author}{\bibinfo{person}{Chandrasekhar Boyapati} {and} \bibinfo{person}{Martin~C. Rinard}.} \bibinfo{year}{2001}\natexlab{}.
\newblock \showarticletitle{A Parameterized Type System for Race-Free Java Programs}. In \bibinfo{booktitle}{\emph{{ACM} {SIGPLAN} Conference on Object-Oriented Programming Systems, Languages and Applications ({OOPSLA})}}. \bibinfo{publisher}{{ACM}}, \bibinfo{pages}{56--69}.
\newblock
\href{https://doi.org/10.1145/504282.504287}{doi:\nolinkurl{10.1145/504282.504287}}


\bibitem[Caires et~al\mbox{.}(2019)]%
        {CairesCONCUR2019}
\bibfield{author}{\bibinfo{person}{Lu{\'{\i}}s Caires}, \bibinfo{person}{Jorge~A. P{\'{e}}rez}, \bibinfo{person}{Frank Pfenning}, {and} \bibinfo{person}{Bernardo Toninho}.} \bibinfo{year}{2019}\natexlab{}.
\newblock \showarticletitle{Domain-Aware Session Types}. In \bibinfo{booktitle}{\emph{30th International Conference on Concurrency Theory ({CONCUR})}} \emph{(\bibinfo{series}{LIPIcs}, Vol.~\bibinfo{volume}{140})}. \bibinfo{publisher}{Schloss Dagstuhl - Leibniz-Zentrum f{\"{u}}r Informatik}, \bibinfo{pages}{39:1--39:17}.
\newblock
\href{https://doi.org/10.4230/LIPICS.CONCUR.2019.39}{doi:\nolinkurl{10.4230/LIPICS.CONCUR.2019.39}}


\bibitem[Clarke et~al\mbox{.}(1998)]%
        {ClarkeOOPSLA1998}
\bibfield{author}{\bibinfo{person}{David~G. Clarke}, \bibinfo{person}{John Potter}, {and} \bibinfo{person}{James Noble}.} \bibinfo{year}{1998}\natexlab{}.
\newblock \showarticletitle{Ownership Types for Flexible Alias Protection}. In \bibinfo{booktitle}{\emph{{ACM} {SIGPLAN} Conference on Object-Oriented Programming Systems, Languages and Applications ({OOPSLA})}}. \bibinfo{publisher}{{ACM}}, \bibinfo{pages}{48--64}.
\newblock
\href{https://doi.org/10.1145/286936.286947}{doi:\nolinkurl{10.1145/286936.286947}}


\bibitem[Constable et~al\mbox{.}(1986)]%
        {ConstableBook1986}
\bibfield{author}{\bibinfo{person}{Robert~L. Constable}, \bibinfo{person}{Stuart~F. Allen}, \bibinfo{person}{Mark Bromley}, \bibinfo{person}{Rance Cleaveland}, \bibinfo{person}{J.~F. Cremer}, \bibinfo{person}{Robert Harper}, \bibinfo{person}{Douglas~J. Howe}, \bibinfo{person}{Todd~B. Knoblock}, \bibinfo{person}{Nax~Paul Mendler}, \bibinfo{person}{Prakash Panangaden}, \bibinfo{person}{James~T. Sasaki}, {and} \bibinfo{person}{Scott~F. Smith}.} \bibinfo{year}{1986}\natexlab{}.
\newblock \bibinfo{booktitle}{\emph{Implementing Mathematics with the {Nuprl} Proof Development System}}.
\newblock \bibinfo{publisher}{Prentice Hall}.
\newblock
\showISBNx{978-0-13-451832-9}
\urldef\tempurl%
\url{http://dl.acm.org/citation.cfm?id=10510}
\showURL{%
\tempurl}


\bibitem[de~Vilhena(2022)]%
        {thesis_paulo}
\bibfield{author}{\bibinfo{person}{Paulo~Em{\'i}lio de Vilhena}.} \bibinfo{year}{2022}\natexlab{}.
\newblock \emph{\bibinfo{title}{{Proof of Programs with Effect Handlers}}}.
\newblock Theses. \bibinfo{school}{{Universit{\'e} Paris Cit{\'e}}}.
\newblock
\urldef\tempurl%
\url{https://inria.hal.science/tel-03891381}
\showURL{%
\tempurl}


\bibitem[Derakhshan et~al\mbox{.}(2021)]%
        {DerakhshanLICS2021}
\bibfield{author}{\bibinfo{person}{Farzaneh Derakhshan}, \bibinfo{person}{Stephanie Balzer}, {and} \bibinfo{person}{Limin Jia}.} \bibinfo{year}{2021}\natexlab{}.
\newblock \showarticletitle{Session Logical Relations for Noninterference}. In \bibinfo{booktitle}{\emph{36th Annual {ACM/IEEE} Symposium on Logic in Computer Science ({LICS})}}. \bibinfo{publisher}{{IEEE} Computer Society}, \bibinfo{pages}{1--14}.
\newblock
\href{https://doi.org/10.1109/LICS52264.2021.9470654}{doi:\nolinkurl{10.1109/LICS52264.2021.9470654}}


\bibitem[Derakhshan et~al\mbox{.}(2024)]%
        {DerakhshanECOOP2024}
\bibfield{author}{\bibinfo{person}{Farzaneh Derakhshan}, \bibinfo{person}{Stephanie Balzer}, {and} \bibinfo{person}{Yue Yao}.} \bibinfo{year}{2024}\natexlab{}.
\newblock \showarticletitle{Regrading Policies for Flexible Information Flow Control in Session-Typed Concurrency}. In \bibinfo{booktitle}{\emph{38th European Conference on Object-Oriented Programming ({ECOOP})}} \emph{(\bibinfo{series}{LIPIcs}, Vol.~\bibinfo{volume}{313})}. \bibinfo{publisher}{Schloss Dagstuhl - Leibniz-Zentrum f{\"{u}}r Informatik}, \bibinfo{pages}{11:1--11:29}.
\newblock
\href{https://doi.org/10.4230/LIPICS.ECOOP.2024.11}{doi:\nolinkurl{10.4230/LIPICS.ECOOP.2024.11}}


\bibitem[Dietl and M{\"{u}}ller(2005)]%
        {DietlMullerARTICLE2005}
\bibfield{author}{\bibinfo{person}{Werner Dietl} {and} \bibinfo{person}{Peter M{\"{u}}ller}.} \bibinfo{year}{2005}\natexlab{}.
\newblock \showarticletitle{Universes: Lightweight Ownership for {JML}}.
\newblock \bibinfo{journal}{\emph{Journal of Object Technology}} \bibinfo{volume}{4}, \bibinfo{number}{8} (\bibinfo{year}{2005}), \bibinfo{pages}{5--32}.
\newblock
\href{https://doi.org/10.5381/JOT.2005.4.8.A1}{doi:\nolinkurl{10.5381/JOT.2005.4.8.A1}}


\bibitem[Elsman and Henriksen(2023)]%
        {elsman-henriksen-23}
\bibfield{author}{\bibinfo{person}{Martin Elsman} {and} \bibinfo{person}{Troels Henriksen}.} \bibinfo{year}{2023}\natexlab{}.
\newblock \showarticletitle{Parallelism in a Region Inference Context}.
\newblock \bibinfo{journal}{\emph{Proc. {ACM} Program. Lang.}} \bibinfo{volume}{7}, \bibinfo{number}{{PLDI}} (\bibinfo{year}{2023}), \bibinfo{pages}{884--906}.
\newblock
\href{https://doi.org/10.1145/3591256}{doi:\nolinkurl{10.1145/3591256}}


\bibitem[Flanagan and Freund(2009)]%
        {FlanaganFreundPLDI2009}
\bibfield{author}{\bibinfo{person}{Cormac Flanagan} {and} \bibinfo{person}{Stephen~N. Freund}.} \bibinfo{year}{2009}\natexlab{}.
\newblock \showarticletitle{FastTrack: efficient and precise dynamic race detection}. In \bibinfo{booktitle}{\emph{{ACM} {SIGPLAN} Conference on Programming Language Design and Implementation ({PLDI})}}. \bibinfo{publisher}{{ACM}}, \bibinfo{pages}{121--133}.
\newblock
\href{https://doi.org/10.1145/1542476.1542490}{doi:\nolinkurl{10.1145/1542476.1542490}}


\bibitem[Georges et~al\mbox{.}(2025)]%
        {georges-et-al-25}
\bibfield{author}{\bibinfo{person}{A{\"{\i}}na~Linn Georges}, \bibinfo{person}{Benjamin Peters}, \bibinfo{person}{Laila Elbeheiry}, \bibinfo{person}{Leo White}, \bibinfo{person}{Stephen Dolan}, \bibinfo{person}{Richard~A. Eisenberg}, \bibinfo{person}{Chris Casinghino}, \bibinfo{person}{Fran{\c{c}}ois Pottier}, {and} \bibinfo{person}{Derek Dreyer}.} \bibinfo{year}{2025}\natexlab{}.
\newblock \showarticletitle{Data Race Freedom {\`{a}} la Mode}.
\newblock \bibinfo{journal}{\emph{Proc. {ACM} Program. Lang.}} \bibinfo{volume}{9}, \bibinfo{number}{{POPL}} (\bibinfo{year}{2025}), \bibinfo{pages}{656--686}.
\newblock
\href{https://doi.org/10.1145/3704859}{doi:\nolinkurl{10.1145/3704859}}


\bibitem[Girard(1972)]%
        {girard-72}
\bibfield{author}{\bibinfo{person}{Jean-Yves Girard}.} \bibinfo{year}{1972}\natexlab{}.
\newblock \emph{\bibinfo{title}{Interprétation fonctionnelle et élimination des coupures de l'arith\-mé\-ti\-que d'ordre supérieur}}.
\newblock Th\`ese d'\'Etat. \bibinfo{school}{Universit{\'e} Paris 7}.
\newblock
\urldef\tempurl%
\url{https://girard.perso.math.cnrs.fr/These.pdf}
\showURL{%
\tempurl}


\bibitem[Gregersen et~al\mbox{.}(2021)]%
        {OddershedeGregersenPOPL2021}
\bibfield{author}{\bibinfo{person}{Simon~Oddershede Gregersen}, \bibinfo{person}{Johan Bay}, \bibinfo{person}{Amin Timany}, {and} \bibinfo{person}{Lars Birkedal}.} \bibinfo{year}{2021}\natexlab{}.
\newblock \showarticletitle{Mechanized logical relations for termination-insensitive noninterference}.
\newblock \bibinfo{journal}{\emph{Proc. {ACM} Program. Lang.}} \bibinfo{volume}{5}, \bibinfo{number}{{POPL}} (\bibinfo{year}{2021}), \bibinfo{pages}{1--29}.
\newblock
\href{https://doi.org/10.1145/3434291}{doi:\nolinkurl{10.1145/3434291}}


\bibitem[Grossman et~al\mbox{.}(2002)]%
        {grossman-et-al-02}
\bibfield{author}{\bibinfo{person}{Dan Grossman}, \bibinfo{person}{J.~Gregory Morrisett}, \bibinfo{person}{Trevor Jim}, \bibinfo{person}{Michael~W. Hicks}, \bibinfo{person}{Yanling Wang}, {and} \bibinfo{person}{James Cheney}.} \bibinfo{year}{2002}\natexlab{}.
\newblock \showarticletitle{Region-Based Memory Management in Cyclone}. In \bibinfo{booktitle}{\emph{Proceedings of the 2002 {ACM} {SIGPLAN} Conference on Programming Language Design and Implementation (PLDI), Berlin, Germany, June 17-19, 2002}}, \bibfield{editor}{\bibinfo{person}{Jens Knoop} {and} \bibinfo{person}{Laurie~J. Hendren}} (Eds.). \bibinfo{publisher}{{ACM}}, \bibinfo{pages}{282--293}.
\newblock
\href{https://doi.org/10.1145/512529.512563}{doi:\nolinkurl{10.1145/512529.512563}}


\bibitem[Guatto et~al\mbox{.}(2018)]%
        {guatto-et-al-18}
\bibfield{author}{\bibinfo{person}{Adrien Guatto}, \bibinfo{person}{Sam Westrick}, \bibinfo{person}{Ram Raghunathan}, \bibinfo{person}{Umut~A. Acar}, {and} \bibinfo{person}{Matthew Fluet}.} \bibinfo{year}{2018}\natexlab{}.
\newblock \showarticletitle{Hierarchical memory management for mutable state}. In \bibinfo{booktitle}{\emph{Proceedings of the 23rd {ACM} {SIGPLAN} Symposium on Principles and Practice of Parallel Programming, PPoPP 2018, Vienna, Austria, February 24-28, 2018}}, \bibfield{editor}{\bibinfo{person}{Andreas Krall} {and} \bibinfo{person}{Thomas~R. Gross}} (Eds.). \bibinfo{publisher}{{ACM}}, \bibinfo{pages}{81--93}.
\newblock
\href{https://doi.org/10.1145/3178487.3178494}{doi:\nolinkurl{10.1145/3178487.3178494}}


\bibitem[Jung et~al\mbox{.}(2018a)]%
        {rustbelt-18}
\bibfield{author}{\bibinfo{person}{Ralf Jung}, \bibinfo{person}{Jacques-Henri Jourdan}, \bibinfo{person}{Robbert Krebbers}, {and} \bibinfo{person}{Derek Dreyer}.} \bibinfo{year}{2018}\natexlab{a}.
\newblock \showarticletitle{{RustBelt}: Securing the Foundations of the {Rust} Programming Language}.
\newblock \bibinfo{journal}{\emph{Proceedings of the ACM on Programming Languages}} \bibinfo{volume}{2}, \bibinfo{number}{{POPL}} (\bibinfo{year}{2018}), \bibinfo{pages}{66:1--66:34}.
\newblock
\urldef\tempurl%
\url{https://people.mpi-sws.org/~dreyer/papers/rustbelt/paper.pdf}
\showURL{%
\tempurl}


\bibitem[Jung et~al\mbox{.}(2018b)]%
        {iris}
\bibfield{author}{\bibinfo{person}{Ralf Jung}, \bibinfo{person}{Robbert Krebbers}, \bibinfo{person}{Jacques-Henri Jourdan}, \bibinfo{person}{Ale{\v s} Bizjak}, \bibinfo{person}{Lars Birkedal}, {and} \bibinfo{person}{Derek Dreyer}.} \bibinfo{year}{2018}\natexlab{b}.
\newblock \showarticletitle{Iris from the ground up: {A} modular foundation for higher-order concurrent separation logic}.
\newblock \bibinfo{journal}{\emph{Journal of Functional Programming}}  \bibinfo{volume}{28} (\bibinfo{year}{2018}), \bibinfo{pages}{e20}.
\newblock
\urldef\tempurl%
\url{https://people.mpi-sws.org/~dreyer/papers/iris-ground-up/paper.pdf}
\showURL{%
\tempurl}


\bibitem[Klabnik and Nichols(2023)]%
        {klabnik2023rust}
\bibfield{author}{\bibinfo{person}{Steve Klabnik} {and} \bibinfo{person}{Carol Nichols}.} \bibinfo{year}{2023}\natexlab{}.
\newblock \bibinfo{booktitle}{\emph{The Rust programming language}}.
\newblock \bibinfo{publisher}{No Starch Press}.
\newblock


\bibitem[Knuth(1998)]%
        {knuth-98}
\bibfield{author}{\bibinfo{person}{Donald~E. Knuth}.} \bibinfo{year}{1998}\natexlab{}.
\newblock \bibinfo{booktitle}{\emph{The Art of Computer Programming, Volume 3: (2nd Ed.) Sorting and Searching}}.
\newblock \bibinfo{publisher}{Addison Wesley Longman Publishing Co., Inc.}, \bibinfo{address}{USA}.
\newblock
\showISBNx{0201896850}


\bibitem[Landin(1964)]%
        {landin-64}
\bibfield{author}{\bibinfo{person}{Peter~J. Landin}.} \bibinfo{year}{1964}\natexlab{}.
\newblock \showarticletitle{The Mechanical Evaluation of Expressions}.
\newblock \bibinfo{journal}{\emph{Computer Journal}} \bibinfo{volume}{6}, \bibinfo{number}{4} (\bibinfo{date}{Jan.} \bibinfo{year}{1964}), \bibinfo{pages}{308--320}.
\newblock


\bibitem[Lorenzen et~al\mbox{.}(2024)]%
        {lorenzen-et-al-24}
\bibfield{author}{\bibinfo{person}{Anton Lorenzen}, \bibinfo{person}{Leo White}, \bibinfo{person}{Stephen Dolan}, \bibinfo{person}{Richard~A. Eisenberg}, {and} \bibinfo{person}{Sam Lindley}.} \bibinfo{year}{2024}\natexlab{}.
\newblock \showarticletitle{Oxidizing OCaml with Modal Memory Management}.
\newblock \bibinfo{journal}{\emph{Proc. {ACM} Program. Lang.}} \bibinfo{volume}{8}, \bibinfo{number}{{ICFP}} (\bibinfo{year}{2024}), \bibinfo{pages}{485--514}.
\newblock
\href{https://doi.org/10.1145/3674642}{doi:\nolinkurl{10.1145/3674642}}


\bibitem[Martin{-}L{\"{o}}f(1982)]%
        {LoefARTICLE1982}
\bibfield{author}{\bibinfo{person}{Per Martin{-}L{\"{o}}f}.} \bibinfo{year}{1982}\natexlab{}.
\newblock \showarticletitle{Constructive Mathematics and Computer Programming}.
\newblock In \bibinfo{booktitle}{\emph{Logic, Methodology and Philosophy of Science VI}}. \bibinfo{series}{Studies in Logic and the Foundations of Mathematics}, Vol.~\bibinfo{volume}{104}. \bibinfo{publisher}{Elsevier}, \bibinfo{pages}{153--175}.
\newblock
\showISSN{0049-237X}
\href{https://doi.org/10.1016/S0049-237X(09)70189-2}{doi:\nolinkurl{10.1016/S0049-237X(09)70189-2}}


\bibitem[Milano et~al\mbox{.}(2022)]%
        {DBLP:conf/pldi/MilanoTM22}
\bibfield{author}{\bibinfo{person}{Mae Milano}, \bibinfo{person}{Joshua Turcotti}, {and} \bibinfo{person}{Andrew~C. Myers}.} \bibinfo{year}{2022}\natexlab{}.
\newblock \showarticletitle{A flexible type system for fearless concurrency}. In \bibinfo{booktitle}{\emph{{PLDI} '22: 43rd {ACM} {SIGPLAN} International Conference on Programming Language Design and Implementation, San Diego, CA, USA, June 13 - 17, 2022}}, \bibfield{editor}{\bibinfo{person}{Ranjit Jhala} {and} \bibinfo{person}{Isil Dillig}} (Eds.). \bibinfo{publisher}{{ACM}}, \bibinfo{pages}{458--473}.
\newblock
\showISBNx{978-1-4503-9265-5}
\href{https://doi.org/10.1145/3519939.3523443}{doi:\nolinkurl{10.1145/3519939.3523443}}


\bibitem[Milner(1978)]%
        {milner-78}
\bibfield{author}{\bibinfo{person}{Robin Milner}.} \bibinfo{year}{1978}\natexlab{}.
\newblock \showarticletitle{A Theory of Type Polymorphism in Programming}.
\newblock \bibinfo{journal}{\emph{J. Comput. System Sci.}} \bibinfo{volume}{17}, \bibinfo{number}{3} (\bibinfo{date}{Dec.} \bibinfo{year}{1978}), \bibinfo{pages}{348--375}.
\newblock
\urldef\tempurl%
\url{http://citeseerx.ist.psu.edu/viewdoc/summary?doi=10.1.1.67.5276}
\showURL{%
\tempurl}


\bibitem[Moine et~al\mbox{.}(2025)]%
        {mechanization}
\bibfield{author}{\bibinfo{person}{Alexandre Moine}, \bibinfo{person}{Stephanie Balzer}, \bibinfo{person}{Alex Xu}, {and} \bibinfo{person}{Sam Westrick}.} \bibinfo{year}{2025}\natexlab{}.
\newblock \bibinfo{booktitle}{\emph{TypeDis: A Type System for Disentanglement (Artifact)}}.
\newblock
\href{https://doi.org/10.5281/zenodo.17336385}{doi:\nolinkurl{10.5281/zenodo.17336385}}


\bibitem[Moine et~al\mbox{.}(2024)]%
        {moine-westrick-balzer-24}
\bibfield{author}{\bibinfo{person}{Alexandre Moine}, \bibinfo{person}{Sam Westrick}, {and} \bibinfo{person}{Stephanie Balzer}.} \bibinfo{year}{2024}\natexlab{}.
\newblock \showarticletitle{DisLog: A Separation Logic for Disentanglement}.
\newblock \bibinfo{journal}{\emph{Proc. ACM Program. Lang.}} \bibinfo{volume}{8}, \bibinfo{number}{POPL}, Article \bibinfo{articleno}{11} (\bibinfo{date}{Jan.} \bibinfo{year}{2024}), \bibinfo{numpages}{30}~pages.
\newblock
\href{https://doi.org/10.1145/3632853}{doi:\nolinkurl{10.1145/3632853}}


\bibitem[M{\"{u}}ller(2002)]%
        {MullerPhD2002}
\bibfield{author}{\bibinfo{person}{Peter M{\"{u}}ller}.} \bibinfo{year}{2002}\natexlab{}.
\newblock \bibinfo{booktitle}{\emph{Modular Specification and Verification of Object-Oriented Programs}}. \bibinfo{series}{Lecture Notes in Computer Science}, Vol.~\bibinfo{volume}{2262}.
\newblock \bibinfo{publisher}{Springer}.
\newblock
\showISBNx{3-540-43167-5}
\href{https://doi.org/10.1007/3-540-45651-1}{doi:\nolinkurl{10.1007/3-540-45651-1}}


\bibitem[Muller et~al\mbox{.}(2017)]%
        {MullerPLDI2017}
\bibfield{author}{\bibinfo{person}{Stefan~K. Muller}, \bibinfo{person}{Umut~A. Acar}, {and} \bibinfo{person}{Robert Harper}.} \bibinfo{year}{2017}\natexlab{}.
\newblock \showarticletitle{Responsive parallel computation: bridging competitive and cooperative threading}. In \bibinfo{booktitle}{\emph{38th {ACM} {SIGPLAN} Conference on Programming Language Design and Implementation ({PLDI})}}. \bibinfo{publisher}{{ACM}}, \bibinfo{pages}{677--692}.
\newblock
\href{https://doi.org/10.1145/3062341.3062370}{doi:\nolinkurl{10.1145/3062341.3062370}}


\bibitem[Noble et~al\mbox{.}(1998)]%
        {NobleECOOP1998}
\bibfield{author}{\bibinfo{person}{James Noble}, \bibinfo{person}{Jan Vitek}, {and} \bibinfo{person}{John Potter}.} \bibinfo{year}{1998}\natexlab{}.
\newblock \showarticletitle{Flexible Alias Protection}. In \bibinfo{booktitle}{\emph{12th European Conference on Object-Oriented Programming ({ECOOP})}} \emph{(\bibinfo{series}{Lecture Notes in Computer Science}, Vol.~\bibinfo{volume}{1445})}. \bibinfo{publisher}{Springer}, \bibinfo{pages}{158--185}.
\newblock
\href{https://doi.org/10.1007/BFB0054091}{doi:\nolinkurl{10.1007/BFB0054091}}


\bibitem[Odersky et~al\mbox{.}(1999)]%
        {odersky-sulzmann-wehr-99}
\bibfield{author}{\bibinfo{person}{Martin Odersky}, \bibinfo{person}{Martin Sulzmann}, {and} \bibinfo{person}{Martin Wehr}.} \bibinfo{year}{1999}\natexlab{}.
\newblock \showarticletitle{Type Inference with Constrained Types}.
\newblock \bibinfo{journal}{\emph{Theory and Practice of Object Systems}} \bibinfo{volume}{5}, \bibinfo{number}{1} (\bibinfo{year}{1999}), \bibinfo{pages}{35--55}.
\newblock
\urldef\tempurl%
\url{https://doi.org/10.1002/(SICI)1096-9942(199901/03)5:1%3C35::AID-TAPO4%3E3.0.CO;2-4}
\showURL{%
\tempurl}


\bibitem[Pierce(2002)]%
        {tapl}
\bibfield{author}{\bibinfo{person}{Benjamin~C. Pierce}.} \bibinfo{year}{2002}\natexlab{}.
\newblock \bibinfo{booktitle}{\emph{Types and Programming Languages}}.
\newblock \bibinfo{publisher}{MIT Press}.
\newblock


\bibitem[Pitts and Stark(1998)]%
        {PittsStarkHOOTS1998}
\bibfield{author}{\bibinfo{person}{Andrew~M. Pitts} {and} \bibinfo{person}{Ian Stark}.} \bibinfo{year}{1998}\natexlab{}.
\newblock \showarticletitle{Operational Reasoning for Functions with Local State}.
\newblock \bibinfo{journal}{\emph{Higher Order Operational Techniques in Semantics (HOOTS)}} (\bibinfo{year}{1998}), \bibinfo{pages}{227--273}.
\newblock


\bibitem[Plotkin(1973)]%
        {PlotkinTR1973}
\bibfield{author}{\bibinfo{person}{Gordon~D. Plotkin}.} \bibinfo{year}{1973}\natexlab{}.
\newblock \bibinfo{booktitle}{\emph{Lambda-definability and logical relations}}.
\newblock \bibinfo{type}{{T}echnical {R}eport}. \bibinfo{institution}{University of Edinburgh}.
\newblock


\bibitem[Pottier and Simonet(2003)]%
        {PottierSimonetARTICLE2003}
\bibfield{author}{\bibinfo{person}{Fran{\c{c}}ois Pottier} {and} \bibinfo{person}{Vincent Simonet}.} \bibinfo{year}{2003}\natexlab{}.
\newblock \showarticletitle{Information flow inference for {ML}}.
\newblock \bibinfo{journal}{\emph{{ACM} Transactions on Programming Languages and Systems ({TOPLAS})}} \bibinfo{volume}{25}, \bibinfo{number}{1} (\bibinfo{year}{2003}), \bibinfo{pages}{117--158}.
\newblock
\href{https://doi.org/10.1145/596980.596983}{doi:\nolinkurl{10.1145/596980.596983}}


\bibitem[Raghunathan et~al\mbox{.}(2016)]%
        {raghunathan-et-al-16}
\bibfield{author}{\bibinfo{person}{Ram Raghunathan}, \bibinfo{person}{Stefan~K. Muller}, \bibinfo{person}{Umut~A. Acar}, {and} \bibinfo{person}{Guy~E. Blelloch}.} \bibinfo{year}{2016}\natexlab{}.
\newblock \showarticletitle{Hierarchical memory management for parallel programs}. In \bibinfo{booktitle}{\emph{Proceedings of the 21st {ACM} {SIGPLAN} International Conference on Functional Programming, {ICFP} 2016, Nara, Japan, September 18-22, 2016}}, \bibfield{editor}{\bibinfo{person}{Jacques Garrigue}, \bibinfo{person}{Gabriele Keller}, {and} \bibinfo{person}{Eijiro Sumii}} (Eds.). \bibinfo{publisher}{{ACM}}, \bibinfo{pages}{392--406}.
\newblock
\href{https://doi.org/10.1145/2951913.2951935}{doi:\nolinkurl{10.1145/2951913.2951935}}


\bibitem[Sabelfeld and Myers(2003)]%
        {SabelfeldIEE2003}
\bibfield{author}{\bibinfo{person}{Andrei Sabelfeld} {and} \bibinfo{person}{Andrew~C. Myers}.} \bibinfo{year}{2003}\natexlab{}.
\newblock \showarticletitle{Language-Based Information-Flow Security}.
\newblock \bibinfo{journal}{\emph{{IEEE} J. Sel. Areas Commun.}} \bibinfo{volume}{21}, \bibinfo{number}{1} (\bibinfo{year}{2003}), \bibinfo{pages}{5--19}.
\newblock


\bibitem[Shun et~al\mbox{.}(2012)]%
        {DBLP:conf/spaa/ShunBFGKST12}
\bibfield{author}{\bibinfo{person}{Julian Shun}, \bibinfo{person}{Guy~E. Blelloch}, \bibinfo{person}{Jeremy~T. Fineman}, \bibinfo{person}{Phillip~B. Gibbons}, \bibinfo{person}{Aapo Kyrola}, \bibinfo{person}{Harsha~Vardhan Simhadri}, {and} \bibinfo{person}{Kanat Tangwongsan}.} \bibinfo{year}{2012}\natexlab{}.
\newblock \showarticletitle{Brief announcement: the problem based benchmark suite}. In \bibinfo{booktitle}{\emph{24th {ACM} Symposium on Parallelism in Algorithms and Architectures, {SPAA} '12, Pittsburgh, PA, USA, June 25-27, 2012}}, \bibfield{editor}{\bibinfo{person}{Guy~E. Blelloch} {and} \bibinfo{person}{Maurice Herlihy}} (Eds.). \bibinfo{publisher}{{ACM}}, \bibinfo{pages}{68--70}.
\newblock
\showISBNx{978-1-4503-1213-4}
\href{https://doi.org/10.1145/2312005.2312018}{doi:\nolinkurl{10.1145/2312005.2312018}}


\bibitem[Simonet(2003)]%
        {SimonetWorkshop2003}
\bibfield{author}{\bibinfo{person}{Vincent Simonet}.} \bibinfo{year}{2003}\natexlab{}.
\newblock \showarticletitle{Flow Caml in a Nutshell}. In \bibinfo{booktitle}{\emph{1st {APPSEM-II} Workshop}}, \bibfield{editor}{\bibinfo{person}{Graham Hutton}} (Ed.).
\newblock


\bibitem[Smith and Volpano(1998)]%
        {SmithVolpanoPOPL1998}
\bibfield{author}{\bibinfo{person}{Geoffrey Smith} {and} \bibinfo{person}{Dennis~M. Volpano}.} \bibinfo{year}{1998}\natexlab{}.
\newblock \showarticletitle{Secure Information Flow in a Multi-Threaded Imperative Language}. In \bibinfo{booktitle}{\emph{{POPL}}}. \bibinfo{publisher}{{ACM}}, \bibinfo{pages}{355--364}.
\newblock


\bibitem[Statman(1985)]%
        {StatmanARTICLE1985}
\bibfield{author}{\bibinfo{person}{Richard Statman}.} \bibinfo{year}{1985}\natexlab{}.
\newblock \showarticletitle{Logical Relations and the Typed $\lambda$-calculus}.
\newblock \bibinfo{journal}{\emph{Information and Control}} \bibinfo{volume}{65}, \bibinfo{number}{2/3} (\bibinfo{year}{1985}), \bibinfo{pages}{85--97}.
\newblock
\href{https://doi.org/10.1016/S0019-9958(85)80001-2}{doi:\nolinkurl{10.1016/S0019-9958(85)80001-2}}


\bibitem[Tait(1967)]%
        {TaitARTICLE1967}
\bibfield{author}{\bibinfo{person}{William~W. Tait}.} \bibinfo{year}{1967}\natexlab{}.
\newblock \showarticletitle{Intensional Interpretations of Functionals of Finite Type {I}}.
\newblock \bibinfo{journal}{\emph{The Journal of Symbolic Logic}} \bibinfo{volume}{32}, \bibinfo{number}{2} (\bibinfo{year}{1967}), \bibinfo{pages}{198--212}.
\newblock
\showISSN{00224812}
\urldef\tempurl%
\url{http://www.jstor.org/stable/2271658}
\showURL{%
\tempurl}


\bibitem[Timany et~al\mbox{.}(2024)]%
        {timany-krebbers-dreyer-birkedal-24}
\bibfield{author}{\bibinfo{person}{Amin Timany}, \bibinfo{person}{Robbert Krebbers}, \bibinfo{person}{Derek Dreyer}, {and} \bibinfo{person}{Lars Birkedal}.} \bibinfo{year}{2024}\natexlab{}.
\newblock \showarticletitle{A Logical Approach to Type Soundness}.
\newblock \bibinfo{journal}{\emph{J. ACM}} \bibinfo{volume}{71}, \bibinfo{number}{6}, Article \bibinfo{articleno}{40} (\bibinfo{date}{Nov.} \bibinfo{year}{2024}), \bibinfo{numpages}{75}~pages.
\newblock
\showISSN{0004-5411}
\href{https://doi.org/10.1145/3676954}{doi:\nolinkurl{10.1145/3676954}}


\bibitem[Tofte and Birkedal(1998)]%
        {DBLP:journals/toplas/TofteB98}
\bibfield{author}{\bibinfo{person}{Mads Tofte} {and} \bibinfo{person}{Lars Birkedal}.} \bibinfo{year}{1998}\natexlab{}.
\newblock \showarticletitle{A Region Inference Algorithm}.
\newblock \bibinfo{journal}{\emph{{ACM} Trans. Program. Lang. Syst.}} \bibinfo{volume}{20}, \bibinfo{number}{4} (\bibinfo{year}{1998}), \bibinfo{pages}{724--767}.
\newblock
\href{https://doi.org/10.1145/291891.291894}{doi:\nolinkurl{10.1145/291891.291894}}


\bibitem[Tofte et~al\mbox{.}(2004)]%
        {tofte-retro-04}
\bibfield{author}{\bibinfo{person}{Mads Tofte}, \bibinfo{person}{Lars Birkedal}, \bibinfo{person}{Martin Elsman}, {and} \bibinfo{person}{Niels Hallenberg}.} \bibinfo{year}{2004}\natexlab{}.
\newblock \showarticletitle{A Retrospective on Region-Based Memory Management}.
\newblock \bibinfo{journal}{\emph{Higher-Order and Symbolic Computation}} \bibinfo{volume}{17}, \bibinfo{number}{3} (\bibinfo{date}{Sept.} \bibinfo{year}{2004}), \bibinfo{pages}{245--265}.
\newblock
\urldef\tempurl%
\url{https://doi.org/10.1023/B:LISP.0000029446.78563.a4}
\showURL{%
\tempurl}


\bibitem[Tofte and Talpin(1997)]%
        {tofte-talpin-97}
\bibfield{author}{\bibinfo{person}{Mads Tofte} {and} \bibinfo{person}{Jean-Pierre Talpin}.} \bibinfo{year}{1997}\natexlab{}.
\newblock \showarticletitle{Region-based memory management}.
\newblock \bibinfo{journal}{\emph{Information and Computation}} \bibinfo{volume}{132}, \bibinfo{number}{2} (\bibinfo{year}{1997}), \bibinfo{pages}{109--176}.
\newblock
\urldef\tempurl%
\url{http://www.irisa.fr/prive/talpin/papers/ic97.pdf}
\showURL{%
\tempurl}


\bibitem[VerifyThis(2022)]%
        {verify-this-22}
\bibfield{author}{\bibinfo{person}{VerifyThis}.} \bibinfo{year}{2022}\natexlab{}.
\newblock \bibinfo{title}{Challenge 3 - The World’s Simplest Lock-Free Hash Set}.
\newblock
\urldef\tempurl%
\url{https://ethz.ch/content/dam/ethz/special-interest/infk/chair-program-method/pm/documents/Verify%20This/Challenges2022/verifyThis2022-challenge3.pdf}
\showURL{%
\tempurl}


\bibitem[Vindum and Birkedal(2021)]%
        {vindum-birkedal-21}
\bibfield{author}{\bibinfo{person}{Simon~Friis Vindum} {and} \bibinfo{person}{Lars Birkedal}.} \bibinfo{year}{2021}\natexlab{}.
\newblock \showarticletitle{Contextual refinement of the {Michael-Scott} queue}. In \bibinfo{booktitle}{\emph{Certified Programs and Proofs (CPP)}}. \bibinfo{pages}{76--90}.
\newblock
\urldef\tempurl%
\url{https://cs.au.dk/~birke/papers/2021-ms-queue-final.pdf}
\showURL{%
\tempurl}


\bibitem[Volpano et~al\mbox{.}(1996)]%
        {VolpanoARTICLE1996}
\bibfield{author}{\bibinfo{person}{Dennis~M. Volpano}, \bibinfo{person}{Cynthia~E. Irvine}, {and} \bibinfo{person}{Geoffrey Smith}.} \bibinfo{year}{1996}\natexlab{}.
\newblock \showarticletitle{A Sound Type System for Secure Flow Analysis}.
\newblock \bibinfo{journal}{\emph{J. Comput. Secur.}} \bibinfo{volume}{4}, \bibinfo{number}{2/3} (\bibinfo{year}{1996}), \bibinfo{pages}{167--188}.
\newblock


\bibitem[Wadler(1990)]%
        {WadlerIFIP1990}
\bibfield{author}{\bibinfo{person}{Philip Wadler}.} \bibinfo{year}{1990}\natexlab{}.
\newblock \showarticletitle{Linear Types Can Change the World!}. In \bibinfo{booktitle}{\emph{{IFIP} Working Group 2.2, 2.3 on Programming Concepts and Methods}}. \bibinfo{publisher}{North-Holland}, \bibinfo{pages}{561}.
\newblock


\bibitem[Wadler(2012)]%
        {WadlerICFP2012}
\bibfield{author}{\bibinfo{person}{Philip Wadler}.} \bibinfo{year}{2012}\natexlab{}.
\newblock \showarticletitle{Propositions as Sessions}. In \bibinfo{booktitle}{\emph{{ACM} {SIGPLAN} International Conference on Functional Programming ({ICFP})}}. \bibinfo{publisher}{{ACM}}, \bibinfo{pages}{273--286}.
\newblock
\href{https://doi.org/10.1145/2364527.2364568}{doi:\nolinkurl{10.1145/2364527.2364568}}


\bibitem[Westrick(2022)]%
        {westrick-thesis-2022}
\bibfield{author}{\bibinfo{person}{Sam Westrick}.} \bibinfo{year}{2022}\natexlab{}.
\newblock \emph{\bibinfo{title}{Efficient and Scalable Parallel Functional Programming through Disentanglement}}.
\newblock \bibinfo{thesistype}{Ph.\,D. Dissertation}. \bibinfo{school}{Department of Computer Science, Carnegie Mellon University}.
\newblock


\bibitem[Westrick et~al\mbox{.}(2022)]%
        {westrick-arora-acar-22}
\bibfield{author}{\bibinfo{person}{Sam Westrick}, \bibinfo{person}{Jatin Arora}, {and} \bibinfo{person}{Umut~A. Acar}.} \bibinfo{year}{2022}\natexlab{}.
\newblock \showarticletitle{Entanglement Detection with Near-Zero Cost}.
\newblock \bibinfo{journal}{\emph{Proc. ACM Program. Lang.}} \bibinfo{volume}{6}, \bibinfo{number}{ICFP}, Article \bibinfo{articleno}{115} (\bibinfo{date}{aug} \bibinfo{year}{2022}), \bibinfo{numpages}{32}~pages.
\newblock
\href{https://doi.org/10.1145/3547646}{doi:\nolinkurl{10.1145/3547646}}


\bibitem[Westrick et~al\mbox{.}(2020)]%
        {westrick-et-al-20}
\bibfield{author}{\bibinfo{person}{Sam Westrick}, \bibinfo{person}{Rohan Yadav}, \bibinfo{person}{Matthew Fluet}, {and} \bibinfo{person}{Umut~A. Acar}.} \bibinfo{year}{2020}\natexlab{}.
\newblock \showarticletitle{Disentanglement in Nested-Parallel Programs}.
\newblock \bibinfo{journal}{\emph{Proc. ACM Program. Lang.}} \bibinfo{volume}{4}, \bibinfo{number}{POPL}, Article \bibinfo{articleno}{47} (\bibinfo{date}{jan} \bibinfo{year}{2020}), \bibinfo{numpages}{32}~pages.
\newblock
\href{https://doi.org/10.1145/3371115}{doi:\nolinkurl{10.1145/3371115}}


\bibitem[Wright(1995)]%
        {wright-restriction-95}
\bibfield{author}{\bibinfo{person}{Andrew~K. Wright}.} \bibinfo{year}{1995}\natexlab{}.
\newblock \showarticletitle{Simple Imperative Polymorphism}.
\newblock \bibinfo{journal}{\emph{Lisp and Symbolic Computation}} \bibinfo{volume}{8}, \bibinfo{number}{4} (\bibinfo{date}{Dec.} \bibinfo{year}{1995}), \bibinfo{pages}{343--356}.
\newblock
\urldef\tempurl%
\url{http://www.cs.rice.edu/CS/PLT/Publications/Scheme/lasc95-w.ps.gz}
\showURL{%
\tempurl}


\end{thebibliography}
\end{document}